\documentclass[12pt]{article}

\usepackage{amsmath, nccmath, amssymb, amsthm}
\usepackage[utf8x]{inputenc}
\usepackage{setspace}
\usepackage{bm}
\usepackage{natbib, authblk}
\usepackage{graphicx}
\usepackage{mathtools}
\usepackage{booktabs}
\usepackage{multirow}
\usepackage[dvipsnames]{xcolor}
\usepackage[plain,noend]{algorithm2e}
\usepackage{physics}
\usepackage{subfig}
\usepackage[margin=1in]{geometry}
\usepackage{enumerate}

\newtheorem{assumption}{Assumption}
\newtheorem{theorem}{Theorem}
\newtheorem{proposition}{Proposition}

\newtheorem{lemma}{Lemma}
\newtheorem{example}{Example}

\newtheorem{definition}{Definition}

\newtheorem{innercustomthm}{Theorem}
\newenvironment{customthm}[1]
{\renewcommand\theinnercustomthm{#1}\innercustomthm}
{\endinnercustomthm}
\newtheorem{innercustomlem}{Lemma}
\newenvironment{customlem}[1]
{\renewcommand\theinnercustomlem{#1}\innercustomlem}
{\endinnercustomlem}
\usepackage[colorlinks=true,linkcolor=blue, citecolor=blue]{hyperref}%

\DeclareMathOperator{\Unif}{Unif}
\newcommand{\eps}{\varepsilon}
\newcommand{\heps}{\widehat{\eps}}

\renewcommand{\P}{\mathbb{P}}
\newcommand{\E}{\mathbb{E}}

\newcommand{\indep}{\perp \!\!\! \perp}

\renewcommand{\complement}{\mathrm{c}}

\newcommand{\cG}{\mathcal{G}_n}

\newcommand{\cN}{\mathcal{N}}

\newcommand{\Real}{\mathbb{R}}

\newcommand{\R}{\mathbb{R}}

\newcommand{\name}{invariance}

\newcommand\myeq[1]{\stackrel{\mathclap{\normalfont\mbox{\tiny #1}}}{=}}

\begin{document}

\title{\vspace{-2.1cm}
Invariance-based Inference in High-Dimensional Regression with Finite-Sample Guarantees
}

\author{Wenxuan Guo and Panos Toulis}
\affil{University of Chicago, Booth School of Business}

\maketitle

\onehalfspacing

\begin{abstract}
In this paper, we develop \name-based procedures for testing and inference in high-dimensional regression models. 
These procedures, also known as randomization tests, provide several important advantages.
First, for the global null hypothesis of significance, our test is valid in finite samples. It is also simple to implement and comes with finite-sample guarantees on statistical power.
Remarkably, despite its simplicity, this testing idea has escaped the attention of earlier analytical work, which mainly concentrated on complex high-dimensional asymptotic methods. 
Under an additional assumption of Gaussian design, we show that this test also achieves the minimax optimal rate against 
certain nonsparse alternatives, a type of result that is rare in the literature.
Second, for partial null hypotheses, we propose residual-based tests and derive theoretical conditions for their validity.
These tests can be made powerful by constructing the test statistic in a way that, first, selects the important 
covariates (e.g., through Lasso) and then orthogonalizes the nuisance parameters.
We illustrate our results through extensive simulations and applied examples. 
One consistent finding is that the strong finite-sample guarantees associated with our procedures result in added robustness when it comes to handling multicollinearity and heavy-tailed covariates.
\end{abstract}


\section{Introduction}\label{sec:intro}
Consider the perennial linear regression model,
\begin{equation}\label{eq:lm}
    y = X \beta + \eps~,
\end{equation}
where $y=(y_{1}, \dots, y_{n})^{\top}$ is the outcome vector, $X \in \mathbb{R}^{n \times p}$ are the covariates, and $\eps = (\eps_1, \dots, \eps_n)^\top$ are the unobserved errors. In this paper, we consider the problem of testing null hypotheses on $\beta$ when either $p > n$, 
or $p < n$ but potentially $p \to \infty$. 

Two particular problems are the global null hypothesis and partial null hypotheses denoted, respectively, as
\begin{equation}\label{eq:H0}
H_0:\beta=0~\text{and}~H_0^S: \beta_S = 0,~S \subset \{1, \ldots, p\},
\end{equation}
where $S$ denotes a proper subset of coefficient parameters, and $\beta_S$ is the 
subvector of $\beta$ that corresponds to the parameters in $S$.

Testing such null hypotheses under high-dimensional regimes ($p>n$) has become important in modern scientific and engineering applications. For instance, \citet{cai2022} discuss a genomics application where $y$ are protein expression data, and $X$ are data on spatial locations of genetic measurements. In this context, null hypotheses such as $H_0$ and $H_0^S$ capture which (if any) genes are significant across the cells and tissue of an organism. In engineering,~\cite{Campi2005, Campi2009} have discussed testing the global null hypothesis from a systems identification viewpoint with applications to process control. In economics, the global null hypothesis could indicate whether high-dimensional controls are relevant for a response of interest, as in the popular abortion-crime study of~\cite{Donohue2001}, which we re-analyze in this paper.

To address these high-dimensional testing problems, prior work has mainly considered high-dimensional versions of the classical $F$-test~\citep{Zhong2011,Cui2018, Li2020}, 
or tests based on the `debiased Lasso'~\citep{vandeGeer2014, Zhang2014,Javanmard2014,Ma2020}.
While these methods are effective in particular scenarios, their inherent asymptotic nature means they lack the ability to offer finite-sample guarantees for either Type I or Type II errors in the respective tests. This deficiency serves as a primary motivation behind our research.
In addition, most of these methods rely on assumptions, such as error homoskedasticity and bounded moment conditions, that can be unrealistic in practice.
For instance, in Section \ref{sec:simu_global_linear}, we show empirically that 
existing methods can be overly sensitive to heavy-tailed errors, whereas the \name-based methods we consider remain robust.

\subsection{Overview of Approach and Results }
In this paper, we develop \name-based tests for $H_0$ and $H_0^S$ under an overarching invariance assumption on the model errors. 
To illustrate the main idea, let us suppose, without loss of generality, that the errors are symmetric such that
for every finite sample $n>0$
\begin{equation}\label{eq:inv0}
(\eps_1, \ldots, \eps_n) \myeq{d} (\pm \eps_1, \ldots, \pm \eps_n).
\end{equation}
Then, the global null hypothesis ($H_0$) implies that $y = \varepsilon$ and so the outcomes are symmetric as well. We can thus test $H_0$ in two simple steps:
\begin{enumerate}
    \item Choose a suitable test statistic $t(y,X)$; e.g., ridge estimator.
    \vspace{-0.8em}
    \item Test $H_0$ through the so-called randomization $p$-value:
    \begin{equation}\label{eq:pval0}
    \mathrm{pval} = \frac{1}{|\cG| } \sum_{g\in\cG} \mathbb{I}\big\{t(g y, X)>t(y, X)\big\},
\end{equation}
\end{enumerate}
where $\cG$ is the  group of $n\times n$ diagonal matrices with $\pm 1$ entries in the diagonal.
The $p$-value in~\eqref{eq:pval0}, in effect, compares the observed value of the test statistic with a reference distribution where the signs of the outcomes are randomly flipped. 
Unlike asymptotic methods that rely on assumptions about the joint distribution of $(y,X)$,
the test in~\eqref{eq:pval0} is {\em invariance-based} as it only relies on error symmetry.
Crucially, this test is  valid for any finite $n>0$  according to standard theory~\citep[Chapter 15]{Lehmann2005}. Notably, this does not require any further assumptions on the data distribution, and can encompass other invariances beyond symmetry, such as exchangeability including clustered types~\citep{Toulis2019}. 

\paragraph{Results for the global null hypothesis.} 
After establishing the finite-sample validity of~\eqref{eq:pval0} for the global null hypothesis ($H_0$) in Theorem~\ref{thm:exact}, we perform a detailed finite-sample analysis for the Type II error of our test, and prove the following results:
%
%
\begin{enumerate}[(i)]
    \vspace{-0.1em}\item Under mild conditions, we derive {\em nonasymptotic bounds} on the test's power, which reveals trade-offs between signal strength and the design's multicollinearity~(Theorem~\ref{thm:nonasymp_power}). 
    
    \vspace{-0.1em}\item The \name-based test for $H_0$ is consistent as long as the coefficients' signal strength grows faster than a certain detection rate; this rate depends on problem 
    characteristics such as the min/max singular values of $X$~(Theorem~\ref{thm:consistent}).
    \vspace{-0.1em}\item Under a Gaussian design, our test achieves the minimax optimal rate of the form $O(\sqrt{p^{1/2}/n})$, just like classical $\chi^2$-based tests~\citep{Ingster2010, carpentier2018}. This result is proven in~Theorem~\ref{thm:normal_minimax}.
\end{enumerate}

For technical reasons, which we discuss throughout the paper, the power analysis in (ii) and (iii) is performed under a regime where the test statistic is based on the ridge estimator, and $p\to\infty$ but $p<n$. This limitation appears  to only be theoretical, however,
as our procedure remains powerful in all empirical settings with $p>n$ 
that we considered (Section~\ref{sec:simu}). This theoretical gap may be due to a
fundamental limitation in the theory of randomization tests, as 
related work in power analysis of randomization tests also imposes the condition $p < n$~\citep{ Toulis2019, Lei2020, Wen2022}. 

\paragraph{Results for the partial null hypothesis.} 
To test $H_0^S$ we extend the idea of {\em residual randomization}~\citep{Toulis2019}, whereby an approximate randomization test is constructed using model residuals, $\heps$. 
Whenever $\heps$ and $\eps$ are ``close enough" in a particular sense that we detail later, the residual-based test will perform similarly to the oracle test, and achieve the nominal level. We specify the test statistic similarly as in the global null setup, and provide conditions for asymptotic validity in Theorem~\ref{thm:partial_null_valid}. Roughly speaking, our test is asymptotically valid as long as the design leverage is not excessive, 
and design multicollinearity is mild. 


\subsection{Related Work}\label{sec:related}
A symmetry invariance, such as the one defined in equation~\eqref{eq:inv0}, is referred to as the `randomization hypothesis' in \citep{Lehmann2005}. Consequently, the procedures we develop in this paper, such as the test presented in equation~\eqref{eq:pval0}, are instances of a randomization test.
However, throughout this paper, we will maintain the term ``invariance-based" to emphasize that the physical act of randomization (e.g., through an experiment) is not a requisite element, thereby avoiding any potential confusion. 

Our work is therefore closely related to standard randomization theory, and particularly the recent works of \cite{Lei2020}, \cite{dobriban2022} and \cite{Toulis2019}.

Towards the global null, \cite{dobriban2022} studied the consistency of randomization tests in the contexts of signal detection and linear regression. Notably, \citet{dobriban2022} achieved a breakthrough in certain signal detection problems, demonstrating that randomization tests can attain minimax optimality.
For linear regression, the paper studied scenarios within a detection radius based on the $L_\infty$ norm, but it remains uncertain whether the corresponding tests can achieve minimax optimality.
%
%
In contrast, our research centers on the detection radius defined by the $L_2$ norm. In doing so, we establish the first instance of minimax optimality for randomization tests in the context of linear regression (see Theorem~\ref{thm:normal_minimax}). Our approach deviates from \cite{dobriban2022}, both in terms of test statistics and analytical methods, introducing novel theoretical insights into the domain of randomization tests.
%
%

For partial null hypotheses, \cite{Lei2020} have developed a cyclic permutation test under 
an assumption of exchangeable errors.
A notable feature of this approach is that it guarantees an exact test. 
 However, this also requires an intricate construction of a permutation subgroup, which can pose computational challenges and potentially result in power loss.
In contrast, our approach is only valid in an asymptotic sense, yet it offers two important advantages.
First, it works under symmetry invariances, beyond exchangeability, making it suitable under error heteroskedasticity. Second, it is parameterized in a way that can improve power, a feature 
we empirically validate in the experiments of Section~\ref{sec:simu_partial_linear}.



Related to our invariance-based tests for the partial null hypothesis, \cite{Toulis2019} has studied residual randomization tests, and derived nonasymptotic results on the Type I-II errors of such tests under several types of error invariances. These results, while insightful, do not encompass the high-dimensional regime. 
Nevertheless, we utilize these findings as a foundational element in the development of our theory. 
%

In a related vein, \cite{Wang2021} and \cite{Wen2022} have both extended the concept of residual randomization tests to the high-dimensional regime. 
\cite{Wen2022}, following a similar approach to~\cite{Lei2020}, introduced a clever construction that ensures an exact test for the partial null hypothesis. Notably,
the authors also showed that their test is nearly minimax optimal.
However, as with \cite{Lei2020}, the construction in~\cite{Wen2022} can be challenging and potentially lead to loss of statistical power.
Additionally, the method of~\cite{Wen2022} is constrained to testing a single coefficient (i.e., $|S|=1$), 
and its validity requires that $p<n$.
Taking a different approach, \cite{Wang2021}  leveraged the debiased Lasso technique and an intricate regularization scheme to construct a residual randomization test for high-dimensional regression. 
Like our test for the partial null hypothesis, the test of \cite{Wang2021} is valid only in the asymptotic sense. However, our approach here distinguishes itself by its simplicity and ease of implementation, relying on a simple projection-based statistic. Furthermore, it comes with verifiable conditions that depend on the problem dimension and the rank of the projection matrix (Theorem~\ref{thm:partial_null_valid}).

\paragraph{Historical notes on randomization tests.}
The key idea behind randomization tests can be traced back to~\cite{fisher1935} and \cite{ pitman1937}; historical details can be found in \cite{david2008}. 
Later, the idea was generalized under a rigorous mathematical framework in the seminal work of \cite{Lehmann2005}. 
Due to their robustness, randomization tests have recently gained wide popularity in the analysis of experiments and causal inference~\citep{Manly1997, edgington2007randomization, gerber2012field, imbens2015causal, ding2017paradox, canay2017randomization} including settings with 
interference~\citep{athey2018exact, basse2019randomization}. 
The idea of employing randomization tests in non-randomized settings is mainly due to \cite{Freedman1983},  who proposed residual-based randomization tests for partial nulls in standard linear regression.

\paragraph{Other methods 
based on the debiased Lasso and the classical $F$-test.}
The $F$-test and $t$-test are arguably the most commonly used tests of significance in linear regression. However, these tests cannot be applied to high-dimensional settings where $p>n$. Moreover, these tests can be sensitive to non-normal errors~\citep{Boos1995,Akritas2000,Calhoun2011}, which is a problem especially in high-dimensional regimes.

Recent work has thus focused on improving these classical tests, or on developing new approaches altogether.
One such innovative approach explored the renowned concept of `debiased Lasso', and developed methodologies for conducting statistical inference on low-dimensional parameters using high-dimensional data~\citep{Zhang2014, vandeGeer2014, Javanmard2014}. This approach effectively involves testing $H_0^S$, with $|S|$ (the cardinality of $S$) considerably smaller than the parameter dimension.
In a different direction,~\cite{Zhong2011,Cui2018} proposed the use of certain U-statistics to effectively estimate $\beta$ (the true vector of coefficients) from data, and performed a power analysis under certain local alternatives.
Related to our paper, this analysis shows that the power of the $F$-test is adversely affected by an increasing $p$ even when $p < n$. Similar ideas have been employed to test partial null hypotheses ($H_0^S$) by \citet{wang2015}, who also studied asymptotic normality in settings where $p/n \to\rho$ with $0 < \rho <1$.
\cite{Li2020} employed a sketching approach to extend the $F$-test to high-dimensional settings. In this approach, the idea is to project data to a low-dimensional space so that the $F$-test can be applied on the projected data. 
A notable theoretical development has also been the study of minimax optimal rates for significance testing in sparse linear regression models~\citep{Ingster2010,carpentier2018}.
In our minimax analysis, we adopt the same notion of minimax error and rely on a lower bound established in \cite{Ingster2010}.

This line of analytical work is particularly distinct from the \name-based approach that we adopt in this paper. 
Although both approaches necessitate regularity conditions on the data distribution, 
current analytical methods are inherently rooted in asymptotic theory, and do not offer finite-sample guarantees for either Type I or Type II errors in the respective tests.
As demonstrated in the experiments of Section~\ref{sec:simu}, such finite-sample guarantees are associated with robustness in settings characterized by multicollinearity and heavy-tailed covariates---issues 
frequently encountered in high-dimensional regression.

\section{Testing the Global Null Hypothesis}\label{sec:global}
\subsection{An Exact Test}\label{sec:rand:lm}
In this section, we develop an invariance-based test for $H_0:\beta = 0$ in model \eqref{eq:lm}, and prove its finite-sample exactness. 
We assume a general form of invariance as follows:
\begin{equation}\label{eq:main_invariance}
    \eps \myeq{d} g\eps \mid X~\text{for all $g\in\cG$ and $X$},
\end{equation}
where $\cG$ is an algebraic group of  $\Real^n\to\Real^n$ transformations under matrix multiplication as the group action, and is defined independently of $X$. 

Given the error invariance in Equation~\eqref{eq:main_invariance}, we can establish that $t(\eps, X) \myeq{d} t(g\eps, X)$ for all $g\in \cG$ conditional on $X$. 
Under the global null hypothesis, $H_0$, we further have that $y = \eps$ and so $t(y, X) \myeq{d} t(gy, X)$.
This justifies the use of the $p$-value in~\eqref{eq:pval0} to test $H_0$.  In practice, however, the cardinality of $\cG$ can be prohibitively large.\footnote{For instance, $|\cG| = n!$ if $\cG$ is the set of all $n\times n$ permutation matrices.} As a remedy, we may approximate the $p$-value in \eqref{eq:pval0} by sampling from $\cG$, as in the following procedure.
\begin{center}
    \textsc{Procedure 1: Feasible \name-based test for the global null $H_0$ in \eqref{eq:lm}.}
\end{center}
\begin{enumerate}
    \item Compute the observed value 
    of the test statistic, $T_n = t(y, X)$. 
    \item Compute the randomized statistic $t(G_r y, X)$, where $G_r\stackrel{iid}{\sim}\Unif(\cG)$, $r = 1, \dots, R$. 
    \item Obtain the one-sided $p$-value 
    (or the derivative two-sided value):
    \begin{equation}\label{eq:pvalrand}
        \mathrm{pval} = \frac{1}{R+1}\left(1 + \sum_{r = 1}^{R}\mathbb{I}\{t(G_r y, X)> T_n \}\right).
    \end{equation}
\end{enumerate}

In Procedure~1 above, $\Unif(\cG)$ denotes the uniform distribution over $\cG$, and $R$ is a fixed number of samples used to approximate the $p$-value in \eqref{eq:pval0}, sampled independently with replacement from $\cG$.
For a level-$\alpha$ test, we reject $H_0$ if the $p$-value~\eqref{eq:pvalrand} is less or equal to $\alpha$ 
through the decision function $\psi_\alpha = \mathbb{I}\{\mathrm{pval}\le\alpha\}$.

The following theorem shows that, for any choice of the test statistic, $\psi_\alpha$ is a valid test at any level $\alpha>0$. The proof adapts standard techniques to our setting; e.g., \cite{Toulis2019} or \cite{dobriban2022}. We use $\E_0(\cdot)$ to denote the expectation under the null, $H_0$. 
\begin{theorem}\label{thm:exact}
Suppose that the invariance in Equation~\eqref{eq:main_invariance} holds. Then, $\E_0(\psi_\alpha) \le \alpha$ at any level $\alpha\in(0,1]$ and any $n>0$ whenever the global null hypothesis, $H_0:\beta=0$, holds.
\end{theorem}

Intuitively, our test achieves finite-sample validity by building the true reference distribution of the test statistic based on the invariance and the null hypothesis. This requires no further 
assumption on the test statistic or the data distribution  as long as the invariance holds.
Remarkably, this straightforward method for creating a finite-sample valid test has 
apparently escaped the attention of earlier analytical work, which predominantly focused on high-dimensional asymptotic methods~\citep{Ingster2010, Cui2018, Li2020}. The distinction between these asymptotic approaches and our \name-based procedure can be substantial in practice.  In Section~\ref{sec:simu_global_linear}, for example, 
we demonstrate through simulated studies that asymptotic methods can exhibit significant size distortion, especially in small-to-moderate sample sizes, whereas our procedure remains robust.

Of course, while the choice of the test statistic does not affect the validity of Procedure 1, it does affect its power properties. In the following section, we study theoretically the power of our \name-based test for particular choices of the test statistic and error invariance.

\subsection{Power Analysis of Procedure 1}\label{sec:power_results}
Here, we study theoretically the Type II error of our \name-based test.
As we will consider both settings with nonrandom $X$ (fixed design) and 
random $X$ (random design), we will use $\E_X(1 - \psi_\alpha)$ and $\E(1 - \psi_\alpha)$ to 
denote the respective Type II errors.

\subsubsection{Setup}
To make progress, we must narrow our attention to specific choices of the test statistic and error invariance. 
Although this may be limiting, it is an essential step because, for any choice of invariance, there 
exist test statistics that would render the resulting test powerless.\footnote{As an example, 
suppose that $\cG$ denotes permutations, and consider $t(y, X) = \bar y$ as the test statistic. Then, $t(gy, X) = t(y, X)$ almost surely for all $g\in\cG$, which implies a `degenerate' randomization test (i.e., no power against any alternative).
}

As a first choice, we assume sign symmetric errors. This is a general form of invariance that allows certain types of heteroskedasticity, which is common in high-dimensional data. 
\begin{assumption}[Symmetry]
    Under this invariance,
    \label{asm:err1}
    $(\eps_1, \ldots, \eps_n) \myeq{d} (\pm \eps_1, \ldots, \pm \eps_n) \mid X$. Let $\cG^{\mathrm{s}}$ denote the group of all $n\times n$ matrices with $\pm 1$ on the diagonal.
\end{assumption}
Second, as the test statistic, we work with the standard ridge regression estimator:
\begin{equation}\label{eq:ridge}
t(y, X) = \|X\widehat{\beta}\|^2\;,\quad \widehat{\beta} = (X^\top X + \lambda I_p)^{-1} X^\top y~,
\end{equation}
where $\lambda\ge 0$ is the ridge tuning parameter. 
This choice makes sense because the ridge estimator is straightforward to compute and is amenable to theoretical analysis.
Finally, our analysis will consider a regime where $p$ diverges with $n$ but $p<n$. 
This constraint is necessary for technical reasons, which we will delve into later. Empirically, we will demonstrate that our theory remains applicable even in settings  with $p>n$ (Section~\ref{sec:simu}).

\subsubsection{Nonasymptotic Results}
In this section, we give nonasymptotic Type II error bounds under a fixed design. 
This is a key technical contribution of our work, 
and combines concentration inequalities from random matrix theory \citep{Tropp2015} with randomization theory.
As a first step, we introduce some necessary regularity conditions as follows.
\begin{assumption}\label{asmp:regularity}
Define $\sigma_{\min}$ and $\sigma_{\max}$ as the minimum and the maximum singular value of $X$, respectively. We assume that $\sigma_{\min}>0$ at any sample size $n>0$.
\end{assumption}
The main implication of Assumption~\ref{asmp:regularity} is that, in the regime where $p<n$, 
the design matrix $X\in{\R^{n\times p}}$ has full column rank. 
We note that this assumption is not strictly necessary; e.g.,  the results of Theorem~\ref{thm:nonasymp_power} that follows could be expressed in terms of the minimum nonzero singular value of $X$ instead of $
\sigma_{\min}$. However, by maintaining Assumption~\ref{asmp:regularity}, the theorem's interpretation
becomes more transparent, and the  power-related trade-offs become more intuitive, as we discuss below.

Next, define $x_* \coloneqq \max_{i\in[n], j\in[p]} |X_{ij}|$, and $\sigma_*^2 \coloneqq \max_{i\in[n]} \E (\eps_i^2)$ as the maximum error second moment. To better interpret the results, we also introduce two quantities, namely $\kappa = \sigma_{\max}/\sigma_{\min}$ and $s = \sigma_{\min}^2/n$. Here, $\kappa$ is the condition number of $X$, and is a measure of multicollinearity in the design matrix; $s$ is a normalized version of $\sigma_{\min}^2$, and it is bounded under regular conditions; e.g., when $(X_{ij})_{i\in[n],j\in[p]}$ are i.i.d. Gaussian (see Lemma~\ref{lem:gauss-cov}). 

We are now ready to state our main nonasymptotic upper bound on the Type II error of our test; see Appendix~\ref{sec:nonasymp_powerproof} for the proof. 
\begin{theorem}\label{thm:nonasymp_power}
Suppose Assumptions \ref{asm:err1} and \ref{asmp:regularity} hold, and that $\beta\neq 0$ in \eqref{eq:lm} and $p<n$. If we choose the ridge penalty parameter such that $\lambda \le \sigma_{\min}^2$, we have
\begin{align}
    \E_X(1 - \psi_\alpha) &= O(A_n(X)) + O\Big(\frac{B_n(X)}{\|\beta\|^2}\Big),~\text{where}~A_n(X) = \frac{p^2 \kappa^4 x_*^2 + 1}{sn}, B_n(X)= \frac{p \kappa^4 \sigma_*^2 + p}{sn}~.\nonumber
\end{align}
\end{theorem}
{\em Remarks.}
The bound in Theorem~\ref{thm:nonasymp_power} reveals how the Type II error of the global test depends on certain key design characteristics. For intuition, we may use $s = O(1)$, which is true in regular settings as mentioned above. 
This allows to rewrite $A_n(X)$ and $B_n(X)$ as
\begin{align*}
    A_n(X) \asymp (1+ \kappa^4 p^2 x_*^2)/n = O(\kappa^4 p^2 x_*^2/n)~,~~
    B_n(X) \asymp (1+ \kappa^4 \sigma_*^2)p/n = O(\kappa^4 \sigma_*^2 p/n),
\end{align*}
where we assumed that $\kappa^4 p^2 x_*^2$ and $\kappa^4 \sigma_*^2$ are greater than one. Then, 
\begin{equation*}
    \E_X(1 - \psi_\alpha) = O\left(\frac{p}{n} \cdot\kappa^4\cdot \left( px_*^2 + \frac{\sigma_*^2}{\|\beta\|^2} \right)\right)~.
\end{equation*}
Intuitively, we see that the Type II error is determined by the problem dimension $p/n$, the multicollinearity measure $\kappa$, the design leverage captured by  $px_*^2$, and the signal-to-noise ratio $\|\beta\|/\sigma_*$.\footnote{Loosely speaking, $x_*^2$ captures the maximum leverage, $M_* = \max_{i\in[n]} X_i^\top (X^\top X)^{-1} X_i$, where $X_i^\top$ denotes the $i$-th row of $X$, because it can be easily shown that $c_1 x_*^2 \le M_*\le c_2 x_*^2$ for $c_1 = 1/sn\kappa^2$ and $c_2 = p/sn$.}

Our \name-based test is therefore consistent as long as 
\begin{itemize}
\item[(i)] $\kappa^4 p^2 x_*^2 = o(n)$, i.e., there is no excessive leverage or
multicollinearity in the design as captured by $x_*$ and $\kappa$, respectively; and 
\item[(ii)] the signal (i.e., $||\beta||^2$) dominates the level of multicollinearity and noise, so 
that $\kappa^4 \sigma_*^2/n = o(||\beta||^2/p)$. 
This determines the ``detection radius" of our test, which we discuss next.~$\blacksquare$
\end{itemize}

Additionally, under sub-Gaussian errors, we can improve the upper bounds in 
Theorem~\ref{thm:nonasymp_power} and establish exponential decay rates over certain model characteristics as follows.
\begin{theorem}\label{thm:nonasymp_power_subgauss}
Suppose Assumptions \ref{asm:err1} and \ref{asmp:regularity} hold, and that $\beta\neq 0$ in \eqref{eq:lm} and $p<n$. Further suppose that $(\eps_i)_{i=1}^n$ are independent sub-Gaussian random variables, such that $\E \exp(t \eps_i )\le \exp(t^2 v/2)$ for some $v>0$. If $3\lambda \le \sigma_{\min}^2$, then
\begin{equation}\label{eq:nonasymp_lb}
    \E_X(1 - {\psi}_\alpha) = O(p) \exp\left\{- C_n(X)\right\} + O(n+p) \exp\left\{- D_n(X) f(\|\beta\|) \right\}~,
\end{equation}
where $f(x) = \min\{x, x^{2/3}\}$, $C_n(X) = 0.03 s^2 n/\kappa^2 p^2 x_*^4$ and 
$
    D_n(X) = \min\left\{ \sqrt{\frac{s^2 n}{p v x_*^2\kappa^2}}, \sqrt[3]{\frac{s^2 n^2}{p v x_*^2\kappa^2}} \right\}~. 
$
\end{theorem}
{\em Remarks.}
Theorem~\ref{thm:nonasymp_power_subgauss} provides a Type II error bound with exponential decay, which significantly improves upon Theorem~\ref{thm:nonasymp_power}. 
To better understand the decay rates, consider $s= O(1)$ and $v=O(1)$. Then we have
\begin{equation*}
    C_n(X) \asymp n/\kappa^2 p^2 x_*^4~,~~
    D_n(X) \asymp \sqrt{\frac{n}{p x_*^2\kappa^2}} \min\left\{ 1, (np x_*^2\kappa^2)^{1/6} \right\} \ge \sqrt{\frac{n}{p x_*^2\kappa^2}}~.
\end{equation*}
Here, the last inequality holds if $x_*$ and $\kappa$ are greater than one, which is usually true in practice. We highlight that $C_n(X)$ and $D_n(X)$ are analogous to $A_n(X)$ and $B_n(X)$ in Theorem~\ref{thm:nonasymp_power}, since $C_n(X)$ captures the error induced by design characteristics and $D_n(X)f(\|\beta\|)$ captures 
the signal-to-noise ratio. 

It may appear limiting that the results in Theorem~\ref{thm:nonasymp_power} and Theorem~\ref{thm:nonasymp_power_subgauss} require that we select the ridge penalty as $\lambda < \sigma_{\min}^2 = sn$. 
In Examples \ref{ex:gauss} and \ref{ex:uniform} that follow, 
we consider i.i.d. designs that require the condition $\lambda= o(n)$.
However, it is important to note that these conditions are automatically satisfied if one selects the ridge penalty parameter according to standard results in the literature. 
For instance,  \citet[Section 1.4]{lecture_ridge} suggests choosing $\lambda = O(1/\|\beta_{np}\|^2)$, and Theorem 1 of \cite{Duchi2022} provides $\lambda = p\sigma_*^2$ as an optimal ridge penalty parameter. Observe that $1/\|\beta_{np}\|^2 = o(n/p)$ 
as $\|\beta_{np}\|^2 = \Omega(p/n)$
and so $p\sigma_*^2 = o(n)~\text{as}~p = o(n^{0.5-\delta})$, whenever 
the error variance is bounded, such that $\sigma_* = O(1)$.
Thus, standard literature implies penalty selections that satisfy our conditions.~$\blacksquare$

\subsubsection{Asymptotic Results}
\label{sec:asymptotic}
In this section, we leverage Theorem~\ref{thm:nonasymp_power} to analyze the asymptotic power of our test as $n\to\infty$.
The power analysis considers the following alternative hypothesis space:
\begin{equation*}
    H_1: \beta\in\Theta(d),~\Theta(d) = \{~\beta\in\R^p:\|\beta\|\ge d~\}~.
\end{equation*}
This particular definition of alternative space, which has also been considered by many prior studies as well~\citep{Ingster2010},  leads to the concept of the {\em detection radius} of a test. 
\begin{definition}[Detection radius]\label{def:radius}
For any two sequences $(a_n)_{n = 1}^\infty$ and $(b_n)_{n = 1}^\infty$, let $a_n = \Omega(b_n)$ denote that $a_n/b_n\to\infty$ as $n \to \infty$.
Then, 
\begin{itemize}
    \item In the fixed design setting, a test $\psi$ has a detection radius $r_{np}(X)$, if for any sequence $d_{np} = \Omega(r_{np}(X))$, it holds that $\lim_{n \to \infty} \sup_{\beta\in\Theta(d_{np})}\E_X(1 - \psi)=0$. 
    \item In the random design setting, a test $\psi$ has a detection radius $r_{np}$, if for any sequence $d_{np} = \Omega(r_{np})$, it holds that $\lim_{n \to \infty} \sup_{\beta\in\Theta(d_{np})}\E(1 - \psi)=0$. 
\end{itemize}
\end{definition}
That is, the detection radius provides a sufficient condition on how strong the signal should be to guarantee the consistency of a test. In simpler terms, a smaller detection radius signifies a more powerful test. 
Our goal in this section will be to derive the detection radius of our \name-based test in both fixed and random design settings, and also analyze our test within the minimax framework of \cite{Ingster2003}. 

\paragraph{Fixed design.}
Based on Theorem~\ref{thm:nonasymp_power}, we can obtain the detection radius of our test in the fixed design setting as follows; see Appendix \ref{sec:asymp_powerproof} for the proof. 
\begin{theorem}\label{thm:consistent}
In the fixed design setting, suppose that Assumptions~\ref{asm:err1} and \ref{asmp:regularity} hold, and 
\begin{equation}\label{eq:A4}
    p^2 \kappa^4 x_*^2 = o(sn)~\text{and}~\lambda = o(sn), ~\text{as }n\to\infty.
\end{equation}
Then, ${\psi}_\alpha$ has detection radius
\begin{equation*}
    r_{np}(X) = B_n^{1/2}(X) = \sqrt{\frac{p \kappa^4 \sigma_*^2}{sn} + \frac{p}{sn}}~.
\end{equation*}
\end{theorem}
{\em Remarks.} Intuitively, the first condition in Equation~\eqref{eq:A4} determines the acceptable rate of increase for the design 
characteristics (such as leverage, multicollinearity, dimension), akin to the Type II error results of Theorem~\ref{thm:nonasymp_power}.
The second condition, $\lambda = o(sn)$, provides an oracle rule for choosing the ridge penalty parameter.~$\blacksquare$

\paragraph{Random design.}
In the random design setting, the design characteristics, $\kappa = \sigma_{\max}/\sigma_{\max}$, $s = \sigma_{\min}^2/n$, and $x_*$ of Theorem~\ref{thm:nonasymp_power} are now random. This complicates the analysis as it requires us to establish a detection radius that does not depend on $X$. 
Our approach is then to employ appropriate concentration bounds in place of these random quantities. 
Specifically, suppose that there exist sequences $(m_n)_{n = 1}^\infty$, $(M_n)_{n = 1}^\infty$, and $(C_n)_{n = 1}^\infty$ such that with probability $1 - o(1)$,
\begin{align*}
    \sigma_{\min}^2\ge m_n,~\sigma_{\max}^2\le M_n,~\text{and}~x_*\le C_n~.
\end{align*}
Then, under a random design, we can obtain the following result on the detection radius of our test using the bounds $m_n$, $M_n$, and $C_n$. The proof can be found in Appendix~\ref{sec:asymp_powerproof}.
\begin{theorem}\label{thm:consistent_random}
In the random design setting, suppose that Assumptions~\ref{asm:err1} and \ref{asmp:regularity} hold, and 
\begin{equation}\label{eq:A4_random}
    p^2 M_n^2 C_n^2 = o(m_n^3),~\text{and}~\lambda = o(m_n)~.
\end{equation}
Then, ${\psi}_\alpha$ has detection radius
\begin{equation*}
    r_{np} = \sqrt{\frac{p M_n^2 \sigma_*^2}{m_n^3} + \frac{p}{m_n}}~.
\end{equation*}
\end{theorem}
To build intuition about Theorem~\ref{thm:consistent_random},  we will narrow our focus to condition \eqref{eq:A4_random}, 
which plays a central role in the theorem. This discussion will also illustrate the technical conditions dictating our focus on the $p < n$ regime.

\begin{example}[Gaussian design]\label{ex:gauss}
Suppose that $X_{ij}\stackrel{iid}{\sim}\cN(0, 1)$, $X\indep \eps$ and $\sigma_*= O(1)$. One may choose $m_n = n/2, M_n = 2n$ based on matrix concentration inequalities \citep{Tropp2015}, and choose $C_n = \sqrt{4\log(np)}$ based on maximal inequalities of Gaussian random variables (see Appendix~\ref{sec:asymp_powerproof} for details). Then, Equation~\eqref{eq:A4_random} reduces to
\begin{equation*}
    p = o(n^{0.5-\delta}),~\text{and}~\lambda = o(n)~,
\end{equation*}
for some $\delta>0$, and the detection radius 
reduces to $r_{np} \asymp \sqrt{p/n}$. 
\end{example}

\begin{example}[Uniform design]\label{ex:uniform}
Suppose $X_{ij}\stackrel{iid}{\sim}\Unif([0, 1])$, $X\indep \eps$ and $\sigma_*= O(1)$. One may choose $m_n = O(n), M_n = n$ based on matrix concentration inequalities \citep{Tropp2015}, and choose $C_n = 1$ due to the boundedness of $X_{ij}$. Then, Equation~\eqref{eq:A4_random} reduces to
\begin{equation*}
    p = o(n^{0.5}),~\text{and}~\lambda = o(n)~,
\end{equation*}
and the detection radius again reduces to $r_{np}\asymp \sqrt{p/n}$. 
\end{example}

\subsubsection{Minimax Optimality of the Global Null Test}\label{sec:minimax}
In the previous section, we examined the detection radius of our test for the global null, which provides a sufficient condition for the test's consistency.
However, it remains unclear whether the detection radius also provides a condition that is necessary. 
In this section, we show that our test achieves minimax optimality, a concept we will explain next.

\begin{definition}[\cite{Ingster2010}]
\label{def:Rd}
Define the minimax testing error as
\begin{equation*}
    \mathcal{R}(d):=\inf_\psi \big(\E_0\psi + \sup_{{\beta} \in \Theta(d)} \E (1 - \psi)\big)~,
\end{equation*}
where the infimum is taken over all tests $\psi$ that map from the sample space to $\{0,1\}$.
\end{definition}
The concept of minimax errors was proposed by \cite{wald1950}, and was later adapted 
by \cite{Ingster1995,Ingster2003} for hypothesis testing. The idea behind Definition~\ref{def:Rd} is to measure the performance of a test by its combined Type I-II error, and then quantify how hard a testing problem is based on the best-performing test. 
It is easy to verify that $\mathcal{R}(d)\in[0, 1]$ and that it is non-increasing in $d>0$. 
This is crucial as it conveys that the testing problem gets harder as the signal gets weaker. In particular, if $\lim_{n\to\infty}\mathcal{R}(d_{np}) = 1$ for a sequence of radii $(d_{np})_{n = 1}^\infty$ then no test can distinguish between the null hypothesis and the sequence of alternatives determined by that sequence. In other words, $d_{np}$ determines the ``least detectable" signal strength for our testing problem. 
This leads to the following concept of minimax optimality.
\begin{definition}
A testing procedure is minimax optimal if it has a detection radius $r_{np}$ and $\lim_{n\to\infty} \mathcal{R}(d_{np})=1$ for any sequence $d_{np} = o(r_{np})$.
\end{definition}
Intuitively, if a sequence $r_{np}$ corresponds to the 
detection radius of a test {\em and} also exactly quantifies the least detectable signal strength, then that 
test is minimax optimal. Getting back to the global null hypothesis in \eqref{eq:H0}, the following result by \cite{Ingster2010} establishes a remarkable minimax optimality result under the independent Gaussian design.
\begin{lemma}[Theorem 4.1 of \cite{Ingster2010}]\label{lem:ingster}
Under the linear model~\eqref{eq:lm} suppose that $X_{ij}\stackrel{iid}{\sim}\cN(0, 1)$, $\eps_i\stackrel{iid}{\sim}\cN(0, \sigma^2)$, $X\indep \eps$, and $p<n$. Then, for the global null testing problem in~\eqref{eq:H0}, it holds that $\lim_{n\to\infty} \mathcal{R}(d) = 1$ for any 
sequence $d_{np} = o(\sqrt{p^{1/2}/n})$.
\end{lemma}

From Lemma~\ref{lem:ingster}, to prove that a test is minimax optimal (under the setup of the lemma) it suffices to show that the test has a detection radius $\sqrt{p^{1/2}/n}$. Below, we prove that our \name-based test is indeed minimax optimal under a Gaussian setup similar to~\cite{Ingster2010}.
This setup allows a sharper analysis than Example \ref{ex:gauss}; Appendix \ref{sec:asymp_powerproof} has
the complete details of the proof.
\begin{theorem}[Minimax optimality]\label{thm:normal_minimax}
Suppose that Assumption~\ref{asm:err1} holds. In addition, suppose that $X_{ij}\stackrel{iid}{\sim }\cN(0, 1)$, the errors $(\eps_{i})_{i=1}^n$ are i.i.d. with zero mean and 
finite fourth moments. Then, if $p = o(n^{0.5 - \delta})$ for some $\delta>0$ and $\lambda = 0$ (i.e., regular OLS), the test, ${\psi}_\alpha$, for the global null hypothesis has detection radius $r_{np} = \sqrt{p^{1/2}/n}$. Hence, ${\psi}_\alpha$ is minimax optimal. 
\end{theorem}
{\em Remarks.} To our best knowledge, Theorem~\ref{thm:normal_minimax} is the first result on the minimax optimality of \name-based tests (i.e., randomization tests) for the global null hypothesis in linear regression with symmetric errors. 
While the theorem concerns the Gaussian design setup, our test can achieve the optimal detection radius, $r_{np} = \sqrt{p^{1/2}/n}$, under more general settings, such as the uniform design discussed in Example~\ref{ex:uniform}. 

On a theoretical note,  Theorem~\ref{thm:normal_minimax} reveals an interesting trade-off
between the classical asymptotic approach of~\cite{Ingster2010} and the \name-based approach we take in this paper.
On one hand, the \name-based test requires that $p=o(n)$ while \cite{Ingster2010} allow for $p = o(n^2)$. As discussed in Section~\ref{sec:intro},  the ``$p<n$ regime" is seemingly a fundamental constraint in the theory  of \name-based/randomization tests unless one is willing to impose additional regularity conditions on the test statistic.
Indeed, all prior work on power analysis of randomization tests has imposed the $p<n$ condition, as seen in \citep{Toulis2019, Lei2020, Wen2022}. 
Whether this constraint is an absolute limitation or not, however, remains an open question for future work.
On the other hand, the method of \cite{Ingster2010} relies on the assumption of standard normal errors, whereas our test does not impose any additional distributional assumptions on the errors beyond symmetry and finite moments.~$\blacksquare$

\section{A Residual Test for the Partial Null Hypothesis}\label{sec:partial}
In this section, we switch our focus to the partial null hypothesis, $H_0^S: \beta_S = 0$, in model~\eqref{eq:lm}. To this end, we write the linear model \eqref{eq:lm} as 
\begin{equation}\label{eq:lm_partial}
    y = X_S\beta_S + X_{S^\complement} \beta_{S^\complement} + \eps~,
\end{equation}
where $S^\complement = \{1, \dots, p\} \setminus S$ and $\beta_{S^\complement}$ are the corresponding coefficients. As in the previous section, we will maintain Assumption~\ref{asm:err1} of error symmetry. As we discuss throughout this section, 
our results are applicable beyond this invariance.

Compared to the global null, the partial null hypothesis presents one unique challenge that does not allow the application of the \name-based test of Section~\ref{sec:rand:lm}.
Concretely, under $H_0^S$, it holds that $y = X_{S^\complement}\beta_{S^\complement} + \eps$,
and so the true errors $\eps$ are no longer available for Step 2 of the test. Our approach to address this challenge will thus be twofold. 

First, inspired by ideas in~\citep{candes2018panning, Lei2020, Wen2022}, we define a test statistic that ``knocks out" the nuisance parameters ($\beta_{S^\complement}$), and so is a known function of $\eps$.
Second, inspired by the ``residual randomization" framework of \citet{Toulis2019}, we perform a test as in Procedure 1 but use the residuals calculated under the partial null hypothesis in place of the unknown errors. In particular, we propose to use the test statistic
\begin{equation*}
    t_A(y, X) = \|Ay\|^2~,~\text{where $A$ is a projection matrix such that $AX_{S^\complement} = 0$}~.
\end{equation*}

The selection of the projection matrix $A$ (i.e., symmetric and idempotent) is crucial for the behavior of the test, but we defer this discussion to Section~\ref{sec:A} that follows.
Given a valid matrix $A$, it holds that
\begin{equation}\label{eq:t_partial}
    t_A(y, X) = \|Ay\|^2 \stackrel{}{=} \|A(X_{S^\complement}\beta_{S^\complement} + \eps)\|^2 = \|A\eps\|^2 := t(\eps)~,
\end{equation}
whenever $H_0^S$ is true.
Thus, an ``oracle" $p$-value could be calculated using Procedure 1 of Section~\ref{sec:rand:lm}. 
In practice, however, $\eps$ are unknown under the partial null hypothesis, 
and so we use residuals calculated under that null. To make sure 
that these residuals are well-defined we introduce an assumption that is analogous to Assumption~\ref{asmp:regularity} in the global null case. 
\begin{assumption}\label{asmp:regularity-partial}
Let $\sigma_{\min}$ and $\sigma_{\max}$ be the minimum and the maximum singular value of $X_{S^\complement}$, respectively. 
We assume that $\sigma_{\min}>0$ and $p-|S|<n$ at any sample size $n>0$.
\end{assumption}
We are now ready to describe our testing procedure.
\begin{center}
\textsc{Procedure 2: Residual \name-based test for the partial null $H_0^S$ in \eqref{eq:lm}}
\end{center}
\begin{enumerate}
    \item Compute the observed value of the test statistic, $T_n = \|A y\|^2$.
    \item Compute the randomized statistic $t(G_r \heps)$, where $\heps = y - X_{S^\complement}(X_{S^\complement}^\top X_{S^\complement})^{-1} X_{S^\complement}^\top y$ and $G_r\stackrel{iid}{\sim}\Unif(\cG)$, $r = 1, \dots, R$, and
    $t(u) = ||Au||^2$ as in~\eqref{eq:t_partial}.
    
    \item Obtain the test decision, $\psi_\alpha^S = \mathbb{I}\{\mathrm{pval} \le \alpha\}$, 
    where 
    \begin{equation*}
        \mathrm{pval} = \frac{1}{R+1} \left(1 + \sum_{r = 1}^R \mathbb{I}\big\{t(G_r\heps)>T_n\big\}\right)~. 
    \end{equation*}
    \item \textbf{Inference.} Let $\psi^j_{\alpha}(b)$ denote the test function for $H_0^j: \beta_j=b$ at level $\alpha$, and define
    \begin{equation*}
        \mathrm{CI}_j=\left\{x \in \mathbb{R}: \min \left(L_{j, \alpha}\right) \leq x \leq \max \left(L_{j, \alpha}\right)\right\},\quad L_{j, \alpha}=\left\{x \in \mathbb{R}: \psi^j_{\alpha}(x)=0\right\}
    \end{equation*}
    Then, $\mathrm{CI}_j$ is the $100(1-\alpha) \%$ confidence interval for $\beta_j$.
\end{enumerate}


\subsection{Asymptotic Validity of Procedure 2 for the Partial Null}
Unlike Procedure 1 for the global null hypothesis, the \name-based test in Procedure 2 
is not finite-sample valid due to using residuals instead of the true errors.
In this section, we derive conditions under which Procedure 2 is {\em asymptotically valid} in the sense that its Type I error is $\alpha + o(1)$.

To that end, we leverage a critical condition 
for the asymptotic validity of approximate randomization tests similar to Procedure 2~\citep[Theorem 1]{Toulis2019}:
\begin{equation}\label{eq:cond}
    \frac{\E_X\left[\big (t(G \widehat{\varepsilon})-t(G \varepsilon)\big)^2\right]}{\E_X\left[\big(t\left(G^{\prime} \varepsilon\right)-t\left(G^{\prime \prime} \varepsilon\right)\big)^2\right]} \to 0~,
\end{equation}
where $G, G', G''\sim \Unif(\cG)$ i.i.d. Roughly speaking, this condition guarantees that the natural variation of the oracle \name-based test that uses the true errors (denominator)
dominates the approximation error of the approximate \name-based test that uses the residuals (numerator). 
Our approach will thus be to derive an upper bound for the ratio in Equation~\eqref{eq:cond}, and then derive concrete and interpretable conditions to guarantee the asymptotic validity of our test.

For our analysis, we assume the fixed design setup as in \citep{Toulis2019}, and use $\E_X$ to denote the expectation taken under a nonrandom design matrix $X$. 
Moreover, we impose the following assumption on the true errors $\eps$.  
\begin{assumption}\label{asm:err_homo}
    $(\eps_i)_{i=1}^n$ satisfy Assumption~\ref{asm:err1} (sign symmetry). 
    Moreover, $\eps_i$ are independent with mean zero and same variance, and $\max_i \E(\eps_i^4) < \infty$ for any sample size $n$. 
\end{assumption}
This assumption can be relaxed to include heteroskedastic errors, but the analysis and the technical statements get more complicated. We refer the interested reader to Appendix~\ref{sec:partial_proof} for a discussion on such technical conditions.

In the partial null setting, we let $x_* \coloneqq \max_{ij} |X_{S^\complement,ij}|$. As before, we define $\kappa = \sigma_{\max}/\sigma_{\min}$ and $s=\sigma_{\min}^2/n$ to ease interpretation. In addition, let $k= \mathrm{rank}(A)$ and $a_n = \sum_{i = 1}^n A_{ii}^2/k$ denote, respectively, the rank of $A$ and its ``diagonal concentration measure". 
The lemma below provides an upper bound for the ratio~\eqref{eq:cond}, and thus the critical condition for asymptotic validity of Procedure 2; see Appendix~\ref{sec:partial_proof} for the proof. 

\begin{lemma}\label{lem:partial_null_cond2}
    Suppose Assumptions \ref{asmp:regularity-partial} and \ref{asm:err_homo} hold. Then, under $H_0^S$, if $a_n<1$ then we have
    \begin{align*}
        \frac{\E_X\left[\left(t(G \heps)-t(G \eps)\right)^2\right]}{\E_X\left[\left(t\left(G^{\prime} \eps\right)-t\left(G^{\prime \prime} \eps\right)\right)^2\right]}
        &= O\left(\sqrt{\frac{k(p - |S|)^3x_*^2\kappa^6}{(1-a_n)^2 sn}}\right)\;.
    \end{align*}
\end{lemma}

\begin{theorem}\label{thm:partial_null_valid}
    Suppose that Assumptions \ref{asmp:regularity-partial} and \ref{asm:err_homo} hold. Suppose also that:
    \begin{enumerate}[(a)]
    \item The projection matrix $A$ satisfies $AX_{S^\complement} = 0$.
        \item 
        $\max_{g, g^{\prime} \in \cG, g \neq g^{\prime}} \P\big(\|Ag\eps\|=\|Ag'\eps\|\big)=o(1).$

 \item $\lim\sup_{n\to\infty} a_n < 1$ and $k(p - |S|)^3x_*^2\kappa^6 = o(sn)$.
    \end{enumerate}
Then, Condition \eqref{eq:cond} holds and the rejection probability of the residual \name-based test of Procedure 2 satisfies $\E_X(\psi_{\alpha}^S)\le \alpha + O(1/R) + o(1)$.
\end{theorem}

{\em Remarks.}
(a) Condition 7(a) can be satisfied with various choices of $A$. We discuss this in more detail and through concrete examples in the following section. Condition 7(b) ensures that the test statistic is not degenerate. It may be satisfied trivially in many settings, e.g., when $\eps$ are continuous random variables.

(b)
Regarding Condition 7(c), it is straightforward to show that $a_n \le 1$, with equality if and only $A$ has  $k$ diagonal elements equal to one and the rest equal to zero. 
As $A_{ii}$ may be interpreted as ``leverage scores" (see next section), 
the condition rules out cases with $O(k)$ high-leverage observations.
Moreover, this condition requires that the problem dimension 
($p - |S|$) and the design's multicollinearity and maximum covariate ($\kappa$ and $x^*$, respectively) grow moderately relative to $n$. Note that we don't need $p<n$ for asymptotic validity.

(c) The $O(1/R)$ term in Theorem~\ref{thm:partial_null_valid} can be `eliminated' to obtain an asymptotically level-$\alpha$ test, e.g., by randomizing the test decision by a factor $\alpha/[\alpha + (R+1)^{-1}] = 1 - O(1/R)$.$\blacksquare$
%

\subsection{A Data-Driven Procedure for Selecting $A$}
\label{sec:A}
In this section, we discuss the selection of matrix $A$ for Procedure~2 designed to improve testing power.
We start with some concrete examples, and then we propose a data-driven procedure for $A$ that takes the conditions of Theorem~\ref{thm:partial_null_valid} into account.

\begin{example}\label{ex:partial1}
    Suppose that each column of $X_{S^\complement}$ is centered so that $\mathbf{1}^\top X_{S^\complement} = 0$, where $\mathbf{1}$ denotes the all-one vector. We may choose $A = \mathbf{1}\mathbf{1}^\top /n$, which satisfies $AX_{S^\complement} = 0$. Then, $\lim\sup_{n\to\infty} a_n < 1$ automatically holds since $a_n = 1/n \to 0$, and Condition 7(c) reduces to $(p - |S|)^3 x_*^2 \kappa^6 = o(sn)$ because $k = 1$.
\end{example}
\begin{example}\label{ex:partial2}
    Consider $A = Z(Z^\top Z)^{-1} Z^\top$ with
    \begin{equation*}
        Z = (I_n - P_{S^\complement}) X_S~, P_{S^\complement} = X_{S^\complement} (X_{S^\complement}^\top X_{S^\complement})^{-1} X_{S^\complement}^\top~.
    \end{equation*}
    It is easy to verify that $AX_{S^\complement} = 0$. Then, the second part of Condition 7(c) reduces to $|S| (p - |S|)^3 x_*^2 \kappa^6 = o(sn)$.
\end{example}

Examples \ref{ex:partial1} and \ref{ex:partial2} guarantee Condition 7(a) and 
clarify Condition 7(c) 
for the asymptotic validity of the partial null test.
One reasonable choice for $A$ would thus be 
\begin{equation}\label{eq:chooseA}
A = Z(Z^\top Z)^{-1} Z,~Z = (I_n - P_{S^\complement}) \tilde X_S,
\end{equation}
where $ P_{S^\complement} = X_{S^\complement} (X_{S^\complement}^\top X_{S^\complement})^{-1} X_{S^\complement}^\top~$ is the projection matrix on the column space of $X_{S^\complement}$ and $\tilde X_S$ is the covariate matrix $X_S$ 
after variable selection, e.g., through Sure Independence Screening (SIS) or Lasso.
The idea is that under $H_0^S$, all variables in $X_S$ are non-significant 
 so that the number of selected variables will generally be  small~\citep{Fan2008, Zhao2006}.\footnote{If Lasso selects no variables, we simply use $\widetilde{X}_S = \mathbf{1}$, as in Example~\ref{ex:partial1}.}
Under the alternative ($\beta_S\neq 0$), such a choice can lead to high statistical power 
because the observed test statistic  will contain ``most of the signal" ($X_S\beta_S$) through the selected variables, $\tilde X_S$, 
whereas the residual-based values, $t(G\heps)$, will not contain the signal due to 
the random transformations of the residuals.

While it's important to note that our recommendation for choosing $A$ in Equation~\eqref{eq:chooseA} is not guaranteed to be optimal, 
our simulations in Section~\ref{sec:simu} consistently demonstrated  effective control of Type II errors. 
Moreover, we report reliable results in the real data examples presented in Section~\ref{sec:realdata} and 
Appendix~\ref{sec:gene}.

\section{Simulations}\label{sec:simu}

\subsection{Global Null Hypothesis in Linear Models}\label{sec:simu_global_linear}
In this section, we illustrate Procedure 1 for the global null hypothesis defined in~\eqref{eq:H0}.
We use the ridge test statistic~\eqref{eq:ridge} with $\lambda = \log(p/n)$ and $R = 2000$, and consider two realistic setups with different covariate and error structures. 
%
For comparison, we also include the tests proposed by \cite{Li2020,Cui2018} (described in Section~\ref{sec:related}), referred to as ``SF" and ``CGZ", respectively.

\subsubsection{Testing a Nonsparse $\beta$}
First, we consider nonsparse linear regression problems, where $X$ is a high-dimensional covariate matrix with a low-dimensional ``intrinsic structure"~\citep{Li2020}, and $\beta$ is nonsparse under the alternative. Concretely, we set:
\begin{itemize}
    \item $X\in\R^{n\times p}$ with $(n, p) = (50, 500)$.
    \item $X = W U\Lambda U^\top$, where $(W_{ij})_{i\in[n], j\in[p]}\stackrel{iid}{\sim}\cN(0, 1)$,  $\Lambda$ is a diagonal matrix with elements $(\lambda_i)_{i\in[p]}$, and $U$ is a randomly generated orthogonal matrix. We consider (1) ``slow-decay" in eigenvalues: $\lambda_i = \log^{-1}(i + 1)$; (2) ``fast-decay": $\lambda_i = i^{-1/2}$.
    \item $\beta_i \stackrel{iid}{\sim} \mathrm{Binom}(3,0.3)+0.3 \mathcal{N}(0,1)$.
    \item $\eps_i\stackrel{iid}{\sim}\cN(0, 1)$.
\end{itemize}
We also scale $\beta$ and set $\Sigma= U\Lambda^2 U^\top$ such that $\|\beta\|=0, 0.5, 1, 2$ and $\|\Sigma\|_F = 100, 300$.

\paragraph{Normal errors.} Table~\ref{tab:decay1} presents the rejection rates in percentage (\%) under both the global null hypothesis ($H_0:\beta=0$) and various alternatives ($\beta\neq 0$), based on 10,000 simulations at level $\alpha = 0.05$. We use bold figures to highlight the highest power under each alternative. 
In particular, Panel A of Table~\ref{tab:decay1} shows that under the null hypothesis (first column in each subpanel) with normal errors, all methods control the Type I error at the nominal level.
Under the alternative hypotheses ($\|\beta\|>0$), the power of the \name-based test (denoted as ``Inv") is larger than both SF and CGZ in most cases, and only slightly trails CGZ when $\|\beta\| = 0.5$ and $\|\Sigma\|_F = 100$.
By contrasting the ``slow-decay'' regime to ``fast-decay'', we observe that ``fast-decay'' improves the power for all three methods. 
This is not surprising because ``fast-decay'' in eigenvalues implies a smaller ``intrinsic dimension'' \citep{Li2020}, where the testing problem gets easier. Moreover, this is aligned with the Type II error bound in Theorem~\ref {thm:nonasymp_power}, which decreases as the dimension $p$ gets smaller.

\paragraph{Errors with heavy tails or unbounded variance.}
Panels B and C of Table~\ref{tab:decay1} present results for heavy-tailed designs. 
Such settings are challenging for standard asymptotic methods, since the covariates and errors may both be heavy-tailed. 
The underlined numbers highlight the uncontrolled type I errors ($>6\%$) when $\beta = 0$. 
Based on Panels B and C, we see that SF and CGZ fail to control the Type I error in these challenging settings. On the contrary, the \name-based test controls the Type I error at the nominal level in all setups. 
Under alternative hypotheses, the power of the \name-based test appears to be smaller, but this should not be a surprise as the analytical tests are generally invalid and over-reject, as shown above.

\renewcommand{\arraystretch}{1.1}
\begin{table}[t]
    \centering
    \begin{tabular}{cl| p{3em}p{3em}p{3em}p{3em}|p{3em}p{3em}p{3em}p{3em}}
    & & \multicolumn{4}{c|}{Slow-decay} 
    &\multicolumn{4}{c}{Fast-decay} \\
    \hline
    & &\multicolumn{4}{c|}{$||\beta||$} 
    &\multicolumn{4}{c}{$||\beta||$} \\
    \multicolumn{2}{c|}{} & 0 & 0.5 & 1 & 2 & 0 & 0.5 & 1 & 2 \\
    \hline
    \multicolumn{5}{l}{Panel A: Normal design, normal errors} &  \multicolumn{5}{c}{} \\
    \hline
    \multirow{3}{6em}{$\|\Sigma\|_F = 100$} & Inv &4.76 &22.94 &\textbf{49.11} &\textbf{67.38} &5.14 &22.34 &\textbf{60.87} &\textbf{89.11} \\
    & SF &4.93 &8.13 &12.42 &15.10 &5.09 &11.60 &25.99 &41.52 \\
    & CGZ &5.22 &\textbf{23.00} &41.86 &51.88 &5.02 &\textbf{26.36} &55.89 &74.99 \\
    \hline
    \multirow{3}{6em}{$\|\Sigma\|_F = 300$} &  Inv & 5.03& \textbf{43.23}& \textbf{65.10}& \textbf{73.34}& 4.71& \textbf{51.13}& \textbf{85.33}& \textbf{95.23}\\ 
    & SF& 5.01& 11.32& 15.11& 16.81& 4.96& 22.95& 38.74& 48.15\\ 
     & CGZ& 4.87& 37.66& 51.08& 55.28& 4.64& 49.50& 72.08& 80.64\\
    \hline


    \multicolumn{5}{l}{Panel B: ${\tt t}_1$ design, normal errors} &  \multicolumn{5}{c}{}  \\
        \hline
        \multirow{3}{6em}{$\|\Sigma\|_F = 100$} & Inv & 4.92& 65.70& 65.57& 66.32& 5.06& 89.64& 90.01& 90.25\\ 
        & SF& 4.95& \textbf{99.89}& \textbf{99.93}& \textbf{99.95}& 5.19& \textbf{99.89}& \textbf{99.80}& \textbf{99.85}\\ 
         & CGZ& \underline{8.56}& 82.93& 83.05& 82.79& \underline{7.81}& 84.81& 84.68& 84.38\\ 
        \hline
        \multirow{3}{6em}{$\|\Sigma\|_F = 300$} & Inv & 4.79& 54.36& 54.69& 55.03& 4.96& 80.51& 79.47& 80.76\\ 
        & SF& 4.96& \textbf{99.87}& \textbf{99.91}& \textbf{99.97}& 5.22& \textbf{99.84}& \textbf{99.87}& \textbf{99.86}\\ 
         & CGZ& \underline{9.01}& 83.76& 83.52& 82.94& \underline{7.56}& 84.55& 84.91& 84.98\\ 
\hline
  \multicolumn{5}{l}{Panel C: ${\tt t}_1$ design, ${\tt t}_1$ errors } & \multicolumn{5}{c}{}  \\
        \hline
        \multirow{3}{6em}{$\|\Sigma\|_F = 100$} & Inv & 5.24& 62.80& 65.46& 66.22& 4.73& 87.55& 89.68& 89.58\\ 
        & SF& \underline{16.27}& \textbf{99.90}& \textbf{99.94}& \textbf{99.97}& \underline{14.68}& \textbf{99.79}& \textbf{99.86}& \textbf{99.85}\\ 
         & CGZ& \underline{7.32}& 83.48& 83.51& 83.40& \underline{5.99}& 83.90& 84.92& 85.30\\ 
        \hline
        \multirow{3}{6em}{$\|\Sigma\|_F = 300$} & Inv & 4.70& 53.12& 53.04& 53.98& 5.00& 79.12& 80.13& 80.01\\ 
        & SF& \underline{17.38}& \textbf{99.96}& \textbf{99.92}& \textbf{99.93}& \underline{13.69}& \textbf{99.80}& \textbf{99.81}& \textbf{99.82}\\ 
         & CGZ& \underline{7.26}& 82.99& 83.06& 83.03& \underline{6.06}& 85.25& 85.02& 84.84\\ 
        \end{tabular}
    \bigskip
    \caption{Rejection rates in percentage (\%) for the nonsparse linear regression setup of Section~\ref{sec:simu_global_linear}. Bold indicates highest reject rate. 
    Underline indicates size distortion. }
    \label{tab:decay1}
\end{table}

\subsubsection{ Testing a Sparse $\beta$}
We now turn to sparse linear regression problems, where testing the global null becomes harder, and generate high-dimensional data $X$ from a moving average model following~\cite{Zhong2011,Cui2018}. Specifically, we set
\begin{itemize}
    \item $(n, p)= (40, 310), (60, 400), (80, 550)$.
    \item $\mu_j\stackrel{iid}{\sim}\Unif([2, 3])$, $\rho_l\stackrel{iid}{\sim}\Unif([0, 1])$ (fixed once generated), then for $T \in \{10, 20\}$,
    \begin{equation*}
        X_{i j}=\rho_1 W_{i j}+\rho_2 W_{i(j+1)}+\cdots+\rho_T W_{i(j+T-1)}+\mu_j;, \quad i\in[n], j\in[p],~~ W_{ij}\stackrel{iid}{\sim}\cN(0, 1).
    \end{equation*}
    As such, $T$ controls the extent of correlation structure in $X$. 
    \item $\beta_j = \mathbb{I}\{j \le 5\}$. After data generation, we scale $\beta$ such that $\|\beta\|^2=0, 0.02, 0.04, 0.06$. 
    \item Normal errors, $\eps_i\stackrel{iid}{\sim}\cN(0, 4)$, or
    lognormal errors, $\eps_i = B_i L_i$, where $L_i\sim\mathcal{LN}(0, \log 2)$ and $B_i\sim \operatorname{Bernoulli}(1/2)$ independently.
\end{itemize}
The rejection rates are collected in Table~\ref{tab:ultrahigh}, based on 10,000 simulations at level $\alpha = 0.05$. 
Same as Table~\ref{tab:decay1}, bold figures highlight the highest power under each alternative.
Under $H_0$ ($\|\beta\|=0$), all three tests have a correct control of Type I errors, as the rejection rates are roughly $5\%$. It is notable that under various sparse alternative hypotheses ($\|\beta\|>0$), our \name-based test has the highest power in all cases. These simulation results support our earlier claim that \name-based/randomization tests, despite their theoretical limitations, can be powerful under both sparse and nonsparse alternatives even in settings with $p>n$.

\begin{table}[t]
    \centering
    \begin{tabular}{cl| p{3em}p{3em}p{3em}p{3em}|p{3em}p{3em}p{3em}p{3em}}
    & & \multicolumn{4}{c|}{$T = 10$} 
    &\multicolumn{4}{c}{$T = 20$} \\
    \hline
    & &\multicolumn{4}{c|}{$||\beta||^2$} 
    &\multicolumn{4}{c}{$||\beta||^2$} \\
    \multicolumn{2}{c|}{} & 0 & 0.02 & 0.04 & 0.06 & 0 & 0.02 & 0.04 & 0.06 \\
    \hline
  \multicolumn{3}{c}{Panel A: Normal errors} &  \multicolumn{7}{c}{}  \\
\hline
    \multirow{3}{4em}{$(n,p) = (40, 310)$} & Inv & 4.78 & $\textbf{29.43}$ & $\textbf{55.91}$ & $\textbf{74.21}$ & 5.04 & $\textbf{28.74}$ & $\textbf{53.09}$ & $\textbf{71.33}$ \\ 
    & SF & 4.90 & 14.42 & 26.82 & 38.94 & 5.15 & 14.44 & 25.25 & 35.56 \\ 
    & CGZ & 5.24 & 6.28 & 7.85 & 9.12 & 4.51 & 9.32 & 13.76 & 20.13 \\ 
    \hline
    \multirow{3}{4em}{$(n,p) = (60, 400)$} & Inv & 4.65 & $\textbf{29.81}$ & $\textbf{57.16}$ & $\textbf{78.10}$ & 5.06 & $\textbf{35.05}$ & $\textbf{65.22}$ & $\textbf{83.91}$ \\ 
    & SF & 5.01 & 17.59 & 33.58 & 50.02 & 4.90 & 16.95 & 33.96 & 49.67 \\ 
    & CGZ & 4.80 & 7.23 & 9.26 & 12.59 & 4.59 & 7.93 & 11.59 & 16.33 \\ 
    \hline
    \multirow{3}{4em}{$(n,p) = (80, 550)$} & Inv & 4.83 & $\textbf{41.55}$ & $\textbf{75.77}$ & $\textbf{92.49}$ & 5.48 & $\textbf{51.23}$ & $\textbf{86.02}$ & $\textbf{97.04}$ \\ 
    & SF & 5.03 & 25.23 & 52.14 & 74.13 & 5.37 & 29.57 & 58.92 & 79.22 \\ 
    & CGZ & 4.48 & 6.46 & 8.40 & 12.13 & 5.30 & 11.78 & 20.98 & 32.67 \\ 
    \hline

    \multicolumn{3}{c}{Panel B: Lognormal errors} & \multicolumn{7}{c}{}  \\
     \hline
     \multirow{3}{4em}{$(n,p) = (40, 310)$} & Inv & 5.00 & $\textbf{34.78}$ & $\textbf{61.92}$ & $\textbf{77.62}$ & 5.11 & $\textbf{33.14}$ & $\textbf{58.97}$ & $\textbf{75.01}$ \\ 
     & SF & 4.88 & 16.43 & 31.36 & 44.34 & 5.37 & 16.34 & 28.36 & 39.91 \\ 
     & CGZ & 4.57 & 5.65 & 7.40 & 9.42 & 4.15 & 9.51 & 14.86 & 21.22 \\ 
     \hline
     \multirow{3}{4em}{$(n,p) = (60, 400)$} & Inv & 4.82 & $\textbf{34.01}$ & $\textbf{62.74}$ & $\textbf{80.40}$ & 5.05 & $\textbf{39.66}$ & $\textbf{69.65}$ & $\textbf{85.42}$ \\ 
     & SF & 4.85 & 19.28 & 37.48 & 54.20 & 4.94 & 19.25 & 37.70 & 54.51 \\ 
     & CGZ & 4.25 & 6.69 & 9.32 & 13.62 & 4.19 & 7.76 & 12.47 & 17.68 \\
     \hline
     \multirow{3}{4em}{$(n,p) = (80, 550)$} & Inv & 4.76 & $\textbf{46.34}$ & $\textbf{78.66}$ & $\textbf{92.44}$ & 5.02 & $\textbf{55.29}$ & $\textbf{86.59}$ & $\textbf{96.01}$ \\ 
     & SF & 4.50 & 27.11 & 55.33 & 75.82 & 5.29 & 31.98 & 62.42 & 80.17 \\
     & CGZ & 4.58 & 6.29 & 9.48 & 12.38 & 4.21 & 12.48 & 23.55 & 35.01 \\ 
    \end{tabular}
    \bigskip
    \caption{Rejection rates in percentage (\%) for the sparse linear regression setup of Section~\ref{sec:simu_global_linear}. }
    \label{tab:ultrahigh}
\end{table}

\subsection{Partial Null Hypothesis in Linear Models}\label{sec:simu_partial_linear}
In this section, we illustrate Procedure 2, the residual-based test of Section~\ref{sec:partial} for testing the partial null $H_0^S$ defined in~Equation~\eqref{eq:H0}. 
It is important to note that our method works for both low-dimensional ($p<n$) {\em and} high-dimensional 
settings ($p>n$). Here, we present only the high-dimensional results comparing to the debiased Lasso.
In Appendix~\ref{appendix:partial_lowdim}, we consider low-dimensional settings and 
compare our method to standard ANOVA.
In the high-dimensional experiment, in particular, we use the {\tt hdi} package to obtain the $p$-value for each coefficient, and then apply correction (Holm-Bonferroni) to test the partial null. See \cite{Dezeure2014HighDimensionalIC} for details. We set the following simulation parameters:
\begin{itemize}
    \item $n = 60$ and $p = 82$ with $S = \{1, \dots, 80\}$ and $S^\complement = \{81, 82\}$. That is, we inspect $|S|=80$ with two nuisance parameters. 
    \item For $X = (X_{S}, X_{S^\complement})$, generate $X_{ij}\stackrel{iid}{\sim}\cN(0, 1)$ or ${\tt t}_1$. 
    \item $\beta_{S^\complement}\in\R^2$ with each entry equal to $1/\sqrt{2}$. When evaluating the power, we consider $\beta_{S, j} = \mathbb{I}\{j\le 2\}/\sqrt{2}$. We rescale $\beta_S$ to inspect $\|\beta_S\| = 0.1, \dots, 1$. 
    \item Generate $\eps_i\stackrel{iid}{\sim}$ from either $\cN(0, 1)$, ${\tt t}_1$, 
 or ${\tt t}_2$.
\end{itemize}
To define our residual test, we follow the procedure outlined in Section~\ref{sec:partial}, and construct $A$ using $\widetilde{X}_S$, the selected significant variables from $X_S$ based on cross-validated 
Lasso.

The rejection rates for both tests are shown in Figure~\ref{fig:partial_highdim} under different covariate and error distributions, based on 10,000 simulations. The black horizontal line denotes the nominal level $\alpha=0.05$. 
We see that, under the partial null ($\|\beta_S\|=0$), both tests can have inflated Type I errors. Our 
residual test (denoted as  ``RR'') has a slight inflation, for instance, under Gaussian design and Gaussian noise. 
The debiased Lasso method (denoted as ``DLasso'') can have large Type I errors when the covariates are heavy-tailed, for instance, as in the ${\tt t}_1$ design and ${\tt t}_1$ noise setting. 
This supports our persistent claim throughout the paper that our \name-based test maintains robustness to heavy-tailed covariates and errors. Furthermore, the Type I error of our test asymptotes to 0.05 if we further increase the sample size. 
 Figure~\ref{fig:partial_highdim}(g) demonstrates this fact, by focusing on the Gaussian setting 
 and plotting the Type I error for $n = 40, 60, \dots, 120$. In this figure, 
 we note that the Type I errors are inflated for small $n$ (for both methods), but the errors asymptote to the nominal level as $n$ gets larger. This fact is aligned with our theoretical results in Theorem~\ref{thm:partial_null_valid}.

Next, under certain alternatives ($\|\beta_S\|>0$), the power of both tests is non-decreasing in $\|\beta_S\|$ and approaches one as $\|\beta_S\|$ gets larger. For Gaussian designs, our \name-based test dominates debiased Lasso for small $\|\beta_S\|$, but is less powerful for $\|\beta_S\|\ge 0.5$. The inferior power for large $\|\beta_S\|$ is related to the variable selection in $\widetilde{X}_S$. When constructing $\widetilde{X}_S$ for our test, Lasso sometimes fails to identify useful variables (in this setup, the first two variables) in $X_S$, which causes a loss of power. 
For ${\tt t}_1$ designs, the power comparison is meaningless as the debiased Lasso has significant Type I error inflation (e.g., approximately 50\% in the ${\tt t}_1$ design/${\tt t}_1$ noise setting).

\begin{figure}[t!]
    \centering
    \subfloat[Gaussian design, Gaussian noise]{{\includegraphics[width=.35\linewidth]{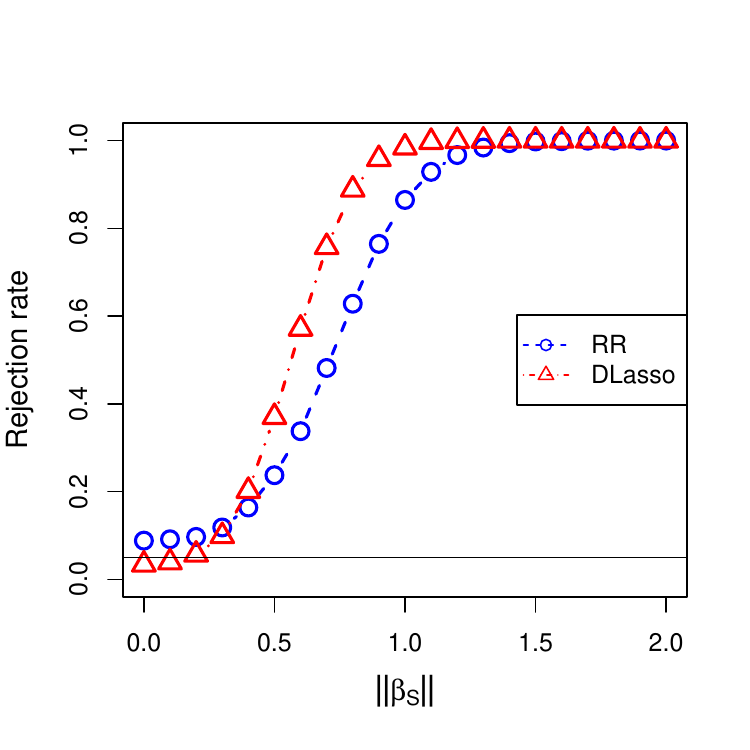}}}%
    \subfloat[Gaussian design, ${\tt t}_1$ noise]{{\includegraphics[width=.35\linewidth]{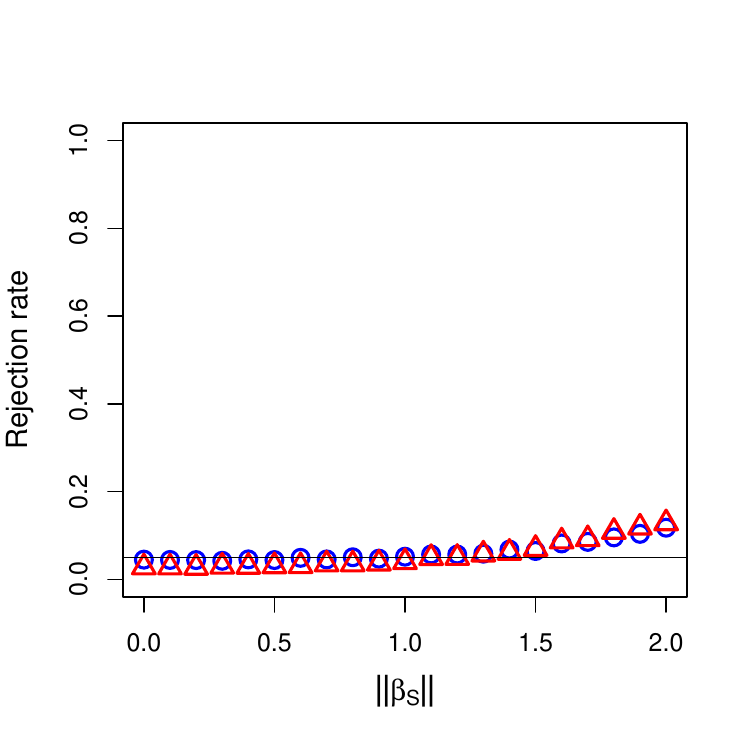}}}%
    \subfloat[Gaussian design, ${\tt t}_2$ noise]{{\includegraphics[width=.35\linewidth]{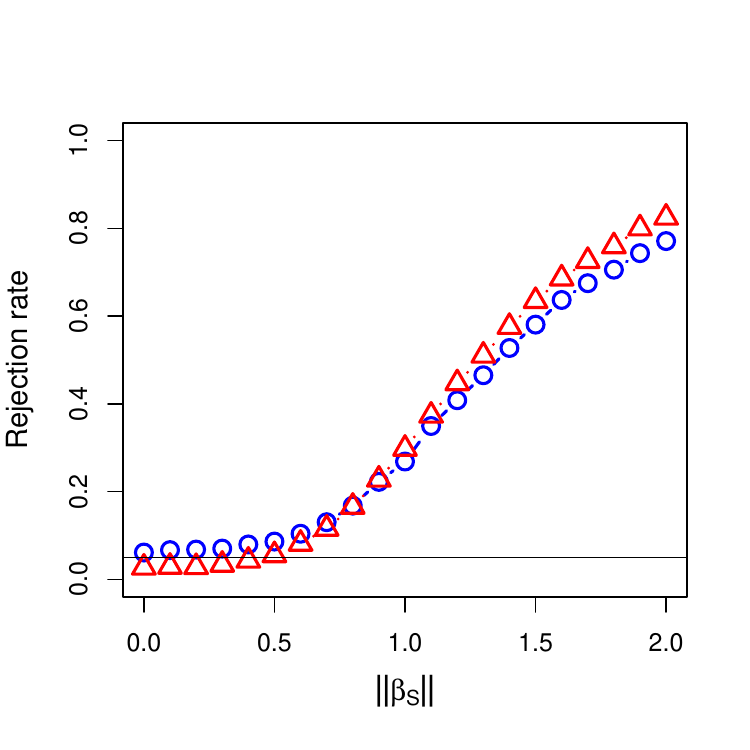}}}\\
    \subfloat[${\tt t}_1$ design, Gaussian noise]{{\includegraphics[width=.35\linewidth]{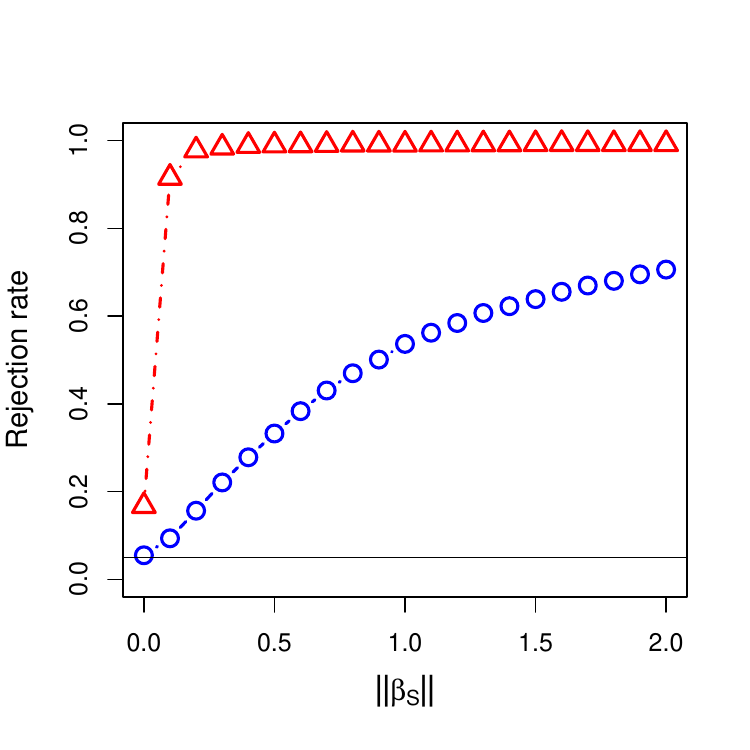}}}%
    \subfloat[${\tt t}_1$ design, ${\tt t}_1$ noise]{{\includegraphics[width=.35\linewidth]{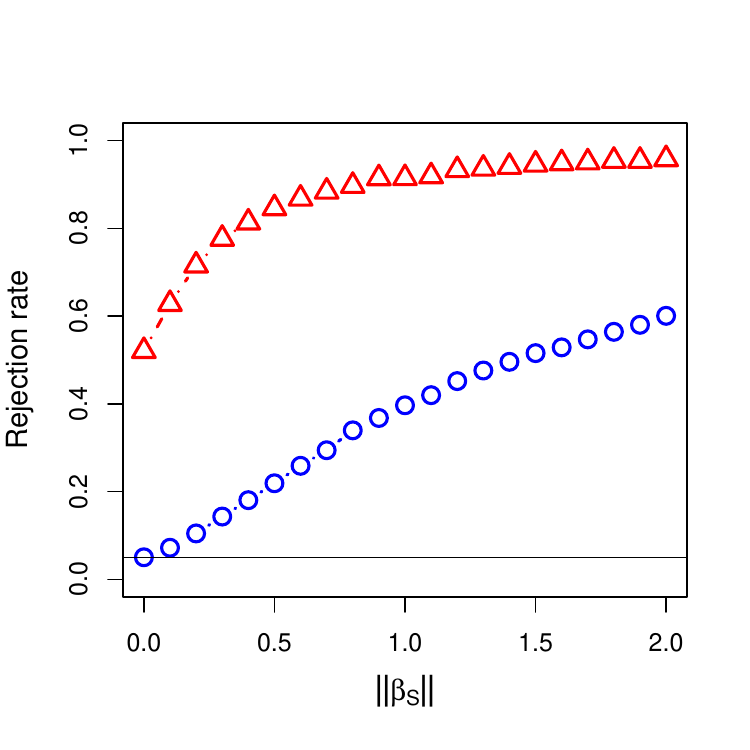}}}%
    \subfloat[${\tt t}_1$ design, ${\tt t}_2$ noise]{{\includegraphics[width=.35\linewidth]{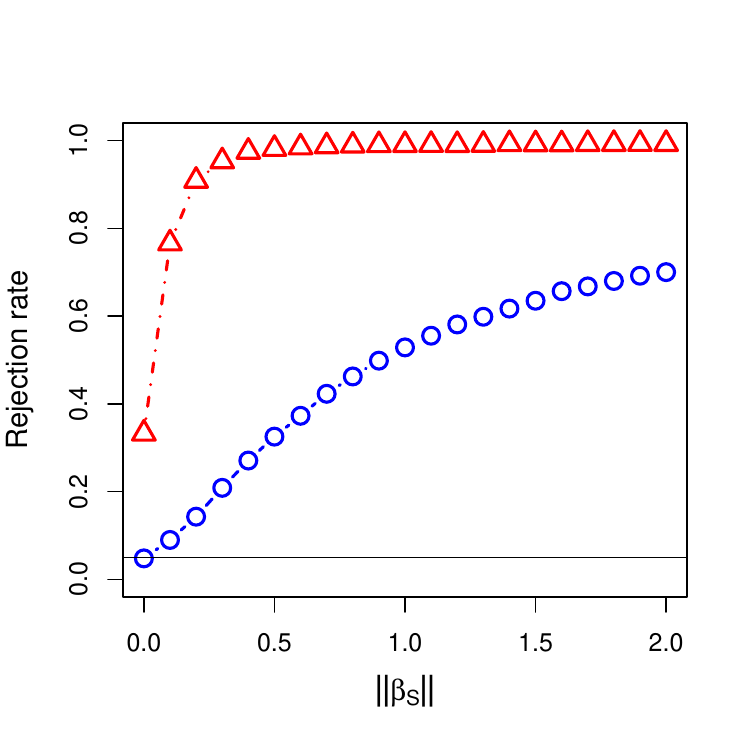}}}\\
    \subfloat[Type I errors under $\beta_S=0$, with Gaussian design and Gaussian noise]{{\includegraphics[width=.35\linewidth]{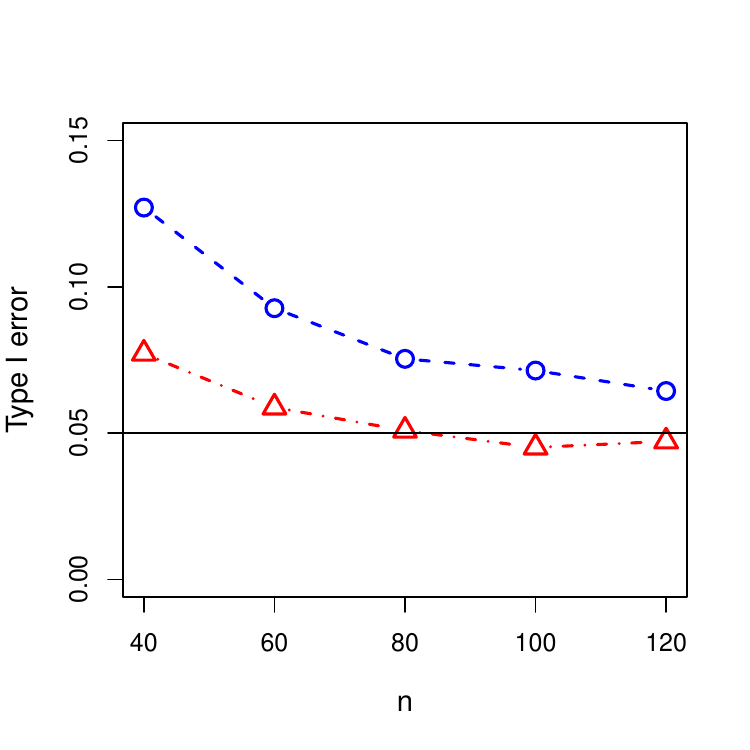}}}%
    \caption{Rejection rates under the partial null and alternative hypotheses. }
    \label{fig:partial_highdim}
\end{figure}

\section{Real Data Example: Abortion and Crime}\label{sec:realdata}
In a seminal paper, \cite{Donohue2001} argued that the expansion of abortion practices reduced subsequent levels of crime. 
Many papers have re-examined the data~\citep{Belloni2013, foote2008}, and some have contested the original findings. 
Without studying in depth the subject-matter question of abortion and crime, 
here we apply our \name-based methods to test the global and partial null hypotheses within the prevailing econometric specifications of the problem.

In particular, following \citep{Belloni2013} we consider the model 
\begin{equation}\label{eq:belloni}
    y_{c i t}-y_{c i t-1}=\alpha_c\left(a_{c i t}-a_{c i t-1}\right)+z_{c i t}^{\top} \kappa_c+\eps_{c i t}~,
\end{equation}
where $i = 1, \dots, 48$ is an index for U.S. state, 
$t = 1, \dots, 13$ is an index for time, 
$c\in\{\text{violent}, \text{property}, \text{murder}\}$ indexes type of crime; $y_{cit}$ are crime rate data, 
$a_{cit}$ are measures of abortion rates that depend on $c$, $z_{cit}$ are control variables,\footnote{$z_{cit}$ contains original control variables of state-level factors in \cite{Donohue2001} and their interactions and higher order terms constructed by \cite{Belloni2013}.} and $\eps_{cit}$ are unobserved errors; $y_{cit}$ and $a_{cit}$ are differenced to eliminate constant state effects as in \citep{Belloni2013}. In summary, the model presented in~Equation~\eqref{eq:belloni} has $48\times (13-1) = 576$ observations and $284$ control variables, and we are interested in testing the global null 
hypothesis and the partial nulls on $\alpha_c$, i.e.,
$$
H_0: (\alpha_c, \kappa_c) = 0, ~~\text{and}~~H_0^S:\alpha_c = 0, 
\text{for}~c\in\{\text{violent}, \text{property}, \text{murder}\}
$$

\paragraph{Testing for the global null.}
We test the global null hypothesis $H_0$ for crime type $c$, 
using our global test in Procedure 1 of Section~\ref{sec:global}, setting $\lambda=1$ and $R = 50,000$. The finite-sample valid $p$-values  for each crime type are shown below:
\begin{center}
     \begin{tabular}{c|c|c}
    \hline
    Violent & Property & Murder\\
    \hline
    $3.40\times 10^{-4}$ & $2.00\times 10^{-5}$ & 0.979\\
    \hline
    \end{tabular}
\end{center}
Based on these $p$-values, we can reject the global null hypothesis for violent and property crimes, but not for murder crimes. 
This indicates that violent and property crimes are significantly associated with covariates, whereas the association between murder crimes and covariates is not significant. 
This non-significance result aligns with \cite{Belloni2013}, who showed that a  Lasso specification on murder crime rates selects no variable.

\paragraph{Inference for the treatment effect.}
Here, we apply our residual test on the partial null hypotheses of the form $\alpha_c = a$ for $a\in \R, c\in\{\text{violent, property, murder}\}$, and construct confidence intervals via test inversion as described in Section~\ref{sec:partial}. To implement our test, we follow the strategy in Section~\ref{sec:partial} and construct the test statistic by choosing $A = X_c(X_c^\top X_c)^{-1} X_c^\top$, where $X_c = A_c - P_{Z_c} A_c$ and $P_{Z_c} = Z_c(Z_c^\top Z_c)^{-1}Z_c^\top$. Here, $A_c\in\R^{576\times 1}$ and $Z_c\in\R^{576\times 284}$ denote, respectively, the covariate matrices for the differenced abortion rates $a_{cit} - a_{cit-1}$ and the 284 control variables $z_{cit}$ for $i = 1, \dots, 48$, $t = 2, \dots, 13$, given crime type $c$.  

In Table~\ref{tab:levitt_CI}, we present the confidence intervals for $\alpha_c$ for every crime type $c$.
The middle column (``no selection") presents results calculated in the original specification, and the ``double selection" column presents the interval after a double selection on $Z_c$ as in \cite{Belloni2013}. 
This double selection procedure first applies Lasso selection by regressing $a_{cit}-a_{cit-1}$ on $z_{cit}$, and then applies another Lasso selection by regressing $y_{cit} - y_{cit-1}$ on $z_{cit}$. 
 The final set of control variables is then comprised of the two sets of selected control variables from the double selection procedure, and the eight original control variables in~\cite{Donohue2001}. 

\renewcommand{\arraystretch}{1.3}
\begin{table}[t!]
    \centering
    \begin{tabular}{c|cc}
    \hline
    
    Method & \multicolumn{2}{c}{$95\%$ CI for $\alpha_c$} \\
    \hline
    \hline
    Invariance-based method & no selection & double selection  \\
    \hline
    violent & \big[-0.773, 1.000\big] & \big[-0.310, 0.135\big]\\
    property & \big[-0.508, 0.005\big] & \big[-0.111, 0.046\big] \\
    murder & \big[-1.824, 5.349\big] & \big[-1.011, 0.880\big] \\
    \hline
    \hline
    \cite{Belloni2013} & no selection &  double selection  \\
    \hline
    violent & \big[-1.395, 1.422\big] & \big[-0.290, 0.126\big] \\
    property & \big[-0.635, 0.245\big] & \big[-0.143, 0.081\big] \\
    murder & \big[-3.141, 7.826\big] & \big[-0.324, 0.460\big] \\
    \hline
    \end{tabular}
    \bigskip
    \caption{Confidence intervals of $\alpha_c$ for $c\in\{\text{violent}, \text{property}, \text{murder}\}$. For \cite{Belloni2013}, we compute confidence intervals through [Effect-$z_{\alpha/2}$Std.\ Err., Effect+$z_{\alpha/2}$Std.\ Err.] from ``All controls" and ``Post-double-selection+" of Table 2 in \cite{Belloni2013}, where $z_{\alpha}$ denotes the $\alpha$ upper quantile of standard normal distribution. }
    \label{tab:levitt_CI}
\end{table}

For comparison, the table also reports the 95\% confidence intervals from \cite{Belloni2013} computed as $[\widehat{\alpha}-z_{0.025}\widehat{\sigma}, ~ \widehat{\alpha}+z_{0.025}\widehat{\sigma}]~$, 
where $\widehat{\alpha}$ and $\widehat{\sigma}$ are the abortion effect and its standard error based on an OLS regression of the crime rate over the abortion rate and corresponding set of control variables~($z_{\alpha}$ denotes the $\alpha$ upper quantile of the standard normal distribution).\footnote{From Table 2 in \cite{Belloni2013}, we have respectively, that $(\widehat{\alpha}, \widehat{\sigma}) = (0.014,0.719)$, $(-0.195,0.225)$, $(2.343,2.798)$ for violent, property, and murder crime under ``no selection'', and $(\widehat{\alpha}, \widehat{\sigma}) = (-0.082, 0.106)$, $(-0.031, 0.057)$, $(-0.068, 0.200)$ under violent, property, and murder crime for ``double selection''.}

We make the following observations. First, all our confidence intervals in Table~\ref{tab:levitt_CI} contain zero, indicating that we cannot reject $\alpha_c = 0$ in all cases. 
This finding is robust even under double selection on the control variables. 
Second, our confidence intervals are qualitatively similar to those obtained from \cite{Belloni2013}. 
Interestingly, the \name-based intervals are generally shorter than 
the Lasso-based intervals under ``no selection", but the intervals look more similar under ``double selection". 
This is likely due to the high multicollinearity in the 
``no selection" model,\footnote{The condition numbers of the covariate matrix of all control variables are orders of magnitude larger than those after double selection. } which translates into larger confidence intervals for the Lasso-based intervals. As we repeatedly emphasized throughout the paper, 
\name-based inference appears to be robust to such  multicollinearity.

\section{Concluding Remarks}\label{sec:discuss}
This paper opens up some problems for future work. 
First, relating to theory, it would be interesting to analyze the power of \name-based tests when $p>n$, and thus fully explain the empirical results of Section~\ref{sec:simu}. 
Another direction could be to extend \name-based tests to nonlinear regression models, and analyze their 
finite-sample properties. For instance, while our current procedures can readily test the global null hypothesis in generalized linear models, analyzing this extension would pose significant theoretical challenges.
\bibliographystyle{natbib}
\bibliography{paper-ref}

\begin{thebibliography}{}

\bibitem[Akritas and Arnold(2000)Akritas and Arnold]{Akritas2000}
Akritas, M. and Arnold, S. (2000).
\newblock Asymptotics for analysis of variance when the number of levels is
  large.
\newblock {\em J. Amer. Statist. Assoc.}, {\bf 95}(449), 212--226.

\bibitem[Ashburner {\em et~al.}(2000)Ashburner, Ball, Blake, Botstein, Butler,
  Cherry, Davis, Dolinski, Dwight, Eppig, Harris, Hill, Issel-Tarver,
  Kasarskis, Lewis, Matese, Richardson, Ringwald, Rubin, and
  Sherlock]{GeneOntology}
Ashburner, M., Ball, C.~A., Blake, J.~A., Botstein, D., Butler, H., Cherry,
  J.~M., Davis, A.~P., Dolinski, K., Dwight, S.~S., Eppig, J.~T., Harris,
  M.~A., Hill, D.~P., Issel-Tarver, L., Kasarskis, A., Lewis, S., Matese,
  J.~C., Richardson, J.~E., Ringwald, M., Rubin, G.~M., and Sherlock, G.
  (2000).
\newblock Gene ontology: tool for the unification of biology.
\newblock {\em Nature Genetics\/}, {\bf 25}(1), 25--29.

\bibitem[Athey {\em et~al.}(2018)Athey, Eckles, and Imbens]{athey2018exact}
Athey, S., Eckles, D., and Imbens, G.~W. (2018).
\newblock Exact p-values for network interference.
\newblock {\em Journal of the American Statistical Association\/}, {\bf
  113}(521), 230--240.

\bibitem[Basse {\em et~al.}(2019)Basse, Feller, and
  Toulis]{basse2019randomization}
Basse, G.~W., Feller, A., and Toulis, P. (2019).
\newblock Randomization tests of causal effects under interference.
\newblock {\em Biometrika\/}, {\bf 106}(2), 487--494.

\bibitem[Belloni {\em et~al.}(2013)Belloni, Chernozhukov, and
  Hansen]{Belloni2013}
Belloni, A., Chernozhukov, V., and Hansen, C. (2013).
\newblock {Inference on Treatment Effects after Selection among
  High-Dimensional Controls†}.
\newblock {\em The Review of Economic Studies\/}, {\bf 81}(2), 608--650.

\bibitem[Boos and Brownie(1995)Boos and Brownie]{Boos1995}
Boos, D.~D. and Brownie, C. (1995).
\newblock A{NOVA} and rank tests when the number of treatments is large.
\newblock {\em Statist. Probab. Lett.}, {\bf 23}(2), 183--191.

\bibitem[Boucheron {\em et~al.}(2013)Boucheron, Lugosi, and
  Massart]{boucheron_lugosi_massart}
Boucheron, S., Lugosi, G., and Massart, P. (2013).
\newblock {\em Concentration inequalities\/}.
\newblock Oxford University Press, Oxford.
\newblock A nonasymptotic theory of independence, With a foreword by Michel
  Ledoux.

\bibitem[Cai {\em et~al.}(2022)Cai, Lei, and Roeder]{cai2022}
Cai, Z., Lei, J., and Roeder, K. (2022).
\newblock Model-free prediction test with application to genomics data.
\newblock {\em Proceedings of the National Academy of Sciences\/}, {\bf
  119}(34), e2205518119.

\bibitem[Calhoun(2011)Calhoun]{Calhoun2011}
Calhoun, G. (2011).
\newblock Hypothesis testing in linear regression when {$k/n$} is large.
\newblock {\em J. Econometrics\/}, {\bf 165}(2), 163--174.

\bibitem[Campi and Weyer(2005)Campi and Weyer]{Campi2005}
Campi, M.~C. and Weyer, E. (2005).
\newblock Guaranteed non-asymptotic confidence regions in system
  identification.
\newblock {\em Automatica J. IFAC\/}, {\bf 41}(10), 1751--1764.

\bibitem[Campi {\em et~al.}(2009)Campi, Ko, and Weyer]{Campi2009}
Campi, M.~C., Ko, S., and Weyer, E. (2009).
\newblock Non-asymptotic confidence regions for model parameters in the
  presence of unmodelled dynamics.
\newblock {\em Automatica J. IFAC\/}, {\bf 45}(10), 2175--2186.

\bibitem[Canay {\em et~al.}(2017)Canay, Romano, and
  Shaikh]{canay2017randomization}
Canay, I.~A., Romano, J.~P., and Shaikh, A.~M. (2017).
\newblock Randomization tests under an approximate symmetry assumption.
\newblock {\em Econometrica\/}, {\bf 85}(3), 1013--1030.

\bibitem[Candes {\em et~al.}(2018)Candes, Fan, Janson, and
  Lv]{candes2018panning}
Candes, E., Fan, Y., Janson, L., and Lv, J. (2018).
\newblock Panning for gold:‘model-x’knockoffs for high dimensional
  controlled variable selection.
\newblock {\em Journal of the Royal Statistical Society Series B: Statistical
  Methodology\/}, {\bf 80}(3), 551--577.

\bibitem[Carpentier {\em et~al.}(2018)Carpentier, Collier, Comminges, Tsybakov,
  and Wang]{carpentier2018}
Carpentier, A., Collier, O., Comminges, L., Tsybakov, A.~B., and Wang, Y.
  (2018).
\newblock Minimax rate of testing in sparse linear regression.

\bibitem[Cheng {\em et~al.}(2022)Cheng, Duchi, and Kuditipudi]{Duchi2022}
Cheng, C., Duchi, J., and Kuditipudi, R. (2022).
\newblock Memorize to generalize: on the necessity of interpolation in high
  dimensional linear regression.
\newblock In P.-L. Loh and M.~Raginsky, editors, {\em Proceedings of Thirty
  Fifth Conference on Learning Theory\/}, volume 178 of {\em Proceedings of
  Machine Learning Research\/}, pages 5528--5560. PMLR.

\bibitem[Cui {\em et~al.}(2018)Cui, Guo, and Zhong]{Cui2018}
Cui, H., Guo, W., and Zhong, W. (2018).
\newblock Test for high-dimensional regression coefficients using refitted
  cross-validation variance estimation.
\newblock {\em Ann. Statist.}, {\bf 46}(3), 958--988.

\bibitem[David(2008)David]{david2008}
David, H.~A. (2008).
\newblock The beginnings of randomization tests.
\newblock {\em The American Statistician\/}, {\bf 62}(1), 70--72.

\bibitem[Dezeure {\em et~al.}(2014)Dezeure, Buhlmann, Meier, and
  Meinshausen]{Dezeure2014HighDimensionalIC}
Dezeure, R., Buhlmann, P., Meier, L., and Meinshausen, N. (2014).
\newblock High-dimensional inference: Confidence intervals, $p$-values and
  r-software hdi.
\newblock {\em Statistical Science\/}, {\bf 30}, 533--558.

\bibitem[Ding(2017)Ding]{ding2017paradox}
Ding, P. (2017).
\newblock A paradox from randomization-based causal inference.
\newblock {\em Statistical science\/}, pages 331--345.

\bibitem[Dobriban(2022)Dobriban]{dobriban2022}
Dobriban, E. (2022).
\newblock {Consistency of invariance-based randomization tests}.
\newblock {\em The Annals of Statistics\/}, {\bf 50}(4), 2443 -- 2466.

\bibitem[Donohue and Levitt(2001)Donohue and Levitt]{Donohue2001}
Donohue, John~J., I. and Levitt, S.~D. (2001).
\newblock {The Impact of Legalized Abortion on Crime*}.
\newblock {\em The Quarterly Journal of Economics\/}, {\bf 116}(2), 379--420.

\bibitem[Edgington and Onghena(2007)Edgington and
  Onghena]{edgington2007randomization}
Edgington, E. and Onghena, P. (2007).
\newblock {\em Randomization tests\/}.
\newblock CRC press.

\bibitem[Fan and Lv(2008)Fan and Lv]{Fan2008}
Fan, J. and Lv, J. (2008).
\newblock Sure independence screening for ultrahigh dimensional feature space.
\newblock {\em Journal of the Royal Statistical Society: Series B (Statistical
  Methodology)\/}, {\bf 70}(5), 849--911.

\bibitem[Fisher(1935)Fisher]{fisher1935}
Fisher, R.~A. (1935).
\newblock {\em {The Design of Experiments}\/}.
\newblock Oliver and Boyd, Edinburgh.

\bibitem[Foote and Goetz(2008)Foote and Goetz]{foote2008}
Foote, C. and Goetz, C. (2008).
\newblock The impact of legalized abortion on crime: Comment.
\newblock {\em The Quarterly Journal of Economics\/}, {\bf 123}(1), 407--423.

\bibitem[Freedman and Lane(1983)Freedman and Lane]{Freedman1983}
Freedman, D. and Lane, D. (1983).
\newblock A nonstochastic interpretation of reported significance levels.
\newblock {\em Journal of Business \& Economic Statistics\/}, {\bf 1}(4),
  292--298.

\bibitem[Gerber and Green(2012)Gerber and Green]{gerber2012field}
Gerber, A.~S. and Green, D.~P. (2012).
\newblock Field experiments: Design, analysis, and interpretation, new york: W.
  w.

\bibitem[Imbens and Rubin(2015)Imbens and Rubin]{imbens2015causal}
Imbens, G.~W. and Rubin, D.~B. (2015).
\newblock {\em Causal inference in statistics, social, and biomedical
  sciences\/}.
\newblock Cambridge University Press.

\bibitem[Ingster(1995)Ingster]{Ingster1995}
Ingster, Y.~I. (1995).
\newblock Minimax testing of hypotheses on the distribution density for
  ellipsoids in \$l\_p \$.
\newblock {\em Theory of Probability \& Its Applications\/}, {\bf 39}(3),
  417--436.

\bibitem[Ingster and Suslina(2003)Ingster and Suslina]{Ingster2003}
Ingster, Y.~I. and Suslina, I.~A. (2003).
\newblock {\em Nonparametric goodness-of-fit testing under {G}aussian
  models\/}, volume 169 of {\em Lecture Notes in Statistics\/}.
\newblock Springer-Verlag, New York.

\bibitem[Ingster {\em et~al.}(2010)Ingster, Tsybakov, and
  Verzelen]{Ingster2010}
Ingster, Y.~I., Tsybakov, A.~B., and Verzelen, N. (2010).
\newblock Detection boundary in sparse regression.
\newblock {\em Electron. J. Stat.}, {\bf 4}, 1476--1526.

\bibitem[Javanmard and Montanari(2014)Javanmard and Montanari]{Javanmard2014}
Javanmard, A. and Montanari, A. (2014).
\newblock Confidence intervals and hypothesis testing for high-dimensional
  regression.
\newblock {\em J. Mach. Learn. Res.}, {\bf 15}(1), 2869–2909.

\bibitem[Lehmann and Romano(2005)Lehmann and Romano]{Lehmann2005}
Lehmann, E.~L. and Romano, J.~P. (2005).
\newblock {\em Testing statistical hypotheses\/}.
\newblock Springer Texts in Statistics. Springer, New York, third edition.

\bibitem[Lei and Bickel(2020)Lei and Bickel]{Lei2020}
Lei, L. and Bickel, P.~J. (2020).
\newblock {An assumption-free exact test for fixed-design linear models with
  exchangeable errors}.
\newblock {\em Biometrika\/}.

\bibitem[Li {\em et~al.}(2020)Li, Kim, and Wei]{Li2020}
Li, Y., Kim, I., and Wei, Y. (2020).
\newblock Randomized tests for high-dimensional regression: A more efficient
  and powerful solution.
\newblock In H.~Larochelle, M.~Ranzato, R.~Hadsell, M.~Balcan, and H.~Lin,
  editors, {\em Advances in Neural Information Processing Systems\/},
  volume~33, pages 4721--4732. Curran Associates, Inc.

\bibitem[Lkhagvadorj {\em et~al.}(2009)Lkhagvadorj, Qu, Cai, Couture, Barb,
  Hausman, Nettleton, Anderson, Dekkers, and Tuggle]{Lkhagvadorj2009}
Lkhagvadorj, S., Qu, L., Cai, W., Couture, O.~P., Barb, C.~R., Hausman, G.~J.,
  Nettleton, D., Anderson, L.~L., Dekkers, J. C.~M., and Tuggle, C.~K. (2009).
\newblock Microarray gene expression profiles of fasting induced changes in
  liver and adipose tissues of pigs expressing the melanocortin-4 receptor
  d298n variant.
\newblock {\em Physiological Genomics\/}, {\bf 38}(1), 98--111.

\bibitem[Lopes {\em et~al.}(2011)Lopes, Jacob, and Wainwright]{Lopes2011}
Lopes, M., Jacob, L., and Wainwright, M.~J. (2011).
\newblock A more powerful two-sample test in high dimensions using random
  projection.
\newblock In J.~Shawe-Taylor, R.~Zemel, P.~Bartlett, F.~Pereira, and
  K.~Weinberger, editors, {\em Advances in Neural Information Processing
  Systems\/}, volume~24. Curran Associates, Inc.

\bibitem[Ma {\em et~al.}(2020)Ma, Cai, and Li]{Ma2020}
Ma, R., Cai, T.~T., and Li, H. (2020).
\newblock Global and simultaneous hypothesis testing for high-dimensional
  logistic regression models.
\newblock {\em Journal of the American Statistical Association\/}, {\bf 0}(0),
  1--15.

\bibitem[Manly(1997)Manly]{Manly1997}
Manly, B. F.~J. (1997).
\newblock {\em Randomization, bootstrap and {M}onte {C}arlo methods in
  biology\/}.
\newblock Texts in Statistical Science Series. Chapman \& Hall, London, second
  edition.

\bibitem[Pitman(1937)Pitman]{pitman1937}
Pitman, E. J.~G. (1937).
\newblock Significance tests which may be applied to samples from any
  populations.
\newblock {\em Supplement to the Journal of the Royal Statistical Society\/},
  {\bf 4}(1), 119--130.

\bibitem[Toulis(2019)Toulis]{Toulis2019}
Toulis, P. (2019).
\newblock Invariant inference via residual randomization.

\bibitem[Tropp(2015)Tropp]{Tropp2015}
Tropp, J.~A. (2015).
\newblock An introduction to matrix concentration inequalities.

\bibitem[van~de Geer {\em et~al.}(2014)van~de Geer, B{\"u}hlmann, Ritov, and
  Dezeure]{vandeGeer2014}
van~de Geer, S., B{\"u}hlmann, P., Ritov, Y., and Dezeure, R. (2014).
\newblock {On asymptotically optimal confidence regions and tests for
  high-dimensional models}.
\newblock {\em The Annals of Statistics\/}, {\bf 42}(3), 1166 -- 1202.

\bibitem[van Wieringen(2015)van Wieringen]{lecture_ridge}
van Wieringen, W.~N. (2015).
\newblock Lecture notes on ridge regression.

\bibitem[Wald(1950)Wald]{wald1950}
Wald, A. (1950).
\newblock {\em Statistical Decision Functions\/}.
\newblock Wiley, New York.

\bibitem[Wang and Cui(2015)Wang and Cui]{wang2015}
Wang, S. and Cui, H. (2015).
\newblock A new test for part of high dimensional regression coefficients.
\newblock {\em Journal of Multivariate Analysis\/}, {\bf 137}, 187--203.

\bibitem[Wang {\em et~al.}(2021)Wang, Lee, Toulis, and Kolar]{Wang2021}
Wang, Y.~S., Lee, S.~K., Toulis, P., and Kolar, M. (2021).
\newblock Robust inference for high-dimensional linear models via residual
  randomization.
\newblock In {\em Proceedings of the 38th International Conference on Machine
  Learning, {ICML} 2021, 18-24 July 2021, Virtual Event\/}, volume 139 of {\em
  Proceedings of Machine Learning Research\/}, pages 10805--10815. {PMLR}.

\bibitem[Wen {\em et~al.}(2022)Wen, Wang, and Wang]{Wen2022}
Wen, K., Wang, T., and Wang, Y. (2022).
\newblock Residual permutation test for high-dimensional regression coefficient
  testing.

\bibitem[Zhang and Zhang(2014)Zhang and Zhang]{Zhang2014}
Zhang, C.-H. and Zhang, S.~S. (2014).
\newblock {Confidence intervals for low dimensional parameters in high
  dimensional linear models}.
\newblock {\em Journal of the Royal Statistical Society Series B\/}, {\bf
  76}(1), 217--242.

\bibitem[Zhao and Yu(2006)Zhao and Yu]{Zhao2006}
Zhao, P. and Yu, B. (2006).
\newblock On model selection consistency of lasso.
\newblock {\em Journal of Machine Learning Research\/}, {\bf 7}(90),
  2541--2563.

\bibitem[Zhong and Chen(2011)Zhong and Chen]{Zhong2011}
Zhong, P.-S. and Chen, S.~X. (2011).
\newblock Tests for high-dimensional regression coefficients with factorial
  designs.
\newblock {\em Journal of the American Statistical Association\/}, {\bf
  106}(493), 260--274.

\end{thebibliography}

\newpage

\appendix

\begin{center}
    \Huge\bf
    Appendix
\end{center}

\small

\section{Proofs of Theorems~\ref{thm:nonasymp_power} and \ref{thm:nonasymp_power_subgauss}}\label{sec:nonasymp_powerproof}
\subsection{Notation and Preliminaries}
Here, we summarize the notation and preliminaries used in the proofs below. In the study of global null ($H_0$), $\sigma_{\min}$ and $\sigma_{\max}$ denote the minimum and the maximum singular value of $X$, and $x_* = \max_{ij} |X_{ij}|$.
For intuition, we use $\kappa = \sigma_{\max}/\sigma_{\min}$ and $s = \sigma_{\min}^2/n$. 
In addition, define $Q_X = (X^\top X + \lambda I)^{-1}$. Using the singular value decomposition (SVD) of $X$, one can easily obtain that $\|Q_X\| = 1/(\sigma_{\min}^2 + \lambda)$; moreover, $Q_X$ is well-defined as Assumption~\ref{asmp:regularity} guarantees $\sigma_{\min}>0$. Last, we denote by $\|A\|$ and $\|A\|_F$ the spectral norm and the Frobenius norm of a matrix $A$, respectively. For $A, B\in\R^{n\times n}$, we write $A\preceq B$ if $B - A$ is positive semidefinite.

Recall that we utilize the $p$-value approximated by $G_r\stackrel{iid}{\sim}\Unif(\cG)$, where $\cG = \cG^{\mathrm{s}}$ consists of all $n\times n$ matrices with $\pm 1$ on the diagonal. As a direct corollary, for $G\sim\Unif(\cG)$, its diagonal elements $(g_i)_{i=1}^n$ follow Rademacher distribution and they are mutually independent; this fact will be invoked in our proofs.

We use the following identities in our proofs: for positive random variables $Y$, $Z$ and positive constants $t$, $s$, we have
\begin{align}
    \P(Y + Z \ge t)&\le \P(Y \ge t/2) + \P(Z \ge t/2)\;,\label{eq:id_summation}\\
    \P(YZ \ge t)&\le \P(Y \ge t/s) + \P(Z \ge s)\;.\label{eq:id_product}
\end{align}

We use $\P$ and $\P_X$ to denote the probability analyzed under random design and fixed design, respectively. Similarly, we denote by $\E$ and $\E_X$ the expectation under random design and fixed design, respectively. For simplicity, we omit the subscript $X$ whenever the random variable does not depend on $X$.

\subsection{Main Proofs}
We state a more detailed version of Theorem~\ref{thm:nonasymp_power} and give its proof. 
\begin{customthm}{2}
Suppose Assumptions \ref{asm:err1} and \ref{asmp:regularity} hold, $\beta\neq 0$ in \eqref{eq:lm} and $p<n$. Then we have
\begin{equation}\label{eq:nonasymp_power}
\begin{aligned}
    \E_X(1 - \psi_\alpha) &= O(\widetilde{A}_n(X)) + O\Big(\frac{\widetilde{B}_n(X)}{\|\beta\|^2}\Big)~,\\
    \widetilde{A}_n(X)&= \frac{p^2 \sigma_{\max}^4 x_*^2 }{\sigma_{\min}^2(\sigma_{\min}^2+\lambda)^2} + \frac{\lambda^2}{(\sigma_{\min}^2+\lambda)^2}~,\\
    \widetilde{B}_n(X) &= \frac{p \sigma_{\max}^4 \sigma_*^2}{\sigma_{\min}^2(\sigma_{\min}^2+\lambda)^2} + \frac{p}{\sigma_{\min}^2+\lambda}~.
\end{aligned}
\end{equation}
If we choose a ridge penalty parameter $\lambda \le \sigma_{\min}^2$, the upper bound reduces to
\begin{align*}
    \E_X(1 - \psi_\alpha) &= O(A_n(X)) + O\Big(\frac{B_n(X)}{\|\beta\|^2}\Big)~, \\
    A_n(X)&= \frac{p^2 \sigma_{\max}^4 x_*^2 }{\sigma_{\min}^6} + \frac{1}{\sigma_{\min}^2} = \frac{p^2 \kappa^4 x_*^2 }{sn} + \frac{1}{sn}~,\\
    B_n(X)&= \frac{p \sigma_{\max}^4 \sigma_*^2}{\sigma_{\min}^6} + \frac{p}{\sigma_{\min}^2} =  \frac{p \kappa^4 \sigma_*^2}{sn} + \frac{p}{sn}~.
\end{align*}
\end{customthm}
\begin{proof}
First, we upper bound the Type II error using a union bound, that is, 
\begin{align*}
    \E_X(1 - {\psi}_\alpha) = \P_X({p}(X, y)>\alpha) &= \P_X\left(\frac{1}{R+1}\left( 1+ \sum_{r = 1}^{R} \mathbb{I}\{t(G_r y, X)>t(y, X)\}\right) > \alpha\right)\\
    &\le \P_X\left(\exists r, t(G_r y, X)>t(y, X)\right)\\
    &\le \sum_{r = 1}^{R} \P_X\left(t(G_r y, X)>t(y, X)\right) \\
    &\stackrel{\text{(i)}}{\le} R\cdot\P_X\left(t(G_1 y, X)>t(y, X)\right)\\
    &{\le} R\Bigl(\P_X\bigl(t(G y, X)\ge c^2\bigr) + \P_X\bigl(t(y, X)\le c^2\bigr)\Bigr)~, \quad G\sim\Unif(\cG)~,
\end{align*}
where (i) follows from the fact that $(G_r)_{r=1}^R$ are i.i.d. In addition, $c^2$ in the last line is a constant to be specified. Next, we analyze $t(Gy, X)$ and $t(y, X)$ separately in order to upper bound two probabilities $\P_X\bigl(t(G y, X)\ge c^2\bigr)$ and $\P_X\bigl(t(y, X)\le c^2\bigr)$ above.

\noindent\underline{Analyzing $t(G y, X)$.}
We first derive a concentration inequality for $t(Gy, X)$. Notice that
\begin{align*}
    t(Gy, X) &= \|X(X^\top X + \lambda I)^{-1}X^\top G y\|^2 \\
    &\le \sigma_{\max}^2\|(X^\top X + \lambda I)^{-1}X^\top G y\|^2 \\
    &= \sigma_{\max}^2\| Q_X X^\top G X \beta_{np} + Q_X X^\top G \eps\|^2 \\
    &\le 2\sigma_{\max}^2\left(\|Q_X X^\top G X\|^2 \|\beta\|^2 + \|Q_X X^\top G \eps\|^2\right)\;,
\end{align*}
where $Q_X = (X^\top X + \lambda I )^{-1}$. Hence, 
\begin{align*}
    \P_X\left(t(Gy, X)\ge c^2 \right) &\le 
    \P_X\left(2\sigma_{\max}^2\left(\|Q_X X^\top G X\|^2 \|\beta\|^2 + \|Q_X X^\top G \eps\|^2\right) \ge c^2 \right)\\
    &\stackrel{\text{(i)}}{\le} \P_X\left(\|Q_X X^\top G X\|^2 \|\beta\|^2 \ge \frac{c^2}{4\sigma_{\max}^2}\right) + \P_X\left(\|Q_X X^\top G \eps\|^2 \ge \frac{c^2}{4\sigma_{\max}^2}\right)\\
    &\stackrel{\text{(ii)}}{\le} \P_X\left(\|Q_X\|\| X^\top G X\| \|\beta\| \ge \frac{c}{2\sigma_{\max}}\right) \\
    &+ \P_X\left(\|Q_X\| \|X^\top G \eps\| \ge \frac{c}{2\sigma_{\max}}\right)\\
    &= \P_X\left(\left\|\frac{1}{n} X^\top G X\right\| \ge \frac{c}{2 \sigma_{\max} n\|\beta\|\|Q_X\|}\right) + \P_X\left( \left\|\frac{1}{n} X^\top G \eps\right\| \ge \frac{c}{2\sigma_{\max} n\|Q_X\| }\right)~.
\end{align*}
In the derivation above, (i) follows from \eqref{eq:id_summation} and (ii) follows from the submultiplicativity of the matrix norm.
We state the following concentration inequalities for $\frac{1}{n}X^\top G X$ and $\frac{1}{n}X^\top G \eps$ above; their proofs can be found in Section~\ref{sec:lemmas}. 
\begin{lemma}\label{lem:cov}
    Suppose that $G\sim\Unif(\cG^{\mathrm{s}})$. We have 
    \begin{equation*}
        \P_X\left( \left\|\frac{1}{n} X^\top G X \right\| \ge t \right)\le\frac{p^2 x_*^2 \sigma_{\max}^2}{n^2 t^2}
    \end{equation*}
    for any $t>0$. 
\end{lemma}
\begin{lemma}\label{lem:err}
    Suppose that Assumption \ref{asm:err1} holds and $G\sim\Unif(\cG^{\mathrm{s}})$. We have 
    \begin{align*}
        \P_X\left( \left\|\frac{1}{n} X^\top G \eps \right\| \ge t \right)\le \frac{p\sigma_{\max}^2 \sigma_*^2 }{n^2t^2}
    \end{align*}
    for any $t>0$. 
\end{lemma}
Based on Lemmas \ref{lem:cov} and \ref{lem:err}, we obtain 
\begin{align*}
   \P_X\left(t(Gy, X)\ge c^2 \right)&\le \frac{4 \sigma_{\max}^2 n^2 \|\beta\|^2 \|Q_X\|^2}{c^2}\cdot \frac{p^2 \sigma_{\max}^2 x_*^2}{n^2} + \frac{4 \sigma_{\max}^2 n^2 \|Q_X\|^2}{c^2} \cdot \frac{p \sigma_{\max}^2 \sigma_*^2}{n^2}\\
    & = \frac{4p\sigma_{\max}^4 \|Q_X\|^2}{c^2} (p x_*^2 \|\beta\|^2 + \sigma_*^2)\;.
\end{align*}

\noindent\underline{Analyzing $t(y, X)$:}
Next, we develop a different concentration inequality for $t(y, X)$. Notice that
\begin{align}
    \P_X\left(t(y, X)\le c^2\right) &\le \P_X\left( \sigma_{\min}^2 \|\widehat{\beta}\|^2 \le c^2\right)\nonumber \\
    &\stackrel{\text{(i)}}{\le} \P_X\left(\sigma_{\min}^2(\|\beta\| - \|\widehat{\beta} - \beta\|)^2\le c^2\right)\,(\because \text{triangle inequality})\nonumber \\
    &= \P_X\left(\left|\|\beta\| - \|\widehat{\beta} - \beta\|\right|\le c/\sigma_{\min}\right)\nonumber\\
    &\le \P_X\left(\|\beta\| - \|\widehat{\beta} - \beta\| \le c/\sigma_{\min}, \|\beta\| \ge \|\widehat{\beta} - \beta\|\right)\label{eq:ridge1} \\
    &+ \P_X\left(\|\widehat{\beta} - \beta\| - \|\beta\| \le c/\sigma_{\min}, \|\widehat{\beta} - \beta\| \ge \|\beta\| \right)\label{eq:ridge2}
\end{align}
where (i) is due to triangle inequality. To further bound \eqref{eq:ridge1}, \eqref{eq:ridge2} on the right hand side, we use
\begin{align*}
    \eqref{eq:ridge1} &\le \P_X\left(\|\beta\| - \|\widehat{\beta} - \beta\| \le c\right) = \P_X\left(\|\widehat{\beta} - \beta\| \ge \|\beta\| - c/\sigma_{\min}\right)\;,\\
    \eqref{eq:ridge2} &\le \P_X\left(\|\widehat{\beta} - \beta\| \ge \|\beta\| \right)\;.
\end{align*}
Hence we have
\begin{align*}
    \P_X\left(t(y, X)\le c^2\right) &\le \P_X\left(\|\widehat{\beta} - \beta\| \ge \|\beta\| - c/\sigma_{\min}\right) + \P_X\left(\|\widehat{\beta} - \beta\| \ge \|\beta\|\right)\\
    &\le 2\P_X\left(\|\widehat{\beta} - \beta\| \ge \|\beta\| - c/\sigma_{\min}\right)\;. 
\end{align*}
We state the following concentration inequality for the term $\|\widehat{\beta} - \beta\|$; its proof can be found in Section~\ref{sec:lemmas}. 
\begin{lemma}\label{lem:ridge}
Under Assumptions \ref{asm:err1} and \ref{asmp:regularity}, we have
\begin{equation*}
    \P_X\left(\|\widehat{\beta} - \beta\|\ge t\right) \le \frac{1}{ t^2} \left( \lambda^2 \|Q_X\|^2 \|\beta\|^2 + p\|Q_X\| \right)\;.
\end{equation*}
for any $t>0$. 
\end{lemma}

Therefore, as long as we properly choose $c<\sigma_{\min}\|\beta\|$, we can apply Lemma \ref{lem:ridge} and obtain 
\begin{equation*}
    \P_X\left(t(y, X)\le c^2\right)\le \frac{2}{(\|\beta\| - c/\sigma_{\min})^2} \left( \lambda^2 \|Q_X\|^2 \|\beta\|^2 + p \|Q_X\| \right)\;. 
\end{equation*}

\noindent\underline{Deriving the final result. }
Combining our upper bounds for $\P_X(t(G y, X)\ge c^2)$ and $\P_X(t(y, X)\le c^2)$, we obtain
\begin{align*}
    \P_X(t(G y, X)\ge c^2) + \P_X(t(y, X)\le c^2) &\le \frac{4p\sigma_{\max}^4 \|Q_X\|^2}{c^2} (p x_*^2 \|\beta_{np}\|^2 + \sigma_*^2) \\
    &+ \frac{2}{(\|\beta_{np}\| - c/\sigma_{\min})^2} \left( \lambda^2 \|Q_X\|^2 \|\beta_{np}\|^2 + p \|Q_X\| \right)\;. 
\end{align*}
By choosing $c = \sigma_{\min}\|\beta_{np}\|/2$, the right hand side above reduces to 
\begin{equation*}
    \frac{16 p^2 \sigma_{\max}^4 x_*^2 \|Q_X\|^2}{\sigma_{\min}^2} + \frac{16p \sigma_{\max}^4 \sigma_*^2\|Q_X\|^2}{\sigma_{\min}^2\|\beta\|^2} + 8 \lambda^2 \|Q_X\|^2 + \frac{8 p \|Q_X\| }{\|\beta\|^2}~.
\end{equation*}
By SVD, we have $\|Q_X\| = 1/(\sigma_{\min}^2 + \lambda)$ and thus obtain Equation~\eqref{eq:nonasymp_power}. Moreover, if $\lambda\le\sigma_{\min}^2$, \eqref{eq:nonasymp_power} further reduces to
\begin{equation*}
    O(A_n(X)) + O\Big(\frac{B_n(X)}{\|\beta\|^2}\Big)~.
\end{equation*}
\end{proof}
By assuming sub-Gaussian errors, we are able to derive an improved type II error bound.
\begin{customthm}{3}
Suppose Assumptions \ref{asm:err1} and \ref{asmp:regularity} hold, and that $\beta\neq 0$ in \eqref{eq:lm} and $p<n$. Further suppose that $(\eps_i)_{i=1}^n$ are independent sub-Gaussian random variables, such that $\E \exp(t \eps_i )\le \exp(t^2 v/2)$ for some $v>0$. If we choose the ridge penalty parameter such that $3\lambda \le \sigma_{\min}^2$, then
\begin{equation*}
    \E_X(1 - {\psi}_\alpha) = O(p) \exp\left\{- C_n(X)\right\} + O(n+p) \exp\left\{- D_n(X) f(\|\beta\|) \right\}~,
\end{equation*}
where $f(x) = \min\{x, x^{2/3}\}$, $C_n(X) = 0.03 s^2 n/\kappa^2 p^2 x_*^4$ and 
$
    D_n(X) = 0.09\min\left\{ \sqrt{\frac{s^2 n}{p v x_*^2\kappa^2}}, \sqrt[3]{\frac{s^2 n^2}{p v x_*^2\kappa^2}} \right\}~. 
$
\end{customthm}
\begin{proof}
From the proof of Theorem \ref{thm:nonasymp_power}, we obtain 
\begin{align*}
    \E_X(1 - {\psi}_\alpha) &\le R \cdot\P_X\left(t(Gy, X)>t(y, X)\right)\\
    &\le R\Bigl(\P_X\bigl(t(Gy, X)\ge c^2\bigr) + \P_X\bigl(t(y, X)\le c^2\bigr)\Bigr)\;. 
\end{align*}
Moreover, we have 
\begin{align*}
    \P_X\bigl(t(Gy, X)\ge c^2\bigr) &\le \P_X\left(\left\|\frac{1}{n} X^\top G X\right\| \ge \frac{c}{2 \sigma_{\max} n\|\beta\|\|Q_X\|}\right) + \P_X\left( \left\|\frac{1}{n} X^\top G \eps\right\| \ge \frac{c}{2\sigma_{\max} n\|Q_X\| }\right)\;, \\
    \P_X\left(t(y, X)\le c^2\right) &\le 2\P_X\left(\|\widehat{\beta} - \beta\| \ge \|\beta\| - c/\sigma_{\min}\right)\;. 
\end{align*}
Using a norm bound on matrix Rademacher series \citep{Tropp2015}, we have the following concentration inequality on the term $\frac{1}{n}X^\top G X$; see Section \ref{sec:lemmas} for its proof.
\begin{lemma}\label{lem:cov2}
Suppose that $G\sim\Unif(\cG^{\mathrm{s}})$. We have 
\begin{align*}
    \P_X\left( \left\|\frac{1}{n} X^\top G X \right\| \ge t \right)&\le p \exp\left( -\frac{n t^2}{2 p^2 x_*^4}\right)
\end{align*}
for any $t>0$. 
\end{lemma}
Under the sub-Gaussian assumption, we have the following concentration inequalities on the terms $\frac{1}{n}X^\top G\eps$ and $\|\widehat{\beta} - \beta\|$. See Section \ref{sec:lemmas} for the proofs. 
\begin{lemma}\label{lem:err2}
Suppose that $G\sim\Unif(\cG^{\mathrm{s}})$, Assumption \ref{asm:err1} holds, and $(\eps_i)_{i=1}^n$ satisfy the sub-Gaussian assumption in Theorem~\ref{thm:nonasymp_power_subgauss}. We have
\begin{align*}
    \P_X\left( \left\|\frac{1}{n} X^\top G \eps \right\| \ge t \right) &\le (2n+p+1) \exp \left(- \frac{3}{8} \min \left\{\sqrt{\frac{n t^2 }{p v x_*^2}}, \sqrt[3]{\frac{n^{2} t^{2} }{p v x_*^2}}\right\}\right)
\end{align*}
for any $t>0$. 
\end{lemma}
\begin{lemma}\label{lem:ridge2}
Suppose that Assumption \ref{asm:err1} holds and $(\eps_i)_{i=1}^n$ satisfy the sub-Gaussian assumption in Theorem~\ref{thm:nonasymp_power_subgauss}. We have
\begin{align*}
    \P_X\left(\|\widehat{\beta} - \beta\|\ge t\right) &\le (2n+p+1) \exp \left(- \frac{3}{8} \min \left\{\sqrt{\frac{ (t - \lambda \|Q_X\| \|\beta\|)^2 }{n p \|Q_X\|^2 v x_*^2}}, \sqrt[3]{\frac{(t - \lambda \|Q_X\| \|\beta\|)^2 }{p \|Q_X\|^2 v x_*^2}}\right\}\right)
\end{align*}
for any $t>\lambda \|Q_X\| \|\beta\|$. 
\end{lemma}
Using Lemmas \ref{lem:cov2}, \ref{lem:err2}, and \ref{lem:ridge2}, we obtain 
\begin{align*}
    \E_X(1 - {\psi}_\alpha) &\le R\Bigl(\P_X\bigl(t(Gy, X)\ge c^2\bigr) + \P_X\bigl(t(y, X)\le c^2\bigr)\Bigr)\\
    &\le O(p) \exp \left(-\frac{c^2}{8 \sigma_{\max}^2 \|\beta\|^2 \|Q_X\|^2 n p^2 x_*^4}\right) \\
    &+  O(n + p) \exp \left(- \frac{3}{8} \min \left\{\sqrt{\frac{ c^2 }{4 \sigma_{\max}^2 n p \|Q_X\|^2 v x_*^2}}, \sqrt[3]{\frac{c^{2} }{4 \sigma_{\max}^2 p \|Q_X\|^2 v x_*^2}}\right\}\right) \\
    &+ O(n+p) \exp \left(- \frac{3}{8} \min \left\{\sqrt{\frac{ (\|\beta\| - \frac{c}{\sigma_{\min}} - \lambda \|Q_X\| \|\beta\|)^2 }{n p \|Q_X\|^2 v x_*^2}}, \sqrt[3]{\frac{(\|\beta\| - \frac{c}{\sigma_{\min}} - \lambda \|Q_X\| \|\beta\|)^2 }{p \|Q_X\|^2 v x_*^2}}\right\}\right)\;. 
\end{align*}
In the derivation above, the number of randomizations, $R$, is omitted as it is a fixed constant. By choosing $c = \sigma_{\min}\|\beta\|/2$, we have 
\begin{align}
    \E_X(1 - {\psi}_\alpha) &\le O(p) \exp \left(-\frac{\sigma_{\min}^2}{32 \sigma_{\max}^2 n \|Q_X\|^2 p^2 x_*^4}\right)\nonumber \\
    &+ O(n+1) \exp \left(- \frac{3}{8} \min \left\{\sqrt{\frac{ \sigma_{\min}^2 \|\beta\|^2 }{16 \sigma_{\max}^2 n p \|Q_X\|^2 v x_*^2}}, \sqrt[3]{\frac{\sigma_{\min}^2\|\beta\|^{2} }{16 \sigma_{\max}^2 p \|Q_X\|^2 v x_*^2}}\right\}\right)\nonumber\\
    &+ O(n+p) \exp \left(- \frac{3}{8} \min \left\{\sqrt{\frac{ (1/2 - \lambda \|Q_X\|)^2 \|\beta\|^2 }{n p \|Q_X\|^2 v x_*^2}}, \sqrt[3]{\frac{(1/2 - \lambda \|Q_X\|)^2 \|\beta\|^2 }{p \|Q_X\|^2 v x_*^2}}\right\}\right)~.\label{eq:power_bound}
\end{align}
As we choose $\lambda \le \sigma_{\min}^2/3$, using the bound $f(x) = \frac{2 - x}{2+x}\ge \frac{1}{3}$ for $x\in[0, 1]$, we have
\begin{align*}
    \lambda \le \frac{1}{3} \sigma_{\min}^2 \le \frac{2 - \sigma_{\min}/\sigma_{\max}}{2+\sigma_{\min}/\sigma_{\max}} \sigma_{\min}^2\;. 
\end{align*}
After some rearrangements using $1/\|Q_X\| = \sigma_{\min}^2 + \lambda$, we obtain 
\begin{equation*}
    \lambda \|Q_X\| \le \frac{1}{2} - \frac{\sigma_{\min}}{4\sigma_{\max}}\;.
\end{equation*}
Applying this bound to \eqref{eq:power_bound}, we have
\begin{align*}
    \E_X(1 - {\psi}_\alpha) &\le O(p) \exp \left(-\frac{\sigma_{\min}^2}{32 \sigma_{\max}^2 n \|Q_X\|^2 p^2 x_*^4}\right) \\
    &+ O(n+p) \exp \left(- \frac{3}{8} \min \left\{\sqrt{\frac{ \sigma_{\min}^2 \|\beta\|^2 }{16 \sigma_{\max}^2 n p \|Q_X\|^2 v x_*^2}}, \sqrt[3]{\frac{\sigma_{\min}^2 \|\beta\|^{2} }{16 \sigma_{\max}^2 p \|Q_X\|^2 v x_*^2}}\right\}\right)\\
    &\stackrel{\text{(i)}}{\le}O(p)\exp \left(-\frac{\sigma_{\min}^2(\sigma_{\min}^2 + \lambda)^2}{32 \sigma_{\max}^2 n p^2 x_*^4}\right) \\
    & + O(n+p) \exp \left(- \frac{3}{32} \min\left\{ \sqrt{\frac{\sigma_{\min}^2(\sigma_{\min}^2 + \lambda)^2}{n p v x_*^2 \sigma_{\max}^2}}, \sqrt[3]{\frac{\sigma_{\min}^2(\sigma_{\min}^2 + \lambda)^2}{p v x_*^2 \sigma_{\max}^2}} \right\} \min \left\{\|\beta\|, \|\beta\|^{2/3}\right\}\right)\\
    &= O(p) \exp\left\{- C_n(X)\right\} + O(n+p) \exp\left\{- D_n(X) f(\|\beta\|) \right\}~.
\end{align*}
Here, (i) follows from $\|Q_X\| = 1/(\sigma_{\min}^2 + \lambda)$ and the inequality that $\min\{a_1b_1, a_2b_2\}\ge \min\{a_1,a_2\}\min\{b_1, b_2\}$ for $a_1, a_2, b_1, b_2>0$. 
\end{proof}

\section{Asymptotic Power Results}\label{sec:asymp_powerproof}
In this section, we prove asymptotic power results presented in Section~\ref{sec:power_results}. The notation is same as Section~\ref{sec:nonasymp_powerproof}.
\subsection{Proofs of Theorems~\ref{thm:consistent} and \ref{thm:consistent_random}}
First, we consider the fixed design case and prove Theorem~\ref{thm:consistent}, which is a direct corollary of Theorem~\ref{thm:nonasymp_power}.
\begin{customthm}{4}
In the fixed design setting, suppose that Assumptions~\ref{asm:err1} and \ref{asmp:regularity} hold, and 
\begin{equation*}
    p^2 \kappa^4 x_*^2 = o(sn)~\text{and}~\lambda = o(sn), ~\text{as }n\to\infty.
\end{equation*}
Then, ${\psi}_\alpha$ has detection radius
\begin{equation*}
    r_{np}(X) = B_n^{1/2}(X) = \sqrt{\frac{p \kappa^4 \sigma_*^2}{sn} + \frac{p}{sn}}~.
\end{equation*}
\end{customthm}
\begin{proof}
By the complete version of Theorem~\ref{thm:nonasymp_power} in Section~\ref{sec:nonasymp_powerproof}, for any nonzero signal $\beta$, we can write
\begin{align*}
    \E_X(1 - \psi_\alpha) &= O\Big(\widetilde{A}_n(X)\Big) + O\Big(\frac{\widetilde{B}_n(X)}{\|\beta\|^2}\Big)~.
\end{align*}
Given Equation~\eqref{eq:A4}, $\widetilde{A}_n(X) = o(1)$ and $\widetilde{B}_n(X) = O(B_n(X))$. According to Definition~\ref{def:radius}, given a sequence $(d_{np})_{n = 1}^\infty$, we have
\begin{equation*}
    \sup_{\beta\in\Theta(d_{np})}\E_X(1 - \psi_\alpha) = o(1) + O\Big(\frac{B_n(X)}{d_{np}^2}\Big)~.
\end{equation*}
Thus, we may choose $r_{np} = \sqrt{B_n(X)}$.
\end{proof}
Next, we consider the random design case and prove Theorem~\ref{thm:consistent_random}.
\begin{customthm}{5}
In the random design setting, suppose that Assumptions~\ref{asm:err1} and \ref{asmp:regularity} hold, and 
\begin{equation*}
    p^2 M_n^2 C_n^2 = o(m_n^3),~\text{and}~\lambda = o(m_n)~.
\end{equation*}
Then, ${\psi}_\alpha$ has detection radius
\begin{equation*}
    r_{np} = \sqrt{\frac{p M_n^2 \sigma_*^2}{m_n^3} + \frac{p}{m_n}}~.
\end{equation*}
\end{customthm}

\begin{proof}
Define $E_n = E_{n, 1}\cap E_{n, 2} \cap E_{n, 3}$, where
\begin{equation*}
    E_{n, 1} = \left\{ \sigma_{\min}^2 \ge m_n\right\},
    E_{n, 2} = \left\{ \sigma_{\max}^2 \le M_n\right\},
    E_{n, 3} = \left\{ x_* \le C_n\right\}.
\end{equation*}
Then, for any nonzero signal $\beta$, we have
\begin{align*}
    \E(1 - \psi_{\alpha}) &= \E(1 - \psi_{\alpha})\mathbb{I}\{X\in E_n\} + \E(1 - \psi_{\alpha})\mathbb{I}\{X\in E_n^\complement\} \\
    &\stackrel{\text{(i)}}{\le} \E \big(\E_X(1 - \psi_\alpha) \mathbb{I}\{X\in E_n\} \big)+ \P(E_n^\complement)\\
    &\stackrel{\text{(ii)}}{=} O(\E A_n(X) \mathbb{I}\{X\in E_n\}) + O\Big(\frac{\E B_n(X)\mathbb{I}\{X\in E_n\}}{\|\beta_{np}\|^2}\Big) + \P(E_n^\complement)~.
\end{align*}
In the derivation above, (i) follows by the law of iterated expectation, $X\indep (\epsilon, (G_r)_{r = 1}^R)$, and the fact that $1 - \psi_\alpha\le 1$ with probability one; (ii) follows from the complete statement of Theorem~\ref{thm:nonasymp_power} in Section~\ref{sec:nonasymp_powerproof}. 

For any $X\in E_n$, we have
\begin{align*}
    A_n(X) &\le \frac{p^2 M_n^2 C_n^2}{m_n^3} + \frac{\lambda}{m_n}~, \\
    B_n(X) &\le \frac{pM_n^2\sigma_*^2}{m_n^3} + \frac{p}{m_n}~.
\end{align*}
Therefore, 
\begin{align*}
    \E(1 - \psi_{\alpha}) 
    = O\Big(\frac{p^2 M_n^2 C_n^2}{m_n^3} + \frac{\lambda}{m_n}\Big) + \frac{1}{\|\beta_{np}\|^2}O\Big(\frac{pM_n^2\sigma_*^2}{m_n^3} + \frac{p}{m_n} \Big) + \P(E_n^\complement)~.
\end{align*}
Moreover, under the assumptions in the theorem, it holds that $\P(E_n^\complement) = o(1)$. If Equation~\eqref{eq:A4_random} holds, the first term on the right hand side is also $o(1)$. According to Definition~\ref{def:radius}, for a sequence $(d_{np})_{n = 1}^\infty$, we have
\begin{equation*}
    \sup_{\beta\in\Theta(d_{np})}\E(1 - \psi_{\alpha}) 
    = \frac{1}{\|d_{np}\|^2}O\Big(\frac{pM_n^2\sigma_*^2}{m_n^3} + \frac{p}{m_n} \Big) + o(1)~.
\end{equation*} 
Thus, we may choose
\begin{equation*}
    r_{np} = \sqrt{\frac{pM_n^2\sigma_*^2}{m_n^3} + \frac{p}{m_n}}~.
\end{equation*}
\end{proof}

\subsection{Proofs of Examples~\ref{ex:gauss} and \ref{ex:uniform}}
Here we specialize Theorem~\ref{thm:consistent_random} to Gaussian and uniform designs, providing details of Examples~\ref{ex:gauss} and \ref{ex:uniform}. 
\begin{proposition}[Gaussian design]
Suppose that Assumption~\ref{asm:err1} holds. In addition, suppose that $X_{ij}\stackrel{iid}{\sim}\cN(0, 1)$, $X\indep \eps$, and $\sigma_*= O(1)$. Then, if it holds that
\begin{equation*}
    p = o(n^{0.5-\delta}),~\text{and}~\lambda = o(n)~,
\end{equation*}
for some $\delta>0$, our test $\psi_\alpha$ has a detection radius $r_{np}=\sqrt{p/n}$. 
\end{proposition}




\begin{proof}
First, for the Gaussian design, $X\in\R^{n\times p}$ has rank $p$ with probability one. Therefore, Assumption~\ref{asmp:regularity} is satisfied almost surely. Next, to obtain the desired detection radius, it suffices to figure out the high probability bounds in Theorem~\ref{thm:consistent_random}. As hinted in Example~\ref{ex:gauss}, we choose
\begin{equation*}
    m_n = n/2,~M_n = 2n,~C_n = \sqrt{4\log(np)}~.
\end{equation*}

We show that the bounds above hold with probability $1 - o(1)$. For $m_n = n/2$ with $n$ large enough, we have
\begin{align*}
    \P(\sigma_{\min}^2 < n/2) &\le \P\left(\sigma_{\min}^2 < n (1 - \sqrt{p/n} - \sqrt{2\log{n}/n})^2\right)\\
    &= \P\left( \sigma_{\min }\left(\frac{1}{n} X^{\top} X\right) < (1-\sqrt{p / n}-\sqrt{2\log n/n})^2 \right) \stackrel{(i)}{\le} \frac{1}{n} = o(1)\;,
\end{align*}
where we obtain $(i)$ by invoking Lemma \ref{lem:gauss-cov} with $t = \sqrt{2\log n/ n}$. 

Similarly, for $M_n = 2n$ with large enough $n$, we have
\begin{align*}
    \P(\sigma_{\max}^2 > 2n) &\le \P\left(\sigma_{\max}^2 > n(1 + \sqrt{p/n} + \sqrt{2\log{n}/n})^2\right)\\
    &= \P\left(\sigma_{\max }\left(\frac{1}{n} X^{\top} X\right) >  (1 + \sqrt{p/n} + \sqrt{2\log{n}/n})^2\right)\stackrel{(i)}{\le} \frac{1}{n} = o(1)\;,
\end{align*}
where we obtain $(i)$ by invoking Lemma \ref{lem:gauss-cov} with $t = \sqrt{2\log n/ n}$. 

For $C_n=\sqrt{4\log(np)}$, we have 
\begin{equation*}
    \P(x_*>\sqrt{4\log(np)}) = \P\left(\max_{\substack{i \in [n] \\ j \in [p]}} |X_{ij}| > \sqrt{4\log(np)}\right) \le \frac{1}{np} = o(1)\;,
\end{equation*}
where the inequality is due to Lemma \ref{lem:maximal_subgauss} with $t = \sqrt{4\log(np)}$ and $v = 1$.

To summarize, the proposed sequences $(m_n)_{n=1}^\infty$, $(M_n)_{n=1}^\infty$, and $(C_n)_{n=1}^\infty$ satisfy the condition in Theorem~\ref{thm:consistent_random}, and thus Equation~\eqref{eq:A4_random} reduces to 
\begin{equation*}
    p^2 (2n)^2 4 \log(np) = o((n/2)^3),~\text{and}~\lambda = o(n/2)~. 
\end{equation*}
Ignoring the constants, we obtain $p^2\log(np) = o(n)$ and $\lambda = o(n)$. Therefore, a sufficient condition for Equation~\eqref{eq:A4_random} to hold is that
\begin{equation*}
    p = o(n^{0.5-\delta}),~\text{and}~\lambda = o(n)~,
\end{equation*}
for some $\delta>0$. Similarly, the detection radius $r_{np}$ simplifies to $\sqrt{p/n}$.
\end{proof}

Next, we consider the uniform design. 
\begin{proposition}[Uniform design]
Suppose that Assumption~\ref{asm:err1} holds. In addition, suppose that $X_{ij}\stackrel{iid}{\sim}\Unif([0, 1])$, $X\indep \eps$, and $\sigma_*= O(1)$. Then, if it holds that
\begin{equation*}
    p = o(n^{0.5}),~\text{and}~\lambda = o(n)~,
\end{equation*}
our test $\psi_\alpha$ has a detection radius $r_{np}=\sqrt{p/n}$. 
\end{proposition}
\begin{proof}
The proof follows the same line of the proof under Gaussian design. The only difference lies in showing the bounds $m_n = n/24$, $M_n = n$ hold with probability $1 - o(1)$. To this end, we will invoke Lemma~\ref{lem:chernoff_psd}. Note that one can write
\begin{equation*}
    X^\top X = \sum_{i = 1}^n A_i,\quad A_i = X_i X_i^\top~,
\end{equation*}
where $X_i^\top$ is the $i$-th row of $X$. Then, by the definition of the spectral norm, we have
\begin{equation*}
    \|A_i\| = \sup_{u\in\R^p, \|u\| = 1} (X_i^\top u)^2
\end{equation*}
Since $X_{ij}\stackrel{iid}{\sim}\Unif([0, 1])$ and $\|u\| = 1$, 
\begin{align*}
    X_i^\top u &\le \sum_{i = 1}^p |u_i| = \|u\|_1 \le \sqrt{p} \|u\|_2 = \sqrt{p}~.
\end{align*}
Therefore, we obtain an upper bound $\|A_i\|\le p$. Moreover, it is easy to verify that all singular values of $\E X^\top X$ equal $n/12$, namely, $\mu_{\min} = \mu_{\max} = n/12$ in Lemma~\ref{lem:chernoff_psd}. Then, we apply Lemma~\ref{lem:chernoff_psd} to obtain
\begin{align*}
\P\left(\sigma_{\min }({X^\top X}) \le nt/12\right) &\le p \exp(-(1-t)^2  n/ 24p)~\text{for}~t\in[0,1)~,\\
\P\left(\sigma_{\max }({X^\top X}) \ge nt/12\right) &\le p \left(\frac{e}{t}\right)^{nt/12p}~\text{for}~t\ge e~.
\end{align*}
By choosing $t = 1/2$ and $t = 12$ in the first and second line, respectively, we obtain
\begin{align*}
    \P\left(\sigma_{\min}^2 < n/24\right)
    &= \P\left( \sigma_{\min }\left(X^{\top} X\right) < \frac{n}{24} \right) {\le} p \exp(-n/96p)~, \\
    \P\left(\sigma_{\max}^2 > n\right)
    &\le \P\left(\sigma_{\max }\left( X^{\top} X\right) >  n\right)\le p\big(\frac{e}{12}\big)^{n/p}~.
\end{align*}
Under the assumption that $p = o(n^{0.5})$, both quantities on the right hand side are $o(1)$. 
\end{proof}

\subsection{Minimax Optimality}
Here, we prove Theorem \ref{thm:normal_minimax} which shows that the randomization test $\psi_\alpha$ is minimax optimal. Technically speaking, we improve the detection radius $\sqrt{p/n}$ in Example \ref{ex:gauss} to $\sqrt{p^{1/2}/n}$ through a sharper analysis.
\begin{customthm}{6}[minimax optimality]
Suppose that Assumption~\ref{asm:err1} hold. In addition, suppose that $X_{ij}\stackrel{iid}{\sim }\cN(0, 1)$, the errors $(\eps_{i})_{i=1}^n$ are i.i.d. with zero mean and 
finite fourth moments. Then, if $p = o(n^{0.5 - \delta})$ for some $\delta>0$ and $\lambda = 0$ (i.e., regular OLS), the randomization test ${\psi}_\alpha$ has a detection radius $r_{np} = \sqrt{p^{1/2}/n}$. Hence, ${\psi}_\alpha$ is minimax optimal. 
\end{customthm}

\begin{proof}
To show the minimax optimality, we follow the common approach of analyzing the mean and variance of two test statistics $t(y, X)$ and $t(Gy, X)$, and show that they are well-separated under the alternative using Chebyshev inequality. In other words, we will show that under our alternative, the squared difference between their means is asymptotically stronger than the joint variance, i.e., 
\begin{equation*}
    \E t(y, X) - \E t(Gy, X) = \Omega\left(\sqrt{\mathrm{var}(t(y, X)) + \mathrm{var}(t(Gy, X))}\right)\;.
\end{equation*}

Note that by definition, we have $t(y, X) = y^\top P_X y$, where $P_X = X(X^\top X)^{-1}X^\top$ is the projection matrix onto $\mathrm{span}(X)$. Without loss of generality, we consider $(\eps_i)_{i=1}^n$ has variance one. In the following, we will analyze $t(y, X)$ and $t(Gy, X)$ separately, similar to Theorem~\ref{thm:nonasymp_power}.

\noindent\underline{Analyzing the observed statistic $t(y, X)$. }
First, notice that
\begin{align*}
    t(y, X) = y^\top P_X y &= \beta^\top X^\top P_X X \beta + 2 \beta^\top X^\top P_X \eps + \eps^\top P_X \eps \\
    &= \beta^\top X^\top X \beta + 2 \beta^\top X^\top \eps + \eps^\top P_X \eps\;.
\end{align*}
With probability one, $X$ has rank $p$, hence $P_X$ is a projection matrix of rank $p$. Therefore, we have $\E \eps^\top P_X \eps = \E \E (\eps^\top P_X \eps|P_X) = p$ due to Lemma \ref{lem:variance_quad}, $P_X \indep \eps$, and $\mathrm{var}(\eps_i) = 1$. As $X$ has i.i.d. $\cN(0, 1)$ entries, we have $X\beta\sim\cN(0, \|\beta\|^2 I_n)$ and $\E\beta^\top X^\top X \beta = n\|\beta\|^2$. In addition, we observe $\E \beta^\top X^\top P_X \eps = 0$ as $\E\eps = 0$ and $X\indep \eps$. Therefore, 
\begin{equation}\label{eq:mean_observed}
    \E t(y, X) = n \|\beta\|^2 + p\;. 
\end{equation}

For $\mathrm{var}(t(y, X)) = \E t(y, X)^2 - (\E t(y, X))^2$, as we have analyzed its first moment, it suffices to compute the second moment $\E t(y, X)^2$. By direction calculation, 
\begin{align*}
    \E t(y, X)^2 &= \E\left( \beta^\top X^\top X \beta + 2 \beta^\top X^\top \eps + \eps^\top P_X \eps \right)^2 \\
    &= \E\left( (\beta^\top X^\top X \beta)^2 + 4(\beta^\top X^\top \eps)^2 + (\eps^\top P_X \eps)^2 \right)\\
    &+ \E \left( 4\beta^\top X^\top X \beta \beta^\top X^\top \eps +  4\beta^\top X^\top \eps \eps^\top P_X \eps + 2 \beta^\top X^\top X \beta \eps^\top P_X \eps\right)\\
    &\stackrel{\text{(i)}}{=} \E\left( (\beta^\top X^\top X \beta)^2 + 4(\beta^\top X^\top \eps)^2 + (\eps^\top P_X \eps)^2 \right) + 2\E \beta^\top X^\top X \beta \eps^\top P_X \eps\;,
\end{align*}
In the last equality (i), we use the facts $\E \beta^\top X^\top X \beta \beta^\top X^\top \eps = 0$ as $\E\eps = 0$ and $X\indep \eps$, and $\E \beta^\top X^\top \eps \eps^\top P_X \eps = 0$ as $\eps\myeq{d} -\eps$. Next, we will analyze the remaining terms in the expression above. 

For $\E(\beta^\top X^\top X \beta)^2$, notice that $X\beta \sim \cN(0, \|\beta\|^2 I_n)$. Therefore, by Lemma \ref{lem:variance_quad}, we have 
\begin{equation*}
    \E(\beta^\top X^\top X \beta)^2 = (n^2 + 2n)\|\beta\|^4\;.
\end{equation*} 

For $\E(\beta^\top X^\top \eps)^2$, we have 
\begin{align*}
    \E(\beta^\top X^\top \eps)^2 &=  \E \beta^\top X^\top \eps \eps^\top X \beta \\
    &\stackrel{\text{(i)}}{=} \E \left(\beta^\top X^\top \E (\eps \eps^\top | X) X \beta \right)\\
    &\stackrel{\text{(ii)}}{=} \E \beta^\top X^\top X \beta = n\|\beta\|^2\;, 
\end{align*}
where (i) follows by the law of iterated expectation, and (ii) follows from the fact that $\eps\indep X$ and $\eps_i$ are i.i.d. random variables with mean zero and variance one. 

For $\E (\eps^\top P_X \eps)^2$, similarly, we have 
\begin{align*}
    \E (\eps^\top P_X \eps)^2 &\stackrel{\text{(i)}}{=} \E \E\left((\eps^\top P_X \eps)^2|X\right)\\
    &\stackrel{\text{(ii)}}{\le} \E (p^2 + 2p + \E\eps_i^4 p) = p^2 + (2 + \E\eps_i^4)p\;, 
\end{align*}
where (i) and (ii) follow from the law of iterated expectation and Lemma \ref{lem:variance_quad}, respectively. 

For $\E \beta^\top X^\top X \beta \eps^\top P_X \eps$, we have 
\begin{align*}
    \E \beta^\top X^\top X \beta \eps^\top P_X \eps &\stackrel{\text{(i)}}{=} \E \left(\beta^\top X^\top X \beta \E (\eps^\top P_X \eps|X)\right) \\
    &\stackrel{\text{(ii)}}{=} p \E \beta^\top X^\top X \beta = np \|\beta\|^2\;, 
\end{align*}
where (i) again uses the law of iterated expectation, and (ii) follows from the fact that $\E (\eps^\top P_X \eps|X) = p$ by Lemma \ref{lem:variance_quad}. 

To sum up, we obtain 
\begin{align*}
    \E t(y, X)^2 &\le (n^2 + 2n)\|\beta\|^4 + 4n\|\beta\|^2 + p^2 + (2 + \E\eps_i^4)p + 2np \|\beta\|^2\;,
\end{align*}
hence
\begin{align}
    \mathrm{var}(t(y, X)) &\le (n^2 + 2n)\|\beta\|^4 + 4n\|\beta\|^2 + p^2 + (2 + \E\eps_i^4)p + 2np \|\beta\|^2 - (n\|\beta\|^2 + p)^2\nonumber \\
    &= 2n \|\beta\|^4 + 4n\|\beta\|^2 + (2 + \E\eps_i^4)p\;. \label{eq:var_observed}
\end{align}

\noindent\underline{Analyzing the randomized statistic $t(Gy, X)$. }
First, notice that
\begin{equation*}
    t(Gy, X) = y^\top G P_X G y = \beta^\top X^\top G P_X G X \beta + 2 \beta^\top X^\top GP_X G\eps + \eps^\top G P_X G \eps\;.
\end{equation*}
Similar to the analysis for $t(y, X)$, we have
\begin{equation*}
    \E t(Gy, X) = \E \beta^\top X^\top G P_X G X \beta + p \;. 
\end{equation*}
To analyze the first term, notice that
\begin{align*} 
    \E \beta^\top X^\top G P_X G X \beta &= \E \tr(\beta^\top X^\top G P_X G X \beta)\\
    &= \E \tr( P_X G X \beta \beta^\top X^\top G P_X)\\
    &\stackrel{\text{(i)}}{\le} \|\beta\|^2 \E \tr( P_X G X X^\top G P_X)\\
    &= \|\beta\|^2 \E\tr(P_X \E (G X X^\top G|X) P_X)\\
    &\stackrel{\text{(ii)}}{=} \|\beta\|^2 \E\tr(P_X D_X P_X)\\
    &= \|\beta\|^2 \E\tr(D_X P_X) \le \|\beta\|^2 \E\max_{i \in [n]} \|X_i\|^2 \tr(P_X) \\
    &= p \|\beta\|^2 \E\max_i \|X_i\|^2\;, 
\end{align*}
where $D_X = \mathrm{diag}(XX^\top)$, and $X_i^\top$ is the $i$-th row of $X$. In the derivation above, (i) uses the fact $\beta \beta^\top \preceq \|\beta\|^2 I_p$ and Lemma \ref{lem:trace_psdorder}; (ii) uses the fact $G\indep X$ and Lemma \ref{lem:quad}. Noticing that $\|X_i\|^2\sim \chi^2_p$ for $i\in[n]$, we have the following maximal inequality
\begin{equation*}
    \E\max_i \|X_i\|^2 \le p + 2 \sqrt{p \log n}+2 \log n
\end{equation*}
due to Example 2.7 of \cite{boucheron_lugosi_massart}. To sum up, we have
\begin{equation}\label{eq:mean_randomiz}
    \E t(Gy, X) \le (p + 2 \sqrt{p \log n}+2 \log n) p \|\beta\|^2 + p =: B_1\;.
\end{equation}

Now, we analyze $\mathrm{var}(t(Gy, X))$. By the law of total variance, we have 
\begin{align*}
    \mathrm{var}(t(Gy, X)) &= \E (\mathrm{var}(t(Gy, X)|X, G)) + \mathrm{var}(\E (t(Gy, X)|X, G))\\
    &= \E (\mathrm{var}( \beta^\top X^\top G P_X G X \beta + 2 \beta^\top X^\top GP_X G\eps + \eps^\top G P_X G \eps |X, G) ) \\
    &+ \mathrm{var}(\E (\beta^\top X^\top G P_X G X \beta + 2 \beta^\top X^\top GP_X G\eps + \eps^\top G P_X G \eps|X, G)) \\
    &\stackrel{\text{(i)}}{=} \underbrace{\E (\mathrm{var}(2 \beta^\top X^\top GP_X G\eps + \eps^\top G P_X G \eps |X, G) )}_\text{(a)} + \underbrace{\mathrm{var}(\beta^\top X^\top G P_X G X \beta)}_\text{(b)}\;. 
\end{align*}
In (i), we use the fact that $\E(2 \beta^\top X^\top GP_X G\eps + \eps^\top G P_X G \eps|X, G) = p$ is a constant, hence it can be ignored in the outer variance. 

To analyze $(a)$, notice that with probability one, 
\begin{align*}
    \mathrm{var}(2 \beta^\top X^\top GP_X G\eps + \eps^\top G P_X G \eps |X, G) &\stackrel{\text{(i)}}{\le} 2 \mathrm{var}(2 \beta^\top X^\top GP_X G\eps|X, G) + 2\mathrm{var}(\eps^\top G P_X G \eps |X, G)\\
    &= 8\E\left((\beta^\top X^\top GP_X G\eps)^2|X, G\right) + 2\mathrm{var}(\eps^\top G P_X G \eps |X, G)\;,
\end{align*}
where (i) uses the identity $\mathrm{var}(X+Y)\le 2 \mathrm{var}(X)+2\mathrm{var}(Y)$. Then, we have $\mathrm{var}(\eps^\top G P_X G \eps |X, G) = \mathrm{var}(\eps^\top P_X \eps |X)$ as $\eps \myeq{d} G\eps$, and 
\begin{align*}
    \E\left((\beta^\top X^\top GP_X G\eps)^2|X, G\right) &=  \E\left(\beta^\top X^\top GP_X G\eps\eps^\top G P_X G X \beta|X, G\right)\\
    &\stackrel{\text{(i)}}{=}\beta^\top X^\top GP_X G\E (\eps\eps^\top) G P_X G X \beta\\
    &= \beta^\top X^\top G P_X G X \beta \\
    &\stackrel{\text{(ii)}}{\le} \beta^\top X^\top X \beta\;,
\end{align*}
where (i) follows by $\eps\indep (X, G)$, and (ii) follows by the fact $GP_XG$ is an idempotent matrix, hence $GP_XG\preceq I_n$. Therefore, we have 
\begin{align*}
    \E \mathrm{var}(2 \beta^\top X^\top GP_X G\eps + \eps^\top G P_X G \eps |X, G) &\le 8\E \beta^\top X^\top X \beta\ + 2\mathrm{var}(\eps^\top P_X \eps |X) \\
    &= 8n\|\beta\|^2 + 2\mathrm{var}(\eps^\top P_X \eps |X) \le 8n\|\beta\|^2 + (4+2\E\eps_i^4)p\;,
\end{align*}
where the last inequality is due to Lemma \ref{lem:variance_quad}. 

To analyze $(b)$, we rely on the following lemma. Its proof is deferred to Section~\ref{sec:lemmas}.
\begin{lemma}\label{lem:variance_random_stats}
    Suppose $G\sim\Unif(\cG^\mathrm{s})$ and $X_{ij}\stackrel{iid}{\sim}\cN(0, 1)$ for $i\in[n]$ and $j\in[p]$. Let $P_X = X (X^\top X)^{-1} X^\top$. Then, for sufficiently large $n, p$ and any $\beta\in\R^p$, we have
    \begin{equation*}
        \mathrm{var}(\beta^\top X^\top G P_X G X \beta) = O(p^4 n^{2/s}\|\beta\|^4)~,
    \end{equation*}
    where $s$ is a positive integer to be specified. 
\end{lemma}
By Lemma \ref{lem:variance_random_stats} we obtain
\begin{equation*}
    (b) = O(p^4 n^{2/s} \|\beta\|^4)\;,
\end{equation*}
where $s$ is a positive integer to be specified later.

To sum up, for $n, p$ large enough, we have
\begin{equation}\label{eq:var_randomiz}
    \mathrm{var}(t(Gy, X)) \le O(p^4 n^{2/s}\|\beta\|^4) + 8n\|\beta\|^2 + (4+2\E\eps_i^4)p\;. 
\end{equation}

\noindent\underline{Deriving the final result. }
Here, we combine all the results above to show the Type II error goes to zero. From the proof of Theorem \ref{thm:nonasymp_power} in Section \ref{sec:nonasymp_powerproof}, we have
\begin{equation*}
    \E(1 - {\psi}_\alpha) \le R \P\left(t(Gy, X)>t(y, X)\right)\le R \left(\P\left(t(Gy, X)\ge c_n\right) + \P\left(t(y, X) \le c_n\right)\right)\;.
\end{equation*}
Therefore, it suffices to construct a sequence of constants $(c_n)_{n=1}^\infty$ such that that $\P\left(t(Gy, X)\ge c_n\right) + \P\left(t(y, X) \le c_n\right)$ converges to zero uniformly over all $\beta\in\Theta(d_{np})$. In the following, we specify 
\begin{equation*}
    c_n = \frac{1}{2}(\E t(y, X) + B_1) = (n + p^2 + 2 p\sqrt{p \log n}+2 p\log n)\|\beta\|^2/2+ p,
\end{equation*}
where $B_1$ is defined in \eqref{eq:mean_randomiz}. 

For the observed statistic, we have 
\begin{align*}
    \P\left(t(y, X) \le c_n\right) &= \P\left(t(y, X) - \E t(y, X) \le c_n - \E t(y, X)\right) \\
    &= \P\left(t(y, X) - \E t(y, X) \le - (n - p^2 - 2 p\sqrt{p \log n}-2 p\log n)\|\beta\|^2/2\right) \\
    &\stackrel{\text{(i)}}{\le} \P\left(t(y, X) - \E t(y, X) \le - n \|\beta\|^2/2\right) \\
    &\stackrel{\text{(ii)}}{\le} \frac{4\mathrm{var}(t(y, X))}{n^2\|\beta\|^4}\\
    &\stackrel{\text{(iii)}}{=} \frac{8n \|\beta\|^4 + 16n\|\beta\|^2 + 4(2+\E\eps_i^4)p}{n^2\|\beta\|^4}\;.
\end{align*}
In the derivation above, (i) holds because $p = o(n^{0.5-\delta})$ and $n$ dominates the quantity $n - p^2 - 2 p\sqrt{p \log n}-2 p\log n$ when $n$ is large; (ii) follows by the Chebyshev inequality; (iii) follows by \eqref{eq:var_observed}. Thus, under the alternative space $\Theta(d_{np})$, we have
\begin{equation*}
    \sup_{\beta\in\Theta(d_{np})}\P\left(t(y, X) \le c_n\right) = O(1/n)+O(1/nd_{np}^2) + O(p/n^2 d_{np}^2)~.
\end{equation*}
Since $d_{np} = \Omega(\sqrt{p^{1/2}/n})$, the right hand side above converges to zero. 

For the randomized statistic, we carry out a similar derivation
\begin{align*}
    \P\left(t(Gy, X) \ge c_n\right) &= \P\left(t(Gy, X) - \E t(Gy, X) \ge c_n - \E t(Gy, X)\right) \\
    &\le \P\left(t(Gy, X) - \E t(Gy, X) \ge (n + p^2 + 2 p\sqrt{p \log n}+2 p\log n)\|\beta\|^2/2+ p - B1\right) \\
    &= \P\left(t(Gy, X) - \E t(Gy, X) \ge (n - p^2 - 2 p\sqrt{p \log n}-2 p\log n)\|\beta\|^2/2\right) \\
    &\stackrel{\text{(i)}}{\le} \P\left(t(Gy, X) - \E t(Gy, X) \ge n\|\beta\|^2/4\right) \\
    &\stackrel{\text{(ii)}}{\le} \frac{16\mathrm{var}(t(Gy, X))}{n^2\|\beta\|^4}\\
    &\stackrel{\text{(iii)}}{=} \frac{16O(p^4 n^{2/s} \|\beta\|^4) + 128n\|\beta\|^2 + 32(2+\E\eps_i^4)p}{n^2\|\beta\|^4}\;,
\end{align*}
In the derivation above, (i) holds because $p = o(n^{0.5-\delta})$ and $n$ dominates the quantity $n - p^2 - 2 p\sqrt{p \log n}-2 p\log n$ when $n$ is large; (ii) follows by the Chebyshev inequality; (iii) follows by \eqref{eq:var_randomiz}. 
Thus, under the alternative space $\Theta(d_{np})$, we have
\begin{equation*}
    \sup_{\beta\in\Theta(d_{np})}\P\left(t(Gy, X) \ge c_n\right) = O(p^4 n^{2/s}/n^2) + O(1/nd_{np}^2) + O(p/n^2d_{np}^2)~.
\end{equation*}
As $p = o(n^{0.5 - \delta})$ for some constant $\delta$, we can choose $s$ sufficiently large to ensure $O(p^4 n^{2/s}/n^2) = o(1)$.\footnote{For instance, let $s = \lceil 1/\delta \rceil$, the smallest integer that is larger or equal to $1/\delta$} Since $d_{np} = \Omega(\sqrt{p^{1/2}/n})$, the second and third term on the right hand side above converges to zero. 

In summary, we show that $\sup_{\beta\in\Theta(d_{np})}\E(1 - \psi_\alpha)$ converges to zero for any $d_{np} = \Omega(\sqrt{p^{1/2}/n})$. By Definition~\ref{def:radius}, the test $\psi_\alpha$ achieves the detection radius $\sqrt{p^{1/2}/n}$. 
\end{proof}

\section{Asymptotic Results for Residual Randomization Tests}\label{sec:partial_proof}
\subsection{Notation}
First, we summarize the notation used in the proofs below. In the study of partial null ($H_0^S$), $\sigma_{\min}$ and $\sigma_{\max}$ denote the minimum and the maximum singular value of $X_{S^\complement}$, and $x_* = \max_{ij} |X_{S^\complement, ij}|$.
Again, we use $\kappa = \sigma_{\max}/\sigma_{\min}$ and $s = \sigma_{\min}^2/n$. 
For $A, B\in\R^{n\times n}$, we write $A\preceq B$ if $B - A$ is positive semidefinite.

As mentioned in the main text, we focus on the fixed design setup. We use $\P_X$ and $\E_X$ to denote the probability and the expectation analyzed under the fixed design setup, respectively. For simplicity, we omit the subscript $X$ whenever the random variable does not depend on $X$.

\subsection{Main Proofs}
To begin with, we give a generalized version of Assumption~\ref{asm:err_homo}, which allows for heteroskedastic errors. 
\begin{assumption}\label{asm:err_hetero}
    $(\eps_i)_{i=1}^n$ satisfy Assumption~\ref{asm:err1} (sign symmetry). In addition, they are independent random variables with mean zero and variance $(s_i^2)_{i=1}^n$ such that $s_{\min} := \min_i s_i>0$, $s_{\max} := \max_i s_i<\infty$, and $\max_i \E\eps_i^4<\infty$ for any sample size $n$. 
\end{assumption}
Next, we give the generalized version of Lemma~\ref{lem:partial_null_cond2} under Assumption~\ref{asm:err_hetero}, which is crucial for us to establish the asymptotic validity. One can recover the homoskedastic version in the main text by setting $s_{\min} = s_{\max}$.
\begin{customlem}{2}
Suppose Assumptions~\ref{asmp:regularity-partial} and \ref{asm:err_hetero} hold. Then, under $H_0^S$, if $a_n<s_{\min}^4/s_{\max}^4$, we have
\begin{align*}
    \frac{\E_X\left[\left(t(G \heps)-t(G \eps)\right)^2\right]}{\E_X\left[\left(t\left(G^{\prime} \eps\right)-t\left(G^{\prime \prime} \eps\right)\right)^2\right]}
    &= O\left(\sqrt{\frac{k(p-|S|)^3x_*^2\kappa^6}{(s_{\min}^4/s_{\max}^4-a_n)^2sn}}\right)\;.
\end{align*}
\end{customlem}
\begin{proof}
    According to Assumption \ref{asm:err_hetero}, we can write $\eps = S_\eps \eps_0$, where $S_\eps$ is the diagonal matrix with elements $s_1, \dots, s_n$, and $\eps_0$ is the normalized random vector of $\eps$. As $A$ is a projection matrix, we have $t(\eps) = \eps^\top A \eps = \eps_0^\top S_\eps A S_\eps \eps_0$ and $\tr(A) = k$. The following proof consists of two parts:
    \begin{enumerate}
        \item Lower bound for $\E_X\left[\left(t\left(G^{\prime} \eps\right)-t\left(G^{\prime \prime} \eps\right)\right)^2\right]$. 
        \item Upper bound for $\E_X\left[\left(t(G \heps)-t(G \eps)\right)^2\right]$. 
    \end{enumerate}
    
    \noindent\underline{Part 1. Lower bound for $\E\left[\left(t\left(G^{\prime} \eps\right)-t\left(G^{\prime \prime} \eps\right)\right)^2\right]$.} First, notice that 
    \begin{equation}\label{eq:conditional_mean}
        \mu(\eps):=\E_X\left(t(G \eps) \mid \eps\right)=\E_X\left(\eps^\top G^{\top} A G \eps \mid \eps\right)\stackrel{\text{(i)}}{=} \eps^\top \E_X\left(G^{\top} A G\right) \eps\stackrel{\text{(ii)}}{=}\eps^\top D_A \eps\;,
    \end{equation}
    where $D_A = \mathrm{diag}(A)$. Here, (i) follows from the fact that $G\indep \eps$, and (ii) follows from Lemma \ref{lem:quad}. Then, we have 
    \begin{align*}
        \E_X\left[\left(t\left(G^{\prime} \eps\right)-t\left(G^{\prime \prime} \eps\right)\right)^2\right] &= \E_X\left[\left(t\left(G^{\prime} \eps\right)\right)^2\right] + \E_X\left[\left(t\left(G^{\prime\prime} \eps\right)\right)^2\right] - 2 \E_X \left[t\left(G^{\prime} \eps\right)t\left(G^{\prime\prime} \eps\right)\right]\\
        &\stackrel{\text{(i)}}{=} 2\E_X\left[\left(t\left(\eps\right)\right)^2\right] - 2 \E_X \left[t\left(G^{\prime} \eps\right)t\left(G^{\prime\prime} \eps\right)\right]\\
        &\stackrel{\text{(ii)}}{=} 2\E_X\left[\left(t\left(\eps\right)\right)^2\right] - 2 \E_X\left[\E_X \left(t\left(G^{\prime} \eps\right)t\left(G^{\prime\prime} \eps\right)|\eps\right)\right]\\
        &\stackrel{\text{(iii)}}{=} 2\E_X\left[\left(t\left(\eps\right)\right)^2\right] - 2 \E_X \left[\mu^2(\eps)\right]\\
        &= 2\E_X\left[(\eps^\top A \eps)^2 - (\eps^\top D_A \eps)^2\right]\;. 
    \end{align*}
    In the derivation above, (i) uses the invariance property that $\eps \myeq{d} G'\eps \myeq{d} G'' \eps$, (ii) follows by the law of iterated expectation, and (iii) follows by $G'\indep G''$ and \eqref{eq:conditional_mean}. Noting that the expression above is a difference between the second moment of two quadratic forms, we apply Lemma \ref{lem:variance_quad} 
    to obtain 
    \begin{align*}
        \E_X(\eps^\top A \eps)^2 &= \E_X(\eps_0^\top S_\eps A S_\eps \eps_0)^2 =  \tr(S_\eps A S_\eps)^2 + 2\tr(S_\eps A S_\eps^2 A S_\eps) + \sum_{i = 1}^n (\E\eps_i^4/s_i^4-3)(S_\eps A S_\eps)_{ii}^2\;,\\
        \E_X(\eps^\top D_A \eps)^2 &= \E_X(\eps_0^\top S_\eps D_A S_\eps \eps_0)^2 =  \tr(S_\eps D_A S_\eps)^2 + 2\tr(S_\eps D_A S_\eps^2 D_A S_\eps) + \sum_{i = 1}^n (\E\eps_i^4/s_i^4-3)(S_\eps D_A S_\eps)_{ii}^2\;.
    \end{align*}
    As $S_\eps$ is a diagonal matrix, one can verify that $S_\eps D_A S_\eps$ and $S_\eps A S_\eps$ have same diagonal values. Therefore, we obtain
    \begin{align*}
        \E_X(\eps^\top A \eps)^2 - \E_X(\eps^\top D_A \eps)^2 &= 2\left(\tr(S_\eps A S_\eps^2 A S_\eps) - \tr(S_\eps D_A S_\eps^2 D_A S_\eps)\right)\;,\\
        \E_X\left[\left(t\left(G^{\prime} \eps\right)-t\left(G^{\prime \prime} \eps\right)\right)^2\right] &= 2\E_X(\eps^\top A \eps)^2 - 2\E_X(\eps^\top D_A \eps)^2\\
        &= 4\left(\tr(S_\eps A S_\eps^2 A S_\eps) - \tr(S_\eps D_A S_\eps^2 D_A S_\eps)\right)\;. 
    \end{align*}
    To lower bound $\tr(S_\eps A S_\eps^2 A S_\eps)$, notice that $S_\eps \succeq s_{\min} I_n$ in the positive semidefinite order. Then, by Lemma \ref{lem:trace_psdorder}, 
    \begin{equation*}
        \tr(S_\eps A S_\eps^2 A S_\eps) \ge s_{\min}^2 \tr(S_\eps A I_n A S_\eps) = s_{\min}^2 \tr(A S_\eps^2 A) \ge s_{\min}^4 \tr(A) = s_{\min}^4 k\;.
    \end{equation*}
    In addition, by direct calculation, 
    \begin{equation*}
        \tr(S_\eps D_A S_\eps^2 D_A S_\eps) = \sum_{i = 1}^n s_i^4 A_{ii}^2 \stackrel{\text{(i)}}{\le} s_{\max}^4 a_n k\;.
    \end{equation*}
    Here, (i) follows by the definition of $a_n$. Therefore, we have the following bound on the denominator: 
    \begin{equation*}
        \E_X\left[\left(t\left(G^{\prime} \eps\right)-t\left(G^{\prime \prime} \eps\right)\right)^2\right] \ge 4s_{\max}^4(s_{\min}^4/s_{\max}^4 - a_n)k\;. 
    \end{equation*}

    \noindent\underline{Part 2. Upper bound for $\E_X\left[\left(t(G \heps)-t(G \eps)\right)^2\right]$.} By the definition of $t_n$, we have 
    \begin{align*}
        t(G \heps)-t(G \eps) = \|A^\top G \heps\|^2 - \|A^\top G\eps\|^2 &= (\|A^\top G \heps\| + \|A^\top G\eps\|)(\|A^\top G \heps\| - \|A^\top G\eps\|)\\
        &\le (\|A^\top G \heps\| + \|A^\top G\eps\|)\|A^\top G (\heps - \eps)\|\;. 
    \end{align*}
    Taking the expectation on both sides, we obtain 
    \begin{align*}
        \E_X\left[\left(t(G \heps)-t(G \eps)\right)^2\right] &\le \E_X (\|A^\top G \heps\| + \|A^\top G\eps\|)^2\|A^\top G (\heps - \eps)\|^2\\
        &\stackrel{\text{(i)}}{\le} \sqrt{\E_X (\|A^\top G \heps\| + \|A^\top G\eps\|)^4} \sqrt{\E_X \|A^\top G (\heps - \eps)\|^4}\\
        &\stackrel{\text{(ii)}}{\le} \sqrt{8(\E_X \|A^\top G \heps\|^4 + \E_X\|A^\top G\eps\|^4)} \sqrt{\E_X \|A^\top G (\heps - \eps)\|^4}\;. 
    \end{align*}
    Here, (i) follows from the Cauchy-Schwarz inequality and (ii) follows by applying the identity $(a+b)^2 \le 2(a^2 + b^2)$ twice. Now, the remaining proof is devoted to bounding each term in the expression above.
    
    To analyze $\E_X\|A^\top G \eps\|^4$, notice that $\E_X\|A^\top G \eps\|^4 = \E_X\|A^\top G S_\eps \eps_0\|^4$, where $\eps_0$ are independent with mean zero, variance one, and fourth moments $\E (\eps_i/s_i)^4$. As $S_{\eps}^2 \preceq s_{\max}^2 I$, we apply Lemma~\ref{lem:variance_random_quad} with $c_0 = s_{\max}^2$ to obtain
    \begin{align*}
        \E\|A^\top G\eps\|^4 &\le s_{\max}^4 k^2 + (2 + \max \E (\eps_i/s_i)^4 )s_{\max}^4 k\\
        &\le s_{\max}^4 k^2 + \left(2s_{\max}^4+(s_{\max}/s_{\min})^4 \max_i \E\eps_i^4\right) k~.
    \end{align*}

    To analyze $\E_X \|A^\top G \heps\|^4$, note that under $H_0^S$, we have $\heps = (I_n - P_{X_{S^\complement}})y = (I_n - P_{X_{S^\complement}})\eps$, where $P_{X_{S^\complement}} = X_{S^\complement}(X_{S^\complement}^\top X_{S^\complement})^{-1}X_{S^\complement}^\top$ is the projection matrix onto $\mathrm{span}(X_{S^\complement})$. Then, we can write $\|A^\top G \heps\| = \|AG(I_n-P_{X_{S^\complement}})S_\eps \eps_0\| = \|AGB\eps_0\|$ with $B = (I_n-P_{X_{S^\complement}})S_\eps$. Then,
    \begin{equation*}
        BB^\top = (I_n-P_{X_{S^\complement}})S_\eps^2 (I_n-P_{X_{S^\complement}}) \preceq s_{\max}^2 (I_n-P_{X_{S^\complement}}) \stackrel{\text{(i)}}{\preceq} s_{\max}^2 I_n\;, 
    \end{equation*}
    where (i) follows from the fact that $I_n - P_{X_{S^\complement}}$ is a projection matrix. Hence, we can apply Lemma \ref{lem:variance_random_quad} with $c_0 = s_{\max}^2$ and obtain
    \begin{align*}
        \E_X \|A^\top G \heps\|^4 = \E_X\|A G B \eps_0\|^4 &\le s_{\max}^4 k^2 + \left(2s_{\max}^4+s_{\max}^4 \max_i \E(\eps_i/s_i)^4\right) k\\
        &\le s_{\max}^4 k^2 + \left(2s_{\max}^4+(s_{\max}/s_{\min})^4 \max_i \E\eps_i^4\right) k\;.
    \end{align*}
    
    To analyze $\E_X \|A^\top G (\heps - \eps)\|^4$, notice that 
    \begin{align*}
        \|A^\top G (\heps - \eps)\|^2 &\stackrel{\text{(i)}}{=} \|A^\top G P_{X_{S^\complement}} \eps\|^2\\
        &= \eps^\top P_{X_{S^\complement}} G A G P_{X_{S^\complement}} \eps\;,
    \end{align*}
    where (i) follows from $\heps = (I_n - P_{X_{S^\complement}})\eps$ under $H_0^S$. As $A$ is a projection matrix of rank $k$, there exists a matrix $U\in\R^{n\times k}$ such that $U^\top U = I_k$ and $A = UU^\top$. Hence, $\|A^\top G (\heps - \eps)\|^2 = \eps^\top P_{X_{S^\complement}} G UU^\top G P_{X_{S^\complement}} \eps = \|U^\top GP_{X_{S^\complement}}\eps\|^2$. Moreover, we have 
    \begin{align*}
        \|U^\top GP_{X_{S^\complement}}\eps\| &= \|U^\top G X_{S^\complement} (X_{S^\complement}^\top X_{S^\complement})^{-1}X_{S^\complement}^\top \eps\| \\
        &= \left\|\frac{1}{\sqrt{n}}U^\top G X_{S^\complement} (\frac{1}{n}X_{S^\complement}^\top X_{S^\complement})^{-1}\frac{1}{\sqrt{n}}X_{S^\complement}^\top \eps\right\|\\
        &\le \left\|\frac{1}{\sqrt{n}}U^\top G X_{S^\complement} \right\|\left\|\left(\frac{1}{n}X_{S^\complement}^\top X_{S^\complement}\right)^{-1}\right\|\left\|\frac{1}{\sqrt{n}}X_{S^\complement}^\top \eps\right\|\;. \\
        \E_X \|A^\top G (\heps - \eps)\|^4 = \E_X \|U^\top GP_{X_{S^\complement}}\eps\|^4 &\le \underbrace{\left\|\left(\frac{1}{n}X_{S^\complement}^\top X_{S^\complement}\right)^{-1}\right\|^4}_\text{(I)}\times \underbrace{\E_X\left\|\frac{1}{\sqrt{n}}U^\top G X_{S^\complement} \right\|^4}_\text{(II)}\times \underbrace{\E_X\left\|\frac{1}{\sqrt{n}}X_{S^\complement}^\top \eps\right\|^4}_\text{(III)}\;. 
    \end{align*}
    Roughly speaking, we include $n$ such that (I) and (III) are upper bounded by some constant, whereas (II) converges to zero. Noting that $\|(X_{S^\complement}^\top X_{S^\complement})^{-1}\| = 1/\sigma_{\min}(X_{S^\complement}^\top X_{S^\complement})$ and $\sigma_{\min}(X_{S^\complement}^\top X_{S^\complement}) = \sigma_{\min}^2$, we have $\text{(I)} = n^4/\sigma_{\min}^8$. 
    
    Next, we have 
    \begin{align*}
        \text{(II)} &= \E_X\left\|\frac{1}{\sqrt{n}}U^\top G X_{S^\complement} \right\|^2 \left\|\frac{1}{\sqrt{n}}U^\top G X_{S^\complement} \right\|^2  \\
        &\stackrel{\text{(i)}}{\le} \E_X\left\|\frac{1}{\sqrt{n}}U^\top G X_{S^\complement} \right\|^2 \frac{1}{n}\|U\|^2 \|G\|^2 \|X_{S^\complement}\|^2\\
        &\stackrel{\text{(ii)}}{\le} \frac{\sigma_{\max}^2}{n}\E_X\left\|\frac{1}{\sqrt{n}}U^\top G X_{S^\complement} \right\|^2 \\
        &\stackrel{\text{(iii)}}{\le} \frac{\sigma_{\max}^2}{n^2}\E_X\|U^\top G X_{S^\complement} \|_F^2\;.
    \end{align*}
    Here, (i) follows from the submultiplicativity of the matrix norm; (ii) follows from the fact that $\|U\|=\|G\| = 1$ and $\|X_{S^\complement}\| = \sigma_{\max}$; (iii) uses the norm inequality $\|A\|\le \|A\|_F$. Let $A_{\cdot i}$ be the $i$-th column of matrix $A$ and $g_i$ be the $i$-th diagonal element in $G$. Then, by definition of the Frobenius norm, 
    \begin{align*}
        \E_X\|U^\top G X_{S^\complement} \|_F^2 &= \sum_{i=1}^k \sum_{j = 1}^{p-|S|} \E_X (U_{\cdot i}^\top G X_{S^\complement, \cdot j})^2\\
        &= \sum_{i=1}^k \sum_{j = 1}^{p-|S|} \E_X (\sum_{k = 1}^n g_k U_{ki} X_{S^\complement,kj})^2\\
        &\stackrel{\text{(i)}}{=} \sum_{i=1}^k \sum_{j = 1}^{p-|S|} \sum_{k = 1}^n U_{ki}^2 X_{S^\complement,kj}^2\\
        &\stackrel{\text{(ii)}}{\le} \sum_{i=1}^k \sum_{j = 1}^{p-|S|} \sum_{k = 1}^n U_{ki}^2 x_*^2\\
        &\stackrel{\text{(iii)}}{=} \sum_{i=1}^k \sum_{j = 1}^{p-|S|} x_*^2 = k(p - |S|) x_*^2\;.
    \end{align*}
    In the derivation above, (i) follows from the fact that $g_i$ are independent, $\E g_i = 0$, and $\E g_i^2 = 1$; (ii) follows by the definition of $x_*$; (iii) follows from the fact that $\|U_{\cdot i}\| = 1$. Therefore, (II) can be upper bounded by
    \begin{equation*}
        \text{(II)} \le \frac{\sigma_{\max}^2k(p - |S|)x_*^2}{n^2}\;. 
    \end{equation*}

    Lastly, we apply Lemma \ref{lem:variance_quad} to (III) and obtain 
    \begin{align*}
        \text{(III)} &= \frac{1}{n^2} \left(\tr(S_\eps X_{S^\complement}X_{S^\complement}^\top S_\eps)^2 + 2\tr(S_\eps X_{S^\complement}X_{S^\complement}^\top S_\eps^2 X_{S^\complement}X_{S^\complement}^\top S_\eps) + \sum_{i = 1}^n (\E(\eps_i/s_i)^4-3)(S_\eps X_{S^\complement}X_{S^\complement}^\top S_\eps)_{ii}^2\right)\\
        &\le \frac{1}{n^2}\left(\tr(X_{S^\complement}^\top S_\eps^2 X_{S^\complement})^2 + 2\tr(X_{S^\complement}^\top S_\eps^2 X_{S^\complement}X_{S^\complement}^\top S_\eps^2 X_{S^\complement}) + \max_i \E(\eps_i/s_i)^4 \sum_{i = 1}^n (S_\eps X_{S^\complement}X_{S^\complement}^\top S_\eps)_{ii}^2\right)\\
        &\le \frac{1}{n^2}\left(s_{\max}^4\tr(X_{S^\complement}^\top X_{S^\complement})^2 + 2s_{\max}^4\tr(X_{S^\complement}^\top X_{S^\complement}X_{S^\complement}^\top X_{S^\complement}) + s_{\max}^4\max_i \E(\eps_i/s_i)^4\sum_{i,j = 1}^n (X_{S^\complement}X_{S^\complement}^\top)_{ij}^2\right)\\
        &= \frac{s_{\max}^4}{n^2}\left(\tr(X_{S^\complement}^\top X_{S^\complement})^2 + (2 + \max_i \E(\eps_i/s_i)^4)\tr(X_{S^\complement}^\top X_{S^\complement}X_{S^\complement}^\top X_{S^\complement}) \right)\\
        &\le \frac{s_{\max}^4}{n^2}\left( (p - |S|)^2 \sigma_{\max}^4 +  (2 + \max_i \E(\eps_i/s_i)^4) (p - |S|) \sigma_{\max}^4 \right)\;. 
    \end{align*}
    Now we combine the upper bound for (I), (II), and (III) to obtain 
    \begin{align*}
        \E_X \|A^\top G (\heps - \eps)\|^4 &\le \frac{n^4}{\sigma_{\min}^8}\times \frac{\sigma_{\max}^2k(p - |S|)x_*^2}{n^2} \times \frac{s_{\max}^4}{n^2}\left((p - |S|)^2 \sigma_{\max}^4 +  (2 + \max_i \E(\eps_i/s_i)^4) (p - |S|) \sigma_{\max}^4 \right)\\
        &= \frac{\sigma_{\max}^6}{\sigma_{\min}^8} k(p - |S|)^2x_*^2 s_{\max}^4 \left((p - |S|) + 2 + \max_i \E(\eps_i/s_i)^4\right)\;. 
    \end{align*}

\noindent\underline{Deriving the final result.}
In summary, we have the following upper bound for the quantity in Condition \eqref{eq:cond}:
\begin{align*}
    & \frac{\E_X\left[\left(t(G \heps)-t(G \eps)\right)^2\right]}{\E_X\left[\left(t\left(G^{\prime} \eps\right)-t\left(G^{\prime \prime} \eps\right)\right)^2\right]}\\
    &\le \frac{4\sqrt{s_{\max}^4 k^2 + \left(2s_{\max}^4+(s_{\max}/s_{\min})^4 \max_i \E\eps_i^4\right) k}\sqrt{k(p - |S|)^2x_*^2s_{\max}^4((p - |S|) + 2 + \max_i \E(\eps_i/s_i)^4)\sigma_{\max}^6/\sigma_{\min}^8}}{4s_{\max}^4(s_{\min}^4/s_{\max}^4 - a_n)k} \\
    &= \frac{\sqrt{1 + (2+\max_i \E\eps_i^4/s_{\min}^4)/k}\sqrt{k(p - |S|)^2x_*^2((p - |S|) + 2 + \max_i \E(\eps_i/s_i)^4)\sigma_{\max}^6/\sigma_{\min}^8}}{s_{\min}^4/s_{\max}^4-a_n}\\
    &\stackrel{\text{(i)}}{=} O\left(\sqrt{\frac{k(p - |S|)^3x_*^2\sigma_{\max}^6}{(s_{\min}^4/s_{\max}^4-a_n)^2\sigma_{\min}^8}}\right)\;.
\end{align*}
In (i), we use the fact that $\max_i \E\eps_i^4/s_{\min}^4$ and $\max_i \E(\eps_i/s_i)^4$ are $O(1)$ due to Assumption~\ref{asm:err_hetero}.
\end{proof}

In the following, we prove the generalized version of Theorem \ref{thm:partial_null_valid} under hetereoskedastic errors. Again, one can recover Theorem \ref{thm:partial_null_valid} in Section \ref{sec:partial} by setting $s_{\min} = s_{\max}$.
\begin{customthm}{7}
Suppose that Assumptions \ref{asmp:regularity-partial} and \ref{asm:err_hetero} hold. Suppose also that:
\begin{enumerate}[(a)]
\item The projection matrix $A$ satisfies $AX_{S^\complement} = 0$.
    \item 
    $\max_{g, g^{\prime} \in \cG, g \neq g^{\prime}} \P\big(\|Ag\eps\|=\|Ag'\eps\|\big)=o(1).$

\item $\lim\sup_{n\to\infty} a_n s_{\max}^4/s_{\min}^4 < 1$ and $k(p - |S|)^3x_*^2\kappa^6 = o(sn)$.
\end{enumerate}
Then, Condition \eqref{eq:cond} holds and the rejection probability of the residual \name-based test of Procedure 2 satisfies $\E_X(\psi_{\alpha}^S)\le \alpha + O(1/R) + o(1)$.
\end{customthm}
\begin{proof}
    Based on Lemma \ref{lem:partial_null_cond2}, we have
    \begin{align*}
        \frac{\E_X\left[\left(t(G \heps)-t(G \eps)\right)^2\right]}{\E_X\left[\left(t\left(G^{\prime} \eps\right)-t\left(G^{\prime \prime} \eps\right)\right)^2\right]} &= O\left(\sqrt{\frac{k(p-|S|)^3x_*^2\kappa^6}{(s_{\min}^4/s_{\max}^4-a_n)^2sn}}\right)\\
        &\stackrel{\text{(i)}}{=} O\left(\sqrt{\frac{k(p-|S|)^3x_*^2\kappa^6}{sn}}\right) \stackrel{\text{(ii)}}{=} o(1)\;.
    \end{align*}
    Here, (i) follows from Condition 1 and (ii) follows from Condition 2. Therefore, Condition \eqref{eq:cond} holds. 

    To obtain asymptotic validity, it remains to apply Theorem 1 of \cite{Toulis2019}. In the regime of fixed design, $X$ is nonrandom, and $\eps$ has finite first and second moments under Assumption~\ref{asm:err_hetero}. Then, it is easy to verify that Assumptions (A1)-(A2) in \cite{Toulis2019} are satisfied. As Condition \eqref{eq:cond} holds, we can apply Theorem 1 of \cite{Toulis2019} to obtain
    \begin{equation*}
        \E_X(\phi_\alpha) \le \alpha+O(1/R)+o(1)~. 
    \end{equation*}
\end{proof}

\section{Supporting Proofs of Sections~\ref{sec:nonasymp_powerproof} and \ref{sec:asymp_powerproof}}\label{sec:lemmas}
\subsection{Proof of Lemma \ref{lem:cov}}
    For any matrix $A\in\R^{p\times p}$, we have 
    \begin{equation*}
        \|A\|^2 \le \|A\|_F^2 = \sum_{i = 1}^p \|A_{\cdot j}\|^2\;,
    \end{equation*}
    where $A_{\cdot j}$ is the $j$-th column of $A$. By letting $A = X^\top G X/n $, we have
    \begin{equation*}
        \P_X\left(\left\|\frac{1}{n} X^\top G X\right\|^2 \ge t^2\right) \le \P_X\left(\sum_{j = 1}^p \left\|\frac{1}{n} X^\top G X_{\cdot j}\right\|^2 \ge t^2\right) = \P_X\left(\sum_{j = 1}^p \left\|X^\top G X_{\cdot j}\right\|^2 \ge n^2t^2\right)\;. 
    \end{equation*}
    By Markov inequality, we obtain
    \begin{align*}
        \P_X\left(\sum_{j = 1}^p \left\|X^\top G X_{\cdot j}\right\|^2 \ge n^2t^2\right) &\le \frac{1}{n^2 t^2} \sum_{j = 1}^p \E \left\|X^\top G X_{\cdot j}\right\|^2\\
        &= \frac{1}{n^2 t^2} \sum_{j = 1}^p \E X_{\cdot j}^\top G X X^\top G X_{\cdot j}\\
        &= \frac{1}{n^2 t^2} \sum_{j = 1}^p \mathrm{tr}(\E X^\top GX_{\cdot j} X_{\cdot j}^\top G X)\;.
    \end{align*}
    For every $j$, by the linearity of expectation, we have
    \begin{align*}
        \E_X (X^\top GX_{\cdot j} X_{\cdot j}^\top G X) &= X^\top \E_X(GX_{\cdot j} X_{\cdot j}^\top G)X \\
        &= X^\top D_j X\;,
    \end{align*}
    where the second equality is due to Lemma \ref{lem:quad}, and $D_j$ is a $n\times n$ diagonal matrix with elements $x_{1j}^2, \dots, x_{nj}^2$. Then we have $\mathrm{tr}(X^\top D_j X) = \mathrm{tr}(D_j X X^\top)$. As the diagonal elements of $D_j$ are upper bounded by $x_*^2$ and the diagonal elements of $X X^\top $ are nonnegative, we further obtain 
    \begin{align*}
        \mathrm{tr}(D_j X X^\top )&\le x_*^2 \mathrm{tr}(X X^\top) \le  p \sigma_{\max}^2 x_*^2 \;. 
    \end{align*}
    Notice that the upper bound above holds for any $j$. Therefore, we apply it to the Markov inequality and get
    \begin{align*}
        \P_X\left( \left\|\frac{1}{n} X^\top G X \right\| \ge t \right)\le  \frac{1}{n^2 t^2} \times p \times p x_*^2 \sigma_{\max}^2 = \frac{p^2 x_*^2 \sigma_{\max}^2}{n^2t^2}\;. 
    \end{align*}

\subsection{Proof of Lemma \ref{lem:err}}
    By Markov inequality, we have
    \begin{align*}
        \P_X\left( \left\|\frac{1}{n} X^\top G \eps \right\| \ge t \right) 
        &= \P_X\left( \left\|\frac{1}{n} X^\top G \eps \right\|^2 \ge t^2 \right)\\
        &\le  \frac{1}{n^2 t^2} \E_X \|X^\top G \eps\|^2\\
        &= \frac{1}{n^2 t^2} \E_X \eps^\top G X X^\top G \eps\\
        &= \frac{1}{n^2 t^2} \mathrm{tr}\left(\E_X (X^\top G \eps \eps^\top G X)\right)
    \end{align*}
    By the law of iterated expectation, we have
    \begin{align*}
        \E_X (X^\top G \eps \eps^\top G X)
        &= X^\top \E\left(\E (G \eps \eps^\top G|G)\right)X\\
        &= X^\top \E\left(G \E(\eps \eps^\top|G) G\right) X\\
        &\stackrel{\text{(i)}}{=} X^\top \E\left(G \E(\eps \eps^\top) G\right) X\\
        &\stackrel{\text{(ii)}}{=} X^\top D X\;, 
    \end{align*}
    where $D$ is a diagonal matrix with elements $\E\eps_i^2, \dots, \E\eps_n^2$. In the derivation above, (i) follows from $G\indep \eps$ and (ii) follows from Lemma \ref{lem:quad}. Then we have $\mathrm{tr}(X^\top D X) = \mathrm{tr}(D X X^\top )$. As the diagonal elements of $D$ are upper bounded by $\sigma_*^2$ and the diagonal elements of $X X^\top $ are nonnegative, we have 
    \begin{align*}
        \mathrm{tr}(D X X^\top )&\le \sigma_*^2 \mathrm{tr}(X X^\top) \le p \sigma_{\max}^2 \sigma_*^2\;. 
    \end{align*}
    By plugging the quantity above into Markov inequality, we get
    \begin{align*}
        \P_X\left( \left\|\frac{1}{n} X^\top G \eps \right\| \ge t \right)\le \frac{p \sigma_{\max}^2 \sigma_*^2 }{n^2t^2}\;. 
    \end{align*}

\subsection{Proof of Lemma \ref{lem:ridge}}
    By Markov inequality, we have
    \begin{align*}
        \P_X( \|\widehat{\beta} - \beta \| \ge t ) 
        &= \P_X( \|\widehat{\beta} - \beta \|^2 \ge t^2 )\\
        &\le \frac{1}{ t^2} \E_X \|\widehat{\beta} - \beta \|^2\;. 
    \end{align*}
    By definition, 
    \begin{equation*}
        \widehat{\beta} - \beta = Q_XX^\top (X\beta + \eps) - \beta = Q_XX^\top X \beta - \beta + Q_XX^\top \eps = -\lambda Q_X \beta + Q_XX^\top \eps\;, 
    \end{equation*}
    where $Q_X = (X^\top X + \lambda I )^{-1}$. Then, we have
    \begin{align*}
        \E_X \|\widehat{\beta} - \beta \|^2
        &= \E_X (\lambda^2 \|Q_X \beta\|^2 - 2\lambda \beta^\top Q_X^2 X^\top \eps +  \eps^\top X Q_X^2 X^\top \eps)\\
        &= \lambda^2 \|Q_X \beta\|^2 + \E_X (\eps^\top X Q_X^2 X^\top \eps)\,(\because \text{$\E \eps = 0$ according to Assumption \ref{asm:err1}})\;. 
    \end{align*}
    For the first term on the right hand side, we can upper bound it by $\|Q_X \beta\| \le \|Q_X\|\|\beta\|$. For the second term, we invoke the properties of trace and get 
    \begin{align*}
        \E_X (\eps^\top X Q_X^2 X^\top \eps)& = \mathrm{tr}(\E_X (\eps^\top X Q_X^2 X^\top \eps))\\
        &= \mathrm{tr}(Q_X X^\top \E (\eps \eps^\top) X Q_X )\\
        &= \mathrm{tr}(Q_X X^\top D X Q_X )\;, 
    \end{align*}
    where $D = \E \eps \eps^\top$ is a diagonal matrix with elements $\E\eps_i^2, \dots, \E\eps_n^2$. Then we have $\mathrm{tr}(Q_X X^\top D X Q_X) = \mathrm{tr}(D X Q_X^2 X^\top)$. As the diagonal elements of $D$ are upper bounded by $\sigma_*^2$ and the diagonal elements of $XQ_X^2 X^\top $ are nonnegative, we have 
    \begin{align*}
        \mathrm{tr}(D X Q_X^2 X^\top)&\le \sigma_*^2 \mathrm{tr}(X Q_X^2 X^\top) = \sigma_*^2 \mathrm{tr}(X^\top X Q_X^2)\;. 
    \end{align*}
    Note that by the definition of $Q_X$, we have $X^\top X = Q_X^{-1} - \lambda I$ and hence
    \begin{equation*}
        X^\top X Q_X^2 = (Q_X^{-1} - \lambda I) Q_X^2 = Q_X - \lambda Q_X^2\;. 
    \end{equation*}

    Now that we have analyzed each term on the right hand side of the Markov inequality, we can plug in our results and get 
    \begin{align*}
        \P_X( \|\widehat{\beta} - \beta \| \ge t ) 
        &\le \frac{1}{ t^2} \E_X \|\widehat{\beta} - \beta \|^2\\
        &\le  \frac{1}{ t^2} \left( \lambda^2 \|Q_X\|^2 \|\beta\|^2 + \mathrm{tr}(Q_X) -\lambda \mathrm{tr}(Q_X^2) \right)\\
        &\le  \frac{1}{ t^2} \left( \lambda^2 \|Q_X\|^2 \|\beta\|^2 + \mathrm{tr}(Q_X) \right)\,(\because \mathrm{tr}(Q_X^2)\ge 0)\\
        &\le \frac{1}{ t^2} \left( \lambda^2 \|Q_X\|^2 \|\beta\|^2 + p\|Q_X\| \right)\;. 
    \end{align*}

\subsection{Proof of Lemma \ref{lem:cov2}}
    As $G$ is a diagonal matrix, we have 
    \begin{equation*}
        \frac{1}{n} X^\top G X = \frac{1}{n}\sum_{i = 1}^n g_i X_i X_i^\top \eqqcolon \frac{1}{n}\sum_{i = 1}^n g_i A_i
    \end{equation*}
    where $X_i^\top$ is the $i$-th row of $X$ and $g_i$ is the $i$-th diagonal element of $G$. Note that each $g_i$ is a Rademacher random variable and each $A_i$ is a symmetric matrix of dimension $p$. Therefore, we invoke Lemma \ref{lem:chernoff_rademacher} to obtain
    \begin{equation*}
        \P_X\left( \left\|\frac{1}{n} X^\top G X \right\| \ge t \right) \le p \exp\left( -\frac{n^2 t^2}{2 v(\sum_{i} g_i A_i)}\right)\;,
    \end{equation*}
    where $v(\sum_{i} g_i A_i) = \|\sum_i A_i^2\| = \|\sum_i \|X_i\|^2 X_i X_i^\top\|$. We can bound the variance statistic by
    \begin{equation*}
        \left\|\sum_i \|X_i\|^2 X_i X_i^\top\right\| \le \sum_i \|X_i\|^2 \| X_i X_i^\top\| \le \sum_i \|X_i\|^4 \le n p^2 x_*^4\;. 
    \end{equation*}
    Hence, we derive that 
    \begin{equation*}
        \P_X\left( \left\|\frac{1}{n} X^\top G X \right\| \ge t \right) \le p \exp\left( -\frac{n t^2}{2 p^2 x_*^4}\right)\;,
    \end{equation*}

\subsection{Proof of Lemma \ref{lem:err2}}
    As $G$ is a diagonal matrix, we have 
    \begin{equation*}
        \frac{1}{n} X^\top G \eps = \frac{1}{n}\sum_{i = 1}^n g_i \eps_i X_i = \frac{1}{n}\sum_{i = 1}^n A_i\;, \quad A_i \coloneqq g_i \eps_i X_i\;,
    \end{equation*}
    where $g_i$ is the $i$-th diagonal element of $G$.
    For $e>0$, define the event $A_e = \left\{\max_{i\in [n]} |\eps_i| \le e\right\}$. Then, we apply the law of total probability to obtain 
    \begin{align*}
        \P_X\left( \left\|\frac{1}{n} X^\top G \eps \right\| \ge t \right) &= \P_X\left( \left\|\frac{1}{n} X^\top G \eps \right\| \ge t \mid A_e \right)\P(A_e) + \P_X\left( \left\|\frac{1}{n} X^\top G \eps \right\| \ge t \mid A_e^\complement \right)\P(A_e^\complement)\\
        &\le \underbrace{\P_X\left( \left\|\frac{1}{n} X^\top G \eps \right\| \ge t \mid A_e \right)}_\text{(I)} + \underbrace{\P(A_e^\complement)}_\text{(II)}\;. 
    \end{align*}
    To prove our lemma, it suffices to bound (I) and (II) separately.

    For (I), notice that $\|A_i\|\le \max_i |\eps_i| \sqrt{p} x_*$ for any $i\in[n]$. Conditional on the event $A_e$, the upper bound reduces to $\gamma = \sqrt{p}e x_*$. Therefore, we invoke Lemma \ref{lem:chernoff_general} to obtain 
    \begin{equation*}
        \text{(I)} \le\left(p+1\right) \exp \left(-\frac{3 n t^2 }{6\gamma^2 + 2\gamma t}\right)\;.
    \end{equation*}
    For (II), under Assumption \ref{asm:err1} and the sub-Gaussian assumption, we apply Lemma \ref{lem:maximal_subgauss} to obtain
    \begin{equation*}
        \text{(II)} \le 2 \P(\max_i \eps_i \ge e) \le 2 n \exp \left(-\frac{e^2}{2 v}\right)\;. 
    \end{equation*}
    Therefore, 
    \begin{align*}
        \text{(I)} + \text{(II)} & \le \left(p+1\right) \exp \left(-\frac{3 n t^2 }{8\gamma^2 + 2\gamma t}\right) + 2 n \exp \left(-\frac{e^2}{2 v }\right)\\
        &\le \left(p+1\right) \max\left\{\exp \left(-\frac{3 n t^2 }{8\gamma^2 }\right), \exp \left(-\frac{3 n t }{8\gamma}\right)\right\} + 2 n \exp \left(-\frac{e^2}{2 v }\right)\\
        &\le \max\left\{ \left(p+1\right)\exp \left(-\frac{3 n t^2 }{8\gamma^2 }\right) + 2 n \exp \left(-\frac{e^2}{2 v }\right), \left(p+1\right)\exp \left(-\frac{3 n t }{8\gamma}\right) + 2 n \exp \left(-\frac{e^2}{2 v }\right) \right\}\\
        &= \max\left\{ \left(p+1\right)\exp \left(-\frac{3 n t^2 }{8{p} e^2 x_*^2 }\right) + 2 n \exp \left(-\frac{e^2}{2 v }\right), \left(p+1\right)\exp \left(-\frac{3 n t }{8\sqrt{p} e x_*}\right) + 2 n \exp \left(-\frac{e^2}{2 v }\right) \right\}\;.
    \end{align*}
    Last, we need to specify $e$ to get the sharpest rate. For the first component on the right hand side, we choose $e^2 = \sqrt{nvt^2/px_*^2}$ so that two exponents have the same rate, that is, ${n t^2 }/{p e^2 x_*^2}= {e^2}/{v}$. Similarly, we choose $e = (nvt/\sqrt{px_*^2})^{1/3}$ for the second component and obtain 
    \begin{align*}
        \text{(I)} + \text{(II)} &\le\max \Biggl\{\left(p+1\right) \exp \left(-\frac{3 \sqrt{n} t }{8 \sqrt{p s x_*^2}}\right) + 2n \exp \left(-\frac{ \sqrt{n} t }{2 \sqrt{p s x_*^2}}\right), \\
        &\left(p+1\right) \exp \left(-\frac{3 n^{2/3} t^{2/3} }{8 (p s x_*^2)^{1/3}}\right) + 2n \exp \left(-\frac{ n^{2/3} t^{2/3} }{2 (p s x_*^2)^{1/3}}\right)\Biggr\}\\
        &\le (2n+p+1) \exp \left(- \frac{3}{8} \min \left\{\sqrt{\frac{n t^2 }{p s x_*^2}}, \sqrt[3]{\frac{n^{2} t^{2} }{p s x_*^2}}\right\}\right)\;. 
    \end{align*}

\subsection{Proof of Lemma \ref{lem:ridge2}}
    By the definition of $\widehat{\beta}$, we have 
    \begin{equation*}
        \widehat{\beta} - \beta = Q_XX^\top (X\beta + \eps) - \beta = Q_XX^\top X \beta - \beta + Q_XX^\top \eps = -\lambda Q_X \beta + Q_XX^\top \eps\;.
    \end{equation*}
    Hence, using the triangle inequality and the submultiplicativity, we obtain
    \begin{equation*}
        \|\widehat{\beta} - \beta\| \le \lambda \|Q_X\| \|\beta\| + \|Q_X\|\|X^\top \eps\|\;.
    \end{equation*}
    Using the same idea as in Lemma \ref{lem:err2}, we can deduce that 
    \begin{equation*}
        \P_X\left( \left\|\frac{1}{n} X^\top \eps \right\| \ge t \right) \le (2n+p+1) \exp \left(- \frac{3}{8} \min \left\{\sqrt{\frac{n t^2 }{p v x_*^2}}, \sqrt[3]{\frac{n^{2} t^{2} }{p v x_*^2}}\right\}\right)
    \end{equation*}
    for any $t>0$. Therefore, we have
    \begin{align*}
        \P_X\left( \|\widehat{\beta} - \beta\| \ge t \right) &\le \P\left( \lambda \|Q_X\| \|\beta\| + \|Q_X\|\|X^\top \eps\| \ge t \right)\\
        &= \P_X\left(\left\|\frac{1}{n} X^\top \eps\right\| \ge \frac{t - \lambda \|Q_X\| \|\beta\|}{n\|Q_X\|} \right)\\
        &\le (2n+p+1) \exp \left(- \frac{3}{8} \min \left\{\sqrt{\frac{ (t - \lambda \|Q_X\| \|\beta\|)^2 }{n p \|Q_X\|^2 v x_*^2}}, \sqrt[3]{\frac{(t - \lambda \|Q_X\| \|\beta\|)^2 }{p \|Q_X\|^2 v x_*^2}}\right\}\right)
    \end{align*}
    for any $t>\lambda \|Q_X\| \|\beta\|$. 

\subsection{Proof of Lemma \ref{lem:variance_random_stats}}
First, notice that
\begin{align}
    \mathrm{var}(\beta^\top X^\top G P_X G X \beta)&\le \E(\beta^\top X^\top G P_X G X \beta)^2\nonumber\\
    &= \E \tr(\beta^\top X^\top G P_X G X \beta\beta^\top X^\top G P_X G X \beta)\nonumber\\
    &\stackrel{\text{(i)}}{\le} \|\beta\|^2 \E \tr(\beta^\top X^\top G P_X G X X^\top G P_X G X \beta)\nonumber\\
    &= \|\beta\|^2 \E \tr(X^\top G P_X G X \beta \beta^\top X^\top G P_X G X)\nonumber\\
    &\stackrel{\text{(ii)}}{\le} \|\beta\|^4 \E \tr(X^\top G P_X G X X^\top G P_X G X)\nonumber \\
    &= \|\beta\|^4 \E \tr(G X X^\top G P_X G X X^\top G P_X)\nonumber\\
    &\stackrel{\text{(iii)}}{=} \|\beta\|^4 \E \tr( \E (G X X^\top G P_X G X X^\top G|X) P_X)\label{eq:var_randomiz1}\;. 
\end{align}
In the derivation above, (i) and (ii) follows from the fact $\beta \beta^\top \preceq \|\beta\|^2 I_p$ and Lemma \ref{lem:trace_psdorder}; (iii) follows by the law of iterated expectation. 

We define $K = XX^\top$ and $P = P_X$ for a simpler notation. Then, we can write the inner expectation in \eqref{eq:var_randomiz1} as $\E(G K G P G K G|X)$. Therefore, we apply Lemma \ref{lem:quad} to obtain
\begin{align*}
    &\E (G X X^\top G P_X G X X^\top G|X) = C\;, \\
    &C_{ij} = \begin{cases}
        (K_{ii}K_{jj}+K_{ij}^2)P_{ij} &\text{ if } i\neq j\;, \\
        \sum_{k = 1}^n K_{ik}^2 P_{kk} &\text{ if } i = j\;. \\
    \end{cases}
\end{align*}
Then, we have 
\begin{align}
    & \E \tr( \E (G X X^\top G P_X G X X^\top G|X) P_X) = \E \tr( C P_X) = \E \sum_{i,j = 1}^n C_{ij} P_{ij}\nonumber\\
    &= \E \sum_{i\neq j} (K_{ii}K_{jj}+K_{ij}^2)P_{ij}^2 + \E \sum_{i = 1}^n \sum_{k = 1}^n K_{ik}^2 P_{kk} P_{ii}\nonumber\\
    &\le \E \sum_{i\neq j} (\max_{i\neq j} K_{ii}K_{jj}+ \max_{i \neq j} K_{ij}^2)P_{ij}^2 + \E \sum_{i = 1}^n \sum_{k = 1}^n (\max_{i,j} K_{ik}^2) P_{kk} P_{ii}\nonumber\\
    &\le \E (\max_{i\neq j} K_{ii}K_{jj}+ \max_{i\neq j} K_{ij}^2) \sum_{i\neq j} P_{ij}^2 + \E \sum_{i\neq k} (\max_{i\neq k} K_{ik}^2) P_{kk} P_{ii} + \E \sum_{i = 1}^n (\max_{i} K_{ii}^2) P_{ii}^2\nonumber\\
    &\le \E (\max_{i\neq j} K_{ii}K_{jj}+ \max_{i\neq j} K_{ij}^2) \sum_{i,j} P_{ij}^2 + \E (\max_{i\neq k} K_{ik}^2) \sum_{k = 1}^n P_{kk} \sum_{i = 1}^n P_{ii} + \E (\max_{i} K_{ii}^2) \sum_{i = 1}^n P_{ii}^2\nonumber\\
    &\stackrel{\text{(i)}}{\le} p (\E \max_{i\neq j} K_{ii}K_{jj}+ \E \max_{i\neq j} K_{ij}^2) + p^2 \E(\max_{i\neq j} K_{ij}^2) + p \E \max_i K_{ii}^2\label{eq:var_randomiz2}\;,
\end{align}
where (i) follows from the facts $\sum_{i, j} P_{ij}^2 = \tr(P^2) = p$, $\sum_{i = 1}^n P_{ii} = \tr(P) = p$, and $P_{ii}\in[0, 1]$ according to Lemma \ref{lem:leverage_score}. Next, we will upper bound $\E\max_{i\neq j} K_{ii}K_{jj}$, $\E \max_{i\neq j} K_{ij}^2$, $\E \max_i K_{ii}^2$ one by one.

\underline{Analyzing $\E\max_{i\neq j} K_{ii}K_{jj}$. }
For any $s>1$, we have
\begin{align*}
    (\E\max_{i\neq j} K_{ii}K_{jj})^s &\stackrel{\text{(i)}}{\le} \E\left(\max_{i\neq j} (K_{ii}K_{jj})^s \right)\\
    &= \E \left(\max_{i> j} (K_{ii}K_{jj})^s \right) \le \E \sum_{i>j} (K_{ii}K_{jj})^s\\
    &\stackrel{\text{(ii)}}{\le} \sum_{i>j} \E(\|X_i\|^2 \|X_j\|^2)^s\\
    &\stackrel{\text{(iii)}}{\le} \frac{n(n-1)}{2} (\E\|X_i\|^{2s})^2\\
    &\stackrel{\text{(iv)}}{\le} n^2 2^{2s} \left(\frac{\Gamma(s+p/2)}{\Gamma(p/2)}\right)^2\;,
\end{align*}
where $\Gamma(x)$ is the gamma function. In the derivation above, (i) follows by Jensen's inequality; (ii) follows by the definition that $K_{ii} = \|X_i\|^2$, where $X_i^\top$ is the $i$-th row of $X$; (iii) follows from the fact that $X_i\indep X_j$ and $X_i\myeq{d} X_j$ for $i\neq j$; (iv) follows from standard results on the moments of chi-squared distribution, as $\|X_i\|^2\sim\chi^2_p$ for any $i$. 

By taking the power $1/s$ on both sides, we obtain 
\begin{equation*}
    \E\max_{i\neq j} K_{ii}K_{jj} \le 4 n^{2/s} \left(\frac{\Gamma(s+p/2)}{\Gamma(p/2)}\right)^{2/s}\;.
\end{equation*}
Without loss of generality, we assume $s$ is an integer. This enables us to use the identity $\Gamma(x+1) = x\Gamma(x)$ to obtain
\begin{equation*}
    \frac{\Gamma(s+p/2)}{\Gamma(p/2)} \le (s+p/2-1)^s\;.
\end{equation*}
In summary, we obtain a maximal inequality that depends on $s$, i.e., 
\begin{equation}\label{eq:maximal1}
    \E\max_{i\neq j} K_{ii}K_{jj} \le 4 n^{2/s} \left(s+p/2-1\right)^{s\times 2/s} = 4 n^{2/s} \left(s+p/2-1\right)^{2}\;. 
\end{equation}

\underline{Analyzing $\E\max_{i\neq j} K_{ij}^2$. }
By the definition $K_{ij} = X_i^\top X_j$, we have
\begin{align}
    \E\max_{i\neq j} K_{ij}^2 &= \E\max_{i\neq j} (X_i^\top X_j)^2 \nonumber\\
    &\stackrel{\text{(i)}}{\le} \E\max_{i\neq j} \|X_i\|^2 \|X_j\|^2 \nonumber\\
    &= \E\max_{i\neq j} K_{ii} K_{jj}\nonumber\\
    &\stackrel{\text{(ii)}}\le 4 n^{2/s} \left(s+p/2-1\right)^{2}\;,\label{eq:maximal2}
\end{align}
where (i) follows by Cauchy-Schwarz inequality and (ii) follows from \eqref{eq:maximal1}. 

\underline{Analyzing $\E\max_{i} K_{ii}^2$.}
Following a similar analysis as above, for any $s>1$, we have
\begin{align*}
    (\E\max_{i} K_{ii}^2)^s &\stackrel{\text{(i)}}{\le} \E\left(\max_{i} K_{ii}^{2s} \right)\\
    &= \E \left(\max_{i} K_{ii}^{2s} \right) \le \E \sum_{i} K_{ii}^{2s}\\
    &\stackrel{\text{(ii)}}{\le} \sum_{i=1}^n \E(\|X_i\|^2)^{2s}\\
    &\stackrel{\text{(iii)}}{=} n \E(\|X_i\|^2)^{2s}\\
    &\stackrel{\text{(iv)}}{\le} n 2^{2s} \frac{\Gamma(2s+p/2)}{\Gamma(p/2)}\;.
\end{align*}
In the derivation above, (i) follows by Jensen's inequality; (ii) follows by the definition that $K_{ii} = \|X_i\|^2$, where $X_i^\top$ is the $i$-th row of $X$; (iii) follows from the fact that $X_i\myeq{d} X_j$ for $i\neq j$; (iv) follows from standard results on the moments of chi-squared distribution, as $\|X_i\|^2\sim\chi^2_p$ for any $i$. 

By taking the power $1/s$ on both sides, we obtain 
\begin{equation*}
    \E\max_{i} K_{ii}^2 \le 4 n^{1/s} \left(\frac{\Gamma(2s+p/2)}{\Gamma(p/2)}\right)^{1/s}\;.
\end{equation*}
Assuming $s$ is an integer, we apply the identity $\Gamma(x+1) = x\Gamma(x)$ again to obtain
\begin{equation*}
    \frac{\Gamma(2s+p/2)}{\Gamma(p/2)} \le (2s+p/2-1)^{2s}\;.
\end{equation*}
In summary, we obtain another maximal inequality that depends on $s$, i.e., 
\begin{equation}\label{eq:maximal3}
    \E\max_{i} K_{ii}^2 \le 4 n^{1/s} \left(2s+p/2-1\right)^{2s\times 1/s} = 4 n^{1/s} \left(2s+p/2-1\right)^{2}\;. 
\end{equation}

\underline{Combining all the maximal inequalities. }
Collecting the results above, we have 
\begin{align*}
    &p (\E \max_{i\neq j} K_{ii}K_{jj}+ \E \max_{i\neq j} K_{ij}^2) + p^2 \E(\max_{i\neq j} K_{ij}^2) + p \E \max_i K_{ii}^2 \\
    &\stackrel{\text{(i)}}{\le} 8 p n^{2/s} \left(s+p/2-1\right)^{2} + 4 p^2 n^{2/s} \left(s+p/2-1\right)^{2} + 4 n^{1/s} \left(2s+p/2-1\right)^{2}\\
    &\stackrel{\text{(ii)}}{=} O(p^3 n^{2/s} + p^4 n^{2/s} + p^2 n^{1/s}) = O(p^4 n^{2/s}) \;. 
\end{align*}
In the derivation above, (i) follows by \eqref{eq:maximal1}, \eqref{eq:maximal2}, \eqref{eq:maximal3}; (ii) follows by taking $n, p$ sufficiently large. Combining the bound above with \eqref{eq:var_randomiz1} and \eqref{eq:var_randomiz2} completes the proof.

\section{Auxiliary Lemmas}\label{sec:auxiliary}
In this section, we introduce auxiliary lemmas used in our proofs, mainly adapted from previous works. First, we give the following matrix concentration inequalities from \cite{Lopes2011} and \cite{Tropp2015}. 
\begin{lemma}[Lemma 4 of \cite{Lopes2011}]\label{lem:gauss-cov}
For $p \le n$, let $X \in \mathbb{R}^{n \times p}$ be a random matrix with i.i.d. $\cN(0,1)$ entries. Then, for all $t>0$, we have
\begin{align*}
\mathbb{P}\left(\sigma_{\max }\left(\frac{1}{n} X^{\top} X\right) \geq(1+\sqrt{p/ n}+t)^2\right) & \leq \exp \left(-n t^2 / 2\right)\;, \\
\mathbb{P}\left(\sigma_{\min }\left(\frac{1}{n} X^{\top} X\right) \leq(1-\sqrt{p / n}-t)^2\right) & \leq \exp \left(- n t^2 / 2\right) \;,
\end{align*}
where $\sigma_{\max}(A)$ and $\sigma_{\min}(A)$ denote the maximum and the minimum singular value of $A$. 
\end{lemma}

The following three lemmas are adapted from \cite{Tropp2015}. 
\begin{lemma}[Theorem 4.1.1 of \cite{Tropp2015}]\label{lem:chernoff_rademacher}
Consider a finite sequence $\{A_i\}$ of fixed symmetric matrices with dimension $p$, and let $\gamma_i\stackrel{iid}{\sim}\operatorname{Bernoulli}(1/2)$. Introduce the matrix Rademacher series
\begin{equation*}
    {Y}=\sum_i \gamma_i A_i\;,
\end{equation*}
and let $v(Y) = \|\sum_k A_i^2\|$. Then, for all $t > 0$,
\begin{equation*}
    \mathbb{P}\left\{\|Y\| \ge t\right\} \le p \exp \left(\frac{-t^2}{2 v(Y)}\right)\;.
\end{equation*}
\end{lemma}

\begin{lemma}[Theorem 5.1.1 of \cite{Tropp2015}]\label{lem:chernoff_psd}
Consider a finite sequence $\{A_i\}$ of independent, random, positive semidefinite matrices with common dimension $p$. Assume that $\|A_i\| \le L$ for any $i$. Introduce the random matrix
$$
Y=\sum_i A_i
$$
Denote by $\mu_{\min }$ and $\mu_{\max }$ the minimum and the maximum eigenvalue of the expectation $\E Y$. Then, we have
\begin{align*}
    \P\left\{\sigma_{\min }({Y}) \le t \mu_{\min}\right\} &\le p \exp(-(1-t)^2 \mu_{\min } / 2 L)~\text{for}~t\in[0,1)~,\\
    \P\left\{\sigma_{\max }({Y}) \ge t \mu_{\max}\right\} &\le p \left(\frac{e}{t}\right)^{t\mu_{\max} / L}~\text{for}~t\ge e~.
\end{align*}
\end{lemma}

\begin{lemma}[Theorem 6.1.1 of \cite{Tropp2015}]\label{lem:chernoff_general}
    Consider a finite sequence $\{A_i\}$ of random matrices in $\R^{d_1 \times d_2}$ and $\E A_i =0$. If $\|A_i\| \le \gamma$ holds almost surely for any $i \in [n]$, then for every $t>0$, we have
\begin{equation*}
    \P\left(\left\|\frac{1}{n} \sum_{i=1}^n A_i\right\|>t \right) \leq\left(d_1+d_2\right) \exp \left(-\frac{3 n t^2 }{6\gamma^2 + 2\gamma t} \right)\;. 
\end{equation*}
\end{lemma}

We introduce the following maximal inequality, which can be found in Theorem 2.5 of \cite{boucheron_lugosi_massart}.
\begin{lemma}[Maximal inequality]\label{lem:maximal_subgauss}
    Suppose that $Z_i$, $j \in [n]$ are sub-Gaussian random variables with common variance factor $v$. We have
    \begin{equation*}
        \mathbb{P}\left(\max_{i \in [n]}Z_i \geq t\right) \le n \exp \left(-\frac{t^2}{2 v}\right)
    \end{equation*}
    for any $t>0$. 
\end{lemma}

The following result follows from the linearity of the trace operation.
\begin{lemma}\label{lem:trace_psdorder}
    Suppose $A\in\R^{n\times n}$ is positive semidefinite and $B\in\R^{p\times n}$. We have 
    \begin{equation*}
        \tr(BAB^\top) \ge 0\;. 
    \end{equation*}
    As a corollary, for two matrices $M, N\in\R^{n\times n}$ such that $M \preceq N$, we have 
    $\tr(BMB^\top) \le \tr(BNB^\top)$. 
\end{lemma}

The following lemma calculates the moments of $G$, which is crucial for analyzing the randomized statistics. 
\begin{lemma}\label{lem:quad}
    Let $A, B\in\R^{n\times n}$ be symmetric matrices and $\cG = \cG^{\mathrm{s}}$ defined in Assumption~\ref{asm:err1}. For $G\sim\Unif(\cG)$, we have
    \begin{equation*}
        \E A G = 0\;,\qquad \E G A G = \mathrm{diag}(A)\;, \qquad \E (G A G B G A G) = C\;,
    \end{equation*}
    where $M = \mathrm{diag}(A)$ is a $n\times n$ matrix with $M_{ii} = A_{ii}$, $M_{ij} = 0$, and 
    \begin{equation*}
        C_{ij} = \begin{cases}
        (A_{ii}A_{jj}+A_{ij}^2)B_{ij} &\text{ if } i\neq j\;, \\
        \sum_{k = 1}^n A_{ik}^2 B_{kk} &\text{ if } i = j\;. \\
    \end{cases}
    \end{equation*}
\end{lemma}
\begin{proof}
As $G$ is a diagonal matrix under the sign-flipping group, we denote the diagonal elements by $g_1, \dots, g_n$. By linear algebra, for any $i, j$, we have
\begin{equation}\label{eq:quad}
\begin{aligned}
    \E (AG)_{ij} &= \E A_{ij} g_j\;,\\ 
\E (G A G)_{ij} &= \E g_i g_j A_{ij}\;, \\
\E (G A G B G A G)_{ij} &= \sum_{p,q=1}^n \E g_i g_p g_q g_j A_{ip} B_{pq} A_{qj}\;. 
\end{aligned}
\end{equation}
For the sign-flipping group, we have
\begin{align*}
    \E g_i &= 0, \\
    \E g_{i}g_{j} &= 
    \begin{cases}
    1 & \text{ if } i = j\;,\\
    0 & \text{ else}\;.
    \end{cases}
\end{align*}
We plug the expectations above into \eqref{eq:quad} and obtain
\begin{equation*}
        \E A G = 0\;,\qquad \E G A G = \mathrm{diag}(A)\;, \qquad \E (G A G B G A G) = C\;.
\end{equation*}
\end{proof}

The following lemma calculates the first and second moment of $\chi^2$-type random variables. 
\begin{lemma}\label{lem:variance_quad}
    Suppose $A\in\R^{n\times n}$ is positive semidefinite and $\eps_1, \dots, \eps_n$ are independent random variables with mean zero, variance one, and $\E\eps^4_i<\infty$. We have 
    \begin{equation*}
        \E \eps^\top A \eps = \tr(A)\;, \quad \E (\eps^\top A \eps)^2 = \tr(A)^2 + 2\tr(AA^\top) + \sum_{i = 1}^n (\E\eps_i^4-3)A_{ii}^2\;,
    \end{equation*}
    where $A_{ii}$ is the $i$-th diagonal of $A$. 
\end{lemma}

\begin{proof}
    Let $A_{ij}$ be the $(i, j)$-th element of $A$. Then, 
    \begin{equation*}
        \eps^\top A \eps = \sum_{i,j = 1}^n \eps_i \eps_j A_{ij}\;.
    \end{equation*}
    As $\E\eps_i = 0$, $\E\eps_i^2 = 1$, and $\eps_i \indep \eps_j$ for $i\neq j$, we have 
    \begin{equation*}
        \E \eps^\top A \eps = \sum_{i,j = 1}^n \E (\eps_i \eps_j) A_{ij} = \sum_i A_{ii} = \tr(A)\;. 
    \end{equation*}

    Similarly, 
    \begin{equation*}
        \E (\eps^\top A \eps)^2 = \sum_{i,j,k,l = 1}^n \E (\eps_i \eps_j \eps_k \eps_l) A_{ij} A_{kl}\;. 
    \end{equation*}
    By some calculation, we obtain
    \begin{equation*}
        \E (\eps_i \eps_j \eps_k \eps_l)
        = \begin{cases}
            1 &\text{ if }i = j \neq k = l\;, \\
            1 &\text{ if }i = k \neq j = l\;, \\
            1 &\text{ if }i = l \neq j = k\;, \\
            \E\eps_i^4 &\text{ if }i = j = k = l\;, \\
            0 &\text{ otherwise}\;. 
        \end{cases}
    \end{equation*}
    Hence, 
    \begin{align*}
        \E (\eps^\top A \eps)^2 &= \sum_{i,k = 1, i\neq k}^n A_{ii}A_{kk} + \sum_{i,j = 1, i\neq j}^n A_{ij}A_{ij} + \sum_{i,j = 1, i\neq j}^n A_{ij}A_{ji} + \sum_{i = 1}^n \E\eps_i^4 A_{ii}^2\\
        &\stackrel{\text{(i)}}{=} \sum_{i,k = 1, i\neq k}^n A_{ii}A_{kk} + 2\sum_{i,j = 1, i\neq j}^n A_{ij}A_{ji} + \sum_{i = 1}^n \E\eps_i^4 A_{ii}^2\\
        &= \sum_{i,k = 1}^n A_{ii}A_{kk} + 2\sum_{i,j = 1}^n A_{ij}A_{ji} + \sum_{i = 1}^n (\E\eps_i^4-3)  A_{ii}^2\\
        &= (\sum_{i = 1}^n A_{ii})^2 + 2\tr(AA^\top) + \sum_{i = 1}^n (\E\eps_i^4-3) A_{ii}^2\\
        &= \tr(A)^2 + 2\tr(AA^\top) + \sum_{i = 1}^n (\E\eps_i^4-3)A_{ii}^2\;. 
    \end{align*}
    Here, (i) follows by the fact that $A = A^\top$.
\end{proof}

We review a standard result on the projection matrix below. 
\begin{lemma}\label{lem:leverage_score}
    Suppose $A\in\R^{n\times n}$ is a projection matrix, that is, symmetric and idempotent. For any $i$, the $i$-th diagonal $A_{ii}$ satisfies $0\le A_{ii}\le 1$.
\end{lemma}

Last, we prove the following moment bound, which is used in our main proofs. 
\begin{lemma}\label{lem:variance_random_quad}
    Suppose $A\in\R^{n\times n}$ is a projection matrix of rank $k$, and $B\in\R^{n\times n}$ satisfies $BB^\top \preceq c_0 I_n$ for some constant $c_0>0$. Also suppose that $\eps_1, \dots, \eps_n$ are independent random variables with mean zero, variance one, and $\E\eps^4_i<\infty$. Let $G\sim \Unif(\cG)$, where $\cG$ is the sign-flipping group. Then we have
    \begin{equation*}
        \E\|A G B \eps\|^4 = \E (\eps^\top B^\top GAGB \eps)^2 \le c_0^2 k^2 + (2+\max_i \E\eps_i^4) c_0^2 k\;.
    \end{equation*}
\end{lemma}
\begin{proof}
    Let $H = B^\top GAGB$, which is a random matrix. Similar to the proof of Lemma \ref{lem:variance_quad}, we have
    \begin{align*}
        \eps^\top H \eps &= \sum_{i,j = 1}^n \eps_i \eps_j H_{ij}\;, \\
        (\eps^\top H \eps)^2 &= \sum_{i,j,k,l = 1}^n \eps_i \eps_j \eps_k \eps_l H_{ij} H_{kl}\;,\\
        \E (\eps^\top H \eps)^2 &= \sum_{i,j,k,l = 1}^n \E\eps_i \eps_j \eps_k \eps_l \E H_{ij} H_{kl}\;,
    \end{align*}
    where the last equality follows from the fact that $G\indep \eps$, hence $H \indep \eps$. 

    By some calculation, we obtain
    \begin{equation*}
        \E (\eps_i \eps_j \eps_k \eps_l)
        = \begin{cases}
            1 &\text{ if }i = j \neq k = l\;, \\
            1 &\text{ if }i = k \neq j = l\;, \\
            1 &\text{ if }i = l \neq j = k\;, \\
            \E\eps_i^4 &\text{ if }i = j = k = l\;, \\
            0 &\text{ otherwise}\;. 
        \end{cases}
    \end{equation*}
    Hence, 
    \begin{align*}
        \E (\eps^\top H \eps)^2 &= \sum_{i,k = 1, i\neq k}^n \E H_{ii}H_{kk} + \sum_{i,j = 1, i\neq j}^n \E H_{ij}H_{ij} + \sum_{i,j = 1, i\neq j}^n \E H_{ij}H_{ji} + \sum_{i = 1}^n \E\eps_i^4 \E H_{ii}^2\\
        &\stackrel{\text{(i)}}{=} \sum_{i,k = 1, i\neq k}^n \E H_{ii}H_{kk} + 2\sum_{i,j = 1, i\neq j}^n \E H_{ij}H_{ji} + \sum_{i = 1}^n \E\eps_i^4\E H_{ii}^2\\
        &= \sum_{i,k = 1}^n \E H_{ii}H_{kk} + 2\sum_{i,j = 1}^n \E H_{ij}H_{ji} + \sum_{i = 1}^n (\E\eps_i^4-3) \E H_{ii}^2\\
        &\le \E [(\sum_{i = 1}^n H_{ii})^2] + 2\E \tr(HH^\top) + \max_i \E\eps_i^4 \sum_{i = 1}^n \E H_{ii}^2\\
        &\le \E [\tr(H)^2] + 2\E\tr(HH^\top) + \max_i \E\eps_i^4 \tr(HH^\top)\\
        &= \E [\tr(H)^2] + (2+\max_i \E\eps_i^4) \E\tr(HH^\top)\;. 
    \end{align*}
    In the derivation above, (i) follows by the fact that $H = H^\top$ with probability one. Next, we will upper bound each term above. 

    Recall that $H = B^\top GAGB$. As $A$ is a projection matrix, we have $A = A^2$, hence $\tr(B^\top GAGB) = \tr(AGBB^\top GA)$. Due to the condition $BB^\top \preceq c_0I_n$, we can apply Lemma \ref{lem:trace_psdorder} to obtain
    \begin{align*}
        0\le \tr(H) &\le c_0\tr(AGI_nGA) = c_0\tr(A) = c_0 k\;,\\
        \E [\tr(H)^2] &\le c_0^2 k^2\;. 
    \end{align*}

    Using the same argument, we can show that 
    \begin{equation*}
        0\le \tr(HH^\top) \le c_0^2 k\;, \quad \E\tr(HH^\top) \le c_0^2 k\;.
    \end{equation*}
    To sum up, we have 
    \begin{equation*}
        \E (\eps^\top H \eps)^2 \le c_0^2 k^2 + (2+\max_i \E\eps_i^4) c_0^2 k\;. 
    \end{equation*}
\end{proof}

\section{Testing the Partial Null Hypothesis under $p<n$}\label{appendix:partial_lowdim}
In our first experiment for the partial null hypothesis, we evaluate the power of our test and ANOVA under a series of alternative hypotheses with $p<n$, under the following parameters:
\begin{itemize}
    \item $(n, p - |S|) = (600, 100)$ with $S^\complement = \{1, \dots, 100\}$ and $S = \{101, \dots, 100+|S|\}$ 
    for $|S| = 10, 20, \dots, 90$. As such, larger $|S|$ implies larger $p$.
    \item $X_{S^\complement} = W\Sigma^{1/2}$, where $\Sigma=\left(2^{-|j-k|}\right)_{j, k \in[(p - |S|)]}$ and $W_{ij}\stackrel{iid}{\sim}\cN(0, 1); {\tt t}_1$. Construct $X_S\in\R^{n\times |S|}$ by $X_{S} = X_{S^\complement} M + e$ so that $X_S$ and $X_{S^\complement}$ are correlated. Here, $M_{ij} = Z_{ij}\mathbb{I}\{(j-1)k+1\le i\le jk\}$, where $Z_{ij}\stackrel{iid}{\sim}\cN(0, 1)$ and $k = \lfloor (p - |S|)/|S| \rfloor$.
    \item $\beta_{S^\complement, j} = \mathbb{I}\{j\le 10\}/\sqrt{10}$ and $\beta_{S, j} = \mathbb{I}\{j\le 10\}/5\sqrt{10}$.
    \item Generate $\eps_i\stackrel{iid}{\sim}\cN(0, 1), {\tt t}_1, {\tt t}_2$ for errors in $y$. Correspondingly, we generate i.i.d. $(e_{ij})_{i\in[n], j\in[|S|]}$ satisfying $e_{ij} \myeq{d} \eps_1$ and $e\indep \eps$ for the `error matrix' $e$ in $X_S$. 
\end{itemize}
We construct $A$ by specifying $\tilde{X}_S = X_S$ (without variable selection) and $R = 1000$.
The rejection rates for different covariates and errors are shown in Figure~\ref{fig:power2}, based on 10,000 simulations. Under the generated alternative hypothesis, our residual test maintains good power even when $|S|>1$. In addition, our test achieves higher power than ANOVA in most of the cases. 
In Figure~\ref{fig:power2}, we observe that the power of both tests decreases with increasing $|S|$, which is expected because given $\|\beta_S\| = 1$, larger $|S|$ implies a harder testing problem. 

\begin{figure}[t!]
    \centering
    \subfloat[Gaussian design, Gaussian noise]{{\includegraphics[width=.35\linewidth]{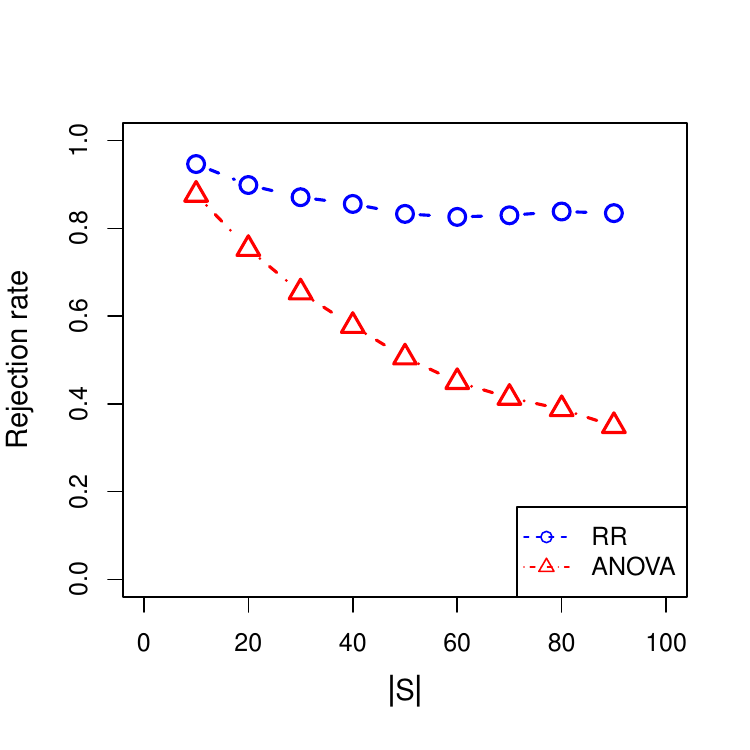}}}%
    \subfloat[Gaussian design, ${\tt t}_1$ noise]{{\includegraphics[width=.35\linewidth]{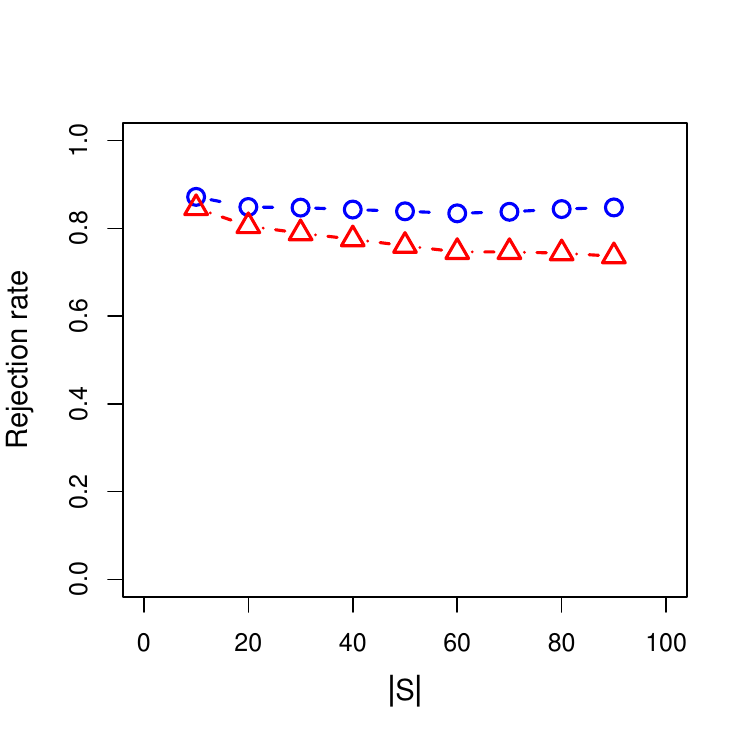}}}%
    \subfloat[Gaussian design, ${\tt t}_2$ noise]{{\includegraphics[width=.35\linewidth]{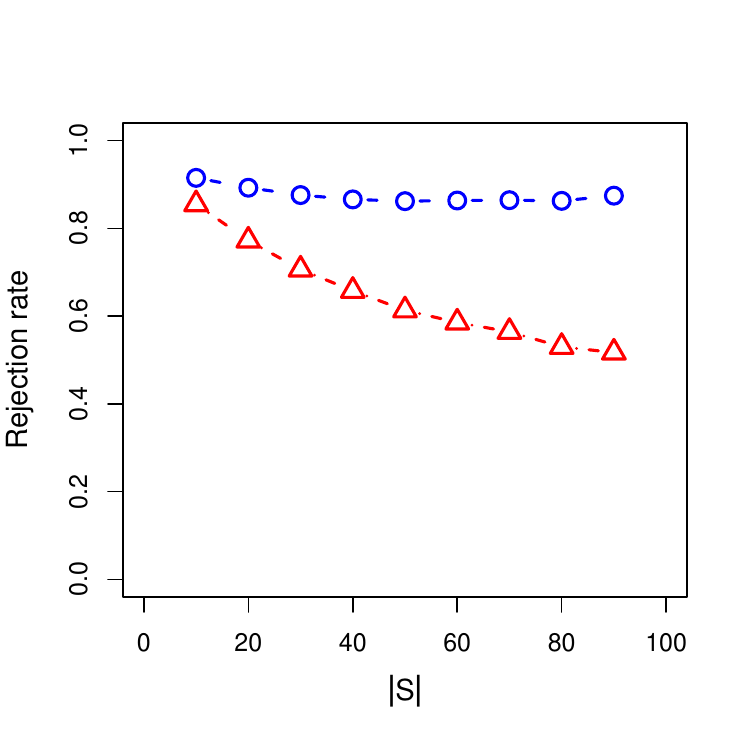}}}\\
    \subfloat[${\tt t}_1$ design, Gaussian noise]{{\includegraphics[width=.35\linewidth]{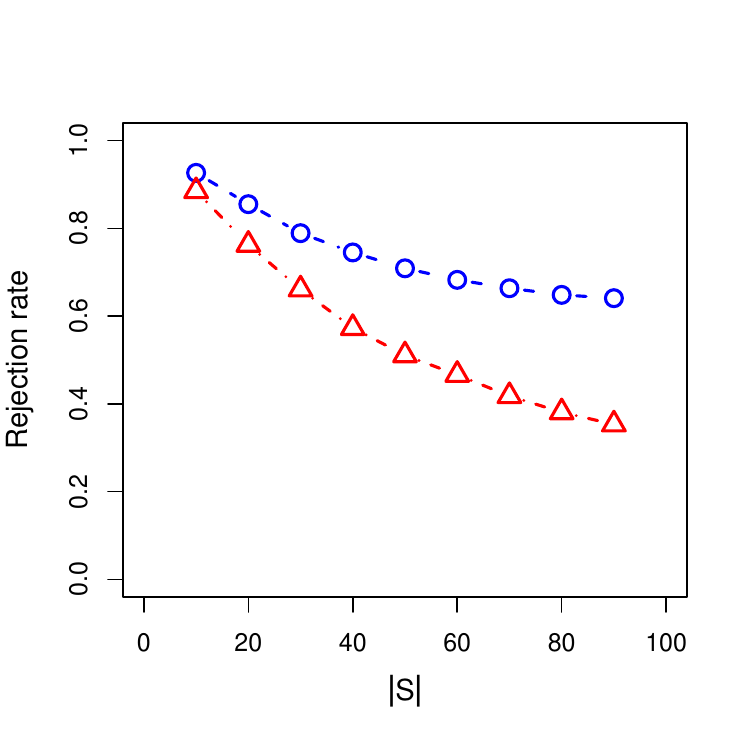}}}%
    \subfloat[${\tt t}_1$ design, ${\tt t}_1$ noise]{{\includegraphics[width=.35\linewidth]{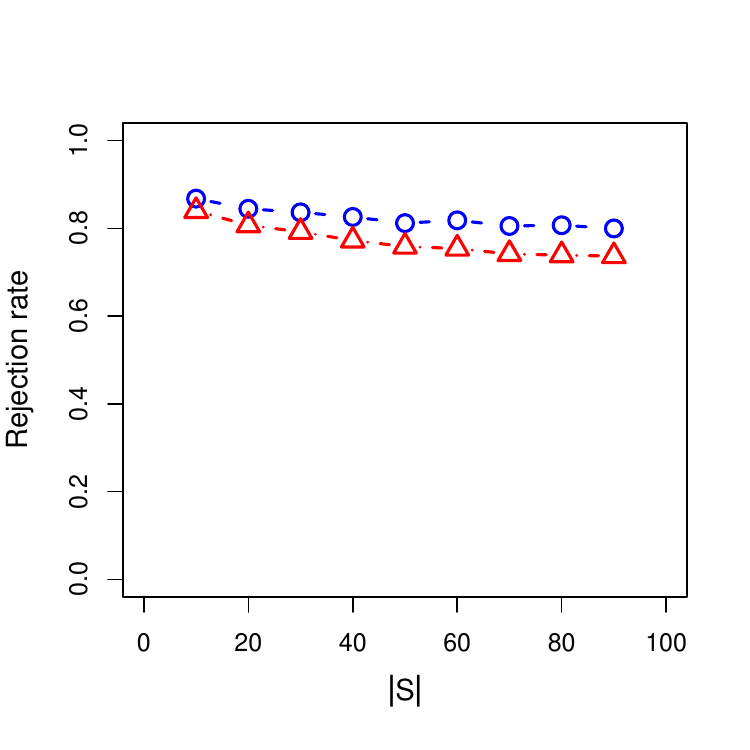}}}%
    \subfloat[${\tt t}_1$ design, ${\tt t}_2$ noise]{{\includegraphics[width=.35\linewidth]{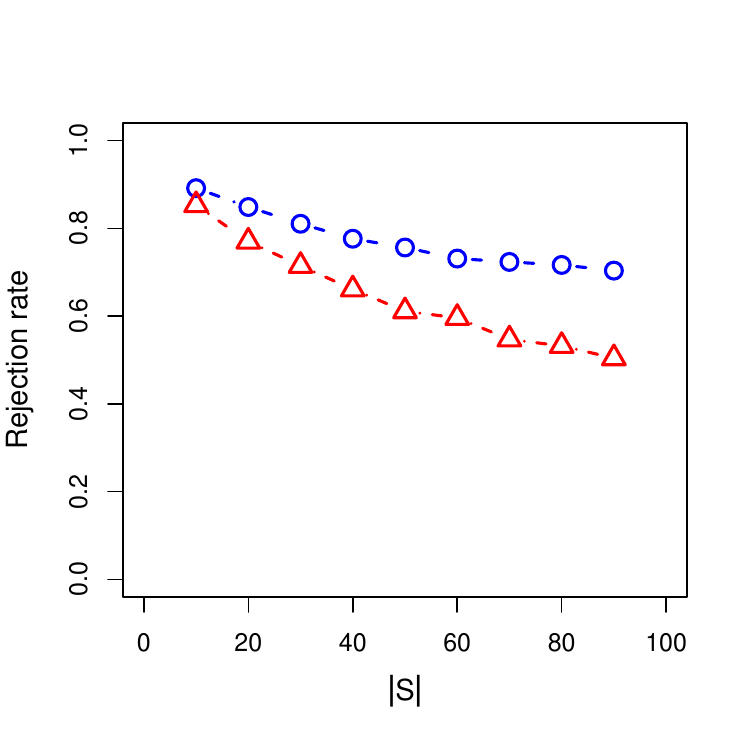}}}%
    \caption{Rejection rates under various setups testing partial null hypotheses with $p<n$. 
    ``RR" stands for the ``residual randomization" test in Procedure 2. 
    In the simulation, $p=|S| + 100$, and so increasing $|S|$ implies increasing $p$. }
    \label{fig:power2}
\end{figure}

\section{Gene Association Tests}\label{sec:gene}
In this section, we apply our residual randomization test to detect associations between gene sets and certain clinical outcomes in a real biological experiment, which was also analyzed by \cite{Zhong2011}. The experiment was conducted on 24 six-month-old Yolkshire gilts. First, the gilts were genotyped according to the melanocortin-4 receptor gene, 12 of them with D298 and the other with N298. Next, two diet treatments were randomly assigned to the 12 gilts in each genotype. One treatment is ad libitum (no restrictions) in the amount of feed consumed and the other is fasting. More details of the experiment can be found in \cite{Lkhagvadorj2009}. The genotypes and the diet treatments are two binary factors potentially associated with the clinical outcome. Our goal is
to identify associations between gene sets and triiodothyronine
(T3) measurement, a vital thyroid hormone.

For each gilt, we obtain gene expression values for 24,123 genes and the T3 measurement. Gene sets are defined by Gene Ontology (GO) terms \citep{GeneOntology}, which classifies 24,123 genes into 6176 different sets according to their biological functions among three broad categories: cellular component, molecular function, and biological process. Therefore, we aim to test possible associations between the T3 measurement and GO terms in the presence of certain nuisance parameters.

Let $k=1, \dots, 24$ index observations. We use $z_k\in\{0, 1\}$ and $w_k\in\{0, 1\}$ to indicate whether the $k$-th observation gets fasting treatment and has genotype D298, respectively. For instance, $(z_k, w_k) = (0, 1)$ indicate the $k$-th observation gets ad libitum treatment with genotype D298. Let $X^g_{k}\in\R^{p_g\times 1}$ be the gene expressions of $k$-th observation for the $g$-th GO term. We consider the following four models, which relate to four different sets of nuisance variables:
\begin{enumerate}[\text{Model} I:]
\item $y_k = \alpha + {X^g_k}^\top \beta^g + \eps^g_k$.
\item $y_{k} = \alpha + \mu z_k + {X^g_k}^\top \beta^g + \eps^g_{ik}$.
\item $y_{k} = \alpha + \tau w_k + {X^g_{k}}^\top \beta^g + \eps^g_{k}$.
\item $y_{k} = \alpha + \mu z_k + \tau w_k + \nu z_k w_k + {X^g_{k}}^\top \beta^g + \eps^g_{k}$.
\end{enumerate}
For each model above, we test multiple partial nulls 
\begin{equation*}
    H_{g, 0}: \beta^g = 0~, ~\text{for}~g = 1, \dots, 6176~.
\end{equation*}

Among the 6176 GO terms, the dimensions $p_g$ of gene sets
range from 1 to 5158, and many gene sets share common
genes. Therefore, the covariates are both high-dimensional and correlated across different GO terms, yielding a difficult multiple testing problem. We select significant GO terms for Models I--IV as below: 
\begin{enumerate}
    \item For $g = 1, \dots, 6176$, if $p_g\ge 5$, we apply the residual randomization test for $H_{g, 0}: \beta^g = 0$ in the corresponding model; otherwise we apply $F$-test. \footnote{Among 6176 GO terms, there are 2314 terms with $p_g\ge 5$.}
    \item In each residual randomization test, choose $A$ to be the projection matrix onto $\Tilde{X}^{g,\perp}$, where 
    \begin{equation*}
        \tilde{X}^{g,\perp} = {\tilde{X}^g} - P_Z {\tilde{X}^g}\;, \quad P_Z = Z({Z}^\top Z)^{-1}{Z}^\top~.
    \end{equation*}
    Here, we apply Lasso selection by regressing $y_k$ on $X_g$ and obtain $\tilde{X}^g$, which encode selected genes in the $g$-th GO term.\footnote{In Lasso selection, the tuning parameter $\lambda$ is obtained from cross-validation.} $Z$ is the covariate matrix of nuisance variables under the corresponding model. For instance, under Model IV, $Z$ is a $24\times 4$ matrix whose $k$-th row is $(1, z_k, w_k, z_k w_k)$.
    \item After collecting $p$-values for all GO terms, apply the Benjamini-Hochberg (BH) procedure to select significant GO terms by controlling the false discovery rate at level $0.05$. 
\end{enumerate}

Our testing procedure (denoted as ``RR'') identifies 119, 45, 102, 82 significant GO terms under Models I--IV, respectively. In Table~\ref{tab:sigGO}, we present 11 GO terms that are significant under all the models. Prior to our work, \cite{Zhong2011} has analyzed the same data set by the procedure described above, using their proposed test, referred to as ZC. For comparison purposes, we provide in Table~\ref{tab:sigGO} the GO terms detected by \cite{Zhong2011}. We observe that our test detects all three GO terms from ZC with eight new GO terms; it indicates that our test has potentially higher detection power. Among new GO terms, GO:0004066 and GO:0006529 are related to the formation of asparagine, and GO:0006865 is related to amino acid transport. As one main function of T3 is to increase protein synthesis, the new GO terms are biologically relevant. 
%

\begin{table}[htbp!]
    \centering
    \begin{tabular}{p{6em}p{6em}p{6em}}
    \hline
    GO term & No. of genes & Satisfied test(s) \\
    \hline
    GO:0004066  &  3  &  RR \\ 
    GO:0005086  &  14  &  RR/ZC \\ 
    GO:0006529  &  3  &  RR \\ 
    GO:0006865  &  24  &  RR \\ 
    GO:0007528  &  8  &  RR/ZC \\ 
    GO:0008747  &  2  &  RR \\ 
    GO:0009952  &  47  &  RR \\ 
    GO:0016575  &  18  &  RR \\ 
    GO:0030036  &  125  &  RR \\
    GO:0032012  &  12  &  RR/ZC \\ 
    GO:0033033  &  2  &  RR \\ 
    \hline
    \end{tabular}
    \bigskip
    \caption{The significant GO terms under all four models and the corresponding number of genes. We drop three GO terms detected by our test that contain only one gene. }
    \label{tab:sigGO}
\end{table}

\end{document}